\documentclass{acmart}

\usepackage{amsmath}
\usepackage{cleveref}
\usepackage{amsthm}
\usepackage{stmaryrd}
\usepackage{marginnote}
\usepackage{adjustbox}
\usepackage{xspace}
\usepackage{mathtools}

\usepackage{thm-restate}

\usepackage{pifont}
\usepackage{url}

\usepackage{enumitem}
\usepackage{multirow}

\usepackage{virginialake}\vlnostructuressyntax
\usepackage{proofzilla}

\newcommand{\apx}[1][]{\mathcal K_{#1}}
\def\dom{\mathfrak D}
\def\oSpdl{\mathsf{o}\Spdl}
\def\lengthof#1{\ell({#1})}
\def\mysup{\bigsqcup}
\def\strat{\mathfrak s}

\def\mces{strategy\xspace}

\def\cf#1{\mathsf{cf}\left(#1\right)}

\def\dK{\mathcal K}
\def\hbh{\mathfrak{h}}%{{\mathfrak{s}{\scalebox{.75}{$\uparrow$}}}}

\newcommand{\resp}[1]{(resp. #1)\xspace}

\mathchardef\mhyphen="2D

\def\N{\mathbb N}

\def\sizeof#1{\left|#1\right|}

\renewcommand{\emptyset}{\varnothing}
\def\set#1{\{#1\}}
\def\Set#1{\left\{#1\right\}}

\def\intset#1#2{\set{#1, \ldots, #2}}

\newcommand{\tuple}[1]{\ensuremath{\langle #1 \rangle}}

\def\IH{\mathbf{IH}} 
\def\CCS{\ensuremath{\mathsf{CCS}}\xspace}
\def\CCSs{\ensuremath{\mathsf{CCS^\ast}}}

\def\ppar{\,|\,}

\def\Act{\tcol{\mathit{Act}}}
\def\peq{\equiv}
\def\prer{\rname{Pre}}
\def\comr{\rname{Com}}
\def\parr{\rname{Par}}
\def\sumr{\rname{Sum}}
\def\resr{\rname{Res}}
\def\recr{\rname{Rec}}
\def\defeq{\stackrel{\mathsf{def}}{=}}

\def\cneg#1{\overline{#1}}

\def\FL#1{\mathsf{FL}(#1)}

\def\PDLf{\mathcal F_{\PDL}}
\def\PDLform{\PDL-formula\xspace}

\def\DL{\ensuremath{\mathsf{DL}}\xspace}
\def\PDL{\ensuremath{\mathsf{PDL}}\xspace}
\newcommand{\POL}[1][O]{\mathsf{#1}\PDL}
\def\OPDL{\ensuremath{\POL}\xspace}

\def\model{\mathfrak M}
\def\mean{\mathfrak m}
\def\meanof#1{\mean\left(#1\right)}

\def\typepol#1{{\color{blue}#1}}
\def\imp{\Rightarrow}%{\supset}
\def\fequiv{\Leftrightarrow}%{\supset\subset}%
\def\atomSet{\typepol{\mathcal A}}

\def\fP{\typepol{p}}
\def\nfP{\cneg \fP}

\def\fA{\typepol{\phi}}
\def\fB{\typepol{\psi}}
\def\fC{\typepol{\chi}}
\def\nfA{\typepol{\cneg\phi}}
\def\nfB{\typepol{\cneg\psi}}

\def\fGam{\typepol{\Gamma}}
\def\fDel{\typepol{\Delta}}
\def\fSig{\typepol{\Sigma}}

\def\ftrue{\typepol{\top}}%{\typepol{t\!t}}
\def\ffalse{\typepol{\bot}}%{\typepol{f\!\!f}}

\newcommand{\tbox}[1][~]{\left[#1\right]}
\newcommand{\tdia}[1][~]{\left\langle#1\right\rangle}

\def\kstarsymb{\ast}
\def\kstar#1{#1^\kstarsymb}

\def\tcol#1{{\color{pzgreen} #1}}
\def\labSet{\tcol{\mathbb L}}

\pzvertex{hid}{~}{}
\defedgetype{OS}{-stealth,draw,thick}{}
\newcommand{\sosto}[2][\qquad]{
	\vhid1 
	#1
	\vhid2
	\OSledges{hid1/hid2/#2}
}

\def\tm{\tcol \mu}
\def\tl{\tcol \lambda}

\def\ttest#1{\tcol{#1 ?}}
\def\tfA{\ttest\phi}
\def\tfB{\ttest\psi}

\def\tat{\tcol a}
\def\ntat{\tcol {\cneg a}}

\def\ta{\tcol \alpha}
\def\tb{\tcol \beta}
\def\tc{\tcol \gamma}
\def\tI{\tcol I}
\def\ttau{\tcol \tau}
\def\tepsi{\tcol \epsilon}
\def\tempt{\tcol \emptyset}

\def\pcol#1{{\color{pzred} #1}}
\def\progSet{\pcol{\mathbb P}}
\def\atomProg{\pcol{\mathbb I}}
\def\atomTests{\pcol{\mathbb T}}

\def\pl{\pcol \lambda}
\def\ppi{\pcol \pi}

\def\pnil{\pcol {\boldsymbol 0}}
\def\pt{\pcol t}

\def\pat{\pcol a}
\def\npat{\pcol {\cneg a}}

\def\pP{\pcol P}
\def\pQ{\pcol Q}
\def\pR{\pcol R}
\def\pa{\pcol \alpha}
\def\pb{\pcol \beta}
\def\pc{\pcol \gamma}
\def\pepsi{\pcol \epsilon}

\def\pres#1#2{\pcol {#2{\setminus}#1}}
\def\pX{\pcol X}

\newcommand{\OS}[1][]{\mathsf{OS}_{#1}}

\def\Spdl{\mathsf{LPD}}
\def\Spdlcut{\Spdl^{\cutr}}
\def\pSpdl{\mathsf{p}\Spdl}
\def\pSpdlcut{\mathsf{p}\Spdlcut}
\def\Sdol{\mathsf{LOPD}}

\def\pSdol{\mathsf{p}\Sdol}

\def\mprule{\rname{MP}}
\def\necrule{\rname{NEC}}
\def\looprule{\rname{LI}}

\def\axA#1{\mathbf{A_{#1}}}

\newcommand*{\rname}[1]{\ensuremath{\textsc{\MakeLowercase{#1}}}}

\def\AXrule{\rname{AX}}
\def\Wrule{\rname{W}}

\def\Krule{\rname{K}_\alpha}
\def\Krulep#1{\rname{K}_{#1}}
\def\Trule{\ftrue}
\def\cutr{\mathsf{cut}}
\def\rrule{\mathsf r}
\def\tree{\mathcal T}
\def\branch{\mathcal B}

\def\Sys{\mathsf{S}}
\def\X{\mathsf{X}}
\def\dD{\mathcal D}
\def\dK{\mathcal K}

\def\btestr{\tbox[?]}
\def\bseqr{\tbox[;]}
\def\bplusr{\tbox[\oplus]}
\def\bstarr{\tbox[\kstarsymb]}
\def\dtestr{\tdia[?]}
\def\dseqr{\tdia[;]}
\def\dplusr{\tdia[\oplus]}
\def\dstarr{\tdia[\kstarsymb]}
\def\bosr{\tbox[\OS]}
\def\dosr{\tdia[\OS]}

\newcommand*{\proves}[1][]{\vdash_{#1}}

\def\cutelim{\rightsquigarrow}

\def\chorpol#1{\pcol{#1}}

\def\iI{\chorpol I}
\def\cC{\chorpol C}

\def\cX{\chorpol X}

\newcommand{\pn}{\mathop{\mathsf{pn}}}
\def\cnil{\chorpol{\boldsymbol 0}}

\newcommand{\pid}[1]{\mathsf{#1}}

\newcommand{\tto}{\mathbin{\boldsymbol{\rightarrow}}}
\def\pp{\pid p}

\def\pq{\pid q}

\newcommand{\lbl}[1]{\textsc{#1}}
\newcommand{\albl}{\lbl{l}}
\newcommand{\condif}{\mathop{\mathsf{if}}}
\newcommand{\condthen}{\mathop{\mathsf{then}}}
\newcommand{\condelse}{\mathop{\mathsf{else}}}

\newcommand{\com}[4]{\chorpol{\pid{#1}.{#2} \tto\pid{#3}.{#4}}}
\newcommand{\sel}[3]{\chorpol{\pid{#1}\tto\pid{#2}[#3]}}
\newcommand{\cont}[2]{\chorpol{#2\colon #1}}
\newcommand{\assign}[3]{\chorpol{\pid{#1}.{#2}\mathbin{\coloneqq}{#3}}}
\newcommand{\ccond}[4]{\chorpol{\condif \pid{#1}.{#2} \condthen #3 \condelse #4}}

\newcommand{\tsel}[3]{\tcol{\pid{#1}\tto\pid{#2}[#3]}}
\newcommand{\tgensel}{\tsel pq\albl}
\newcommand{\tcont}[2]{\tcol{#2\colon #1}}
\newcommand{\tgencont}{\tcont X{\pid q}}

\newcommand{\tassign}[3]{\tcol{\pid{#1}.{#2}\mathbin{\coloneqq}{#3}}}
\newcommand{\tgenassign}{\tassign pxe}
\newcommand{\tcom}[4]{\tcol{\pid{#1}.{#2} \tto\pid{#3}.{#4}}}
\newcommand{\tgencom}{\tcom peqx}

\newcommand{\gensel}{\sel pq\albl}
\newcommand{\gencom}{\com peqx}

\newcommand{\gencall}{\chorpol{X}}
\newcommand{\genassign}{\assign pxe}
\newcommand{\gencond}{\ccond pb{C_1}{C_2}}
\newcommand{\gentest}{\ttest{\pp.b}}
\newcommand{\ngentest}{\ttest{\cneg{\pp.b}}}

\def\seqEval#1#2{e \downarrow_\ldia v}

\def\ratm{\rname{Atomic}}

\def\rthen{\rname{Cond-Then}}
\def\relse{\rname{Cond-Else}}
\def\rcall{\rname{Call}}
\def\ri{\rname{I}}
\def\rdeli{\rname{Delay-I}}
\def\rdelc{\rname{Delay-Cond}}

\def\qquand{\quad\mbox{and}\qquad}

\def\defn#1{\textbf{#1}}

\newtheorem{theorem}{Theorem}
\newtheorem{lemma}[theorem]{Lemma}
\newtheorem{proposition}[theorem]{Proposition}
\newtheorem{corollary}[theorem]{Corollary}

\theoremstyle{definition}
\newtheorem{definition}[theorem]{Definition}
\newtheorem{example}[theorem]{Example}
\newtheorem{nota}[theorem]{Notation}
\newtheorem{remark}[theorem]{Remark}

\newcommand{\mydot}{\, .}

\begin{document}
%%%%%%%%%%%%%%%%%%%%%%%%%%%%%%%%%%%%%%%%%%%%%%%%%%%%%%%%%%%%
%%%%%%%%%%%%%%%%%%%%%%%%%%%%%%%%%%%%%%%%%%%%%%%%%%%%%%%%%%%%
%%%%%%%%%%%%%%%%%%%%%%%%%%%%%%%%%%%%%%%%%%%%%%%%%%%%%%%%%%%%
%%%%%%%%%%%%%%%%%%%%%%%%%%%%%%%%%%%%%%%%%%%%%%%%%%%%%%%%%%%%
\title{On Propositional Dynamic Logic and Concurrency}


\author{Matteo Acclavio}
% \authornote{Supported by Villum Fonden, grant no. 50079}
%\email{mail}
%\affiliation{%
%	\institution{University of Southern Denmark}
%	\city{Odense}
%	\country{Denmark}
%}
\affiliation{%
	\institution{University of Southern Denmark}
	\city{Odense}
	\country{Denmark}
}

\author{Fabrizio Montesi}
%\email{mail}
%\orcid{orcid}
\affiliation{%
	\institution{University of Southern Denmark}
	\city{Odense}
	\country{Denmark}
}

\author{Marco Peressotti}
%\email{mail}
\affiliation{%
	\institution{University of Southern Denmark}
	\city{Odense}
	\country{Denmark}
}

\renewcommand{\shortauthors}{Acclavio et al.}

%%%%%%%%%%%%%%%%%%%%%%%%%%%%%%%%%%%%%%%%%%%%%%%%%%%%%%%%%%%%
%%%%%%%%%%%%%%%%%%%%%%%%%%%%%%%%%%%%%%%%%%%%%%%%%%%%%%%%%%%%
\begin{abstract}
	Dynamic logic is a powerful approach to reasoning about programs and their executions, obtained by extending classical logic with modalities that can express program executions as formulas.
	However, the use of dynamic logic in the setting of concurrency has proved problematic because of the challenge of capturing interleaving.
	This challenge stems from the fact that, traditionally, programs are represented by their sets of traces. These sets are then expressed as elements of a Kleene algebra, for which it is not possible to decide equality in the presence of the commutations required to model interleaving.
	
	In this work, we generalise propositional dynamic logic (PDL) to a logic framework we call operational propositional dynamic logic (OPDL), which departs from tradition by distinguishing programs from their traces. Traces are generated by an arbitrary operational semantics that we take as a parameter, making our approach applicable to different program syntaxes and semantics.
	% in which we are able to reason on sets of programs provided with arbitrary operational semantics. 
	To develop our framework, we provide the first proof of cut-elimination for a finitely-branching non-wellfounded sequent calculus for PDL.
	Thanks to this result we can effortlessly prove adequacy for PDL, and extend these results to OPDL. 
	We conclude by discussing OPDL for two representative cases of concurrency: the Calculus of Communicating Systems (CCS), where interleaving is obtained by parallel composition, and Choreographic Programming, where interleaving is obtained by out-of-order execution.
\end{abstract}
%%%%%%%%%%%%%%%%%%%%%%%%%%%%%%%%%%%%%%%%%%%%%%%%%%%%%%%%%%%%
%%%%%%%%%%%%%%%%%%%%%%%%%%%%%%%%%%%%%%%%%%%%%%%%%%%%%%%%%%%%




%\begin{teaserfigure*}
%\end{teaserfigure*}


%%%%%%%%%%%%%%%%%%%%%%%%%%%%%%%%%%%%%%%%%%%%%%%%%%%%%%%%%%%%
%%%%%%%%%%%%%%%%%%%%%%%%%%%%%%%%%%%%%%%%%%%%%%%%%%%%%%%%%%%%
%%%%%%%%%%%%%%%%%%%%%%%%%%%%%%%%%%%%%%%%%%%%%%%%%%%%%%%%%%%%
\maketitle
%%%%%%%%%%%%%%%%%%%%%%%%%%%%%%%%%%%%%%%%%%%%%%%%%%%%%%%%%%%%
%%%%%%%%%%%%%%%%%%%%%%%%%%%%%%%%%%%%%%%%%%%%%%%%%%%%%%%%%%%%
%%%%%%%%%%%%%%%%%%%%%%%%%%%%%%%%%%%%%%%%%%%%%%%%%%%%%%%%%%%%

\section{Introduction}
\label{sec:intro}

Logic, and in particular proof theory, offers several approaches to reason about different computational properties of programs.
In the Curry-Howard correspondence, programs are represented by proofs, thus providing a strong foundation for the development of type systems~\cite{W15,CP10,gir:proot}.
In logic programming, a program is an inference system, which allows for using proof search as the means of execution~\cite{lloyd2012foundations,MILLER1991125}.
In \emph{dynamic logic} (\DL), programs are part of the language of formulas itself, which enables the direct use of the logic to reason about the semantics of programs~\cite{DLbook}. Under the latter view, the purpose of programs is to change the truth value of a formula.
At the syntactic level, each program $\ta$
defines the modalities $\tbox[\ta]$ and $\tdia[\ta]$
and a formula $\tbox[\ta] \fA$ is interpreted as 
`every state reached after executing $\ta$ satisfies the formula $\fA$' while a formula $\tdia[\ta] \fA$ is interpreted as `there is a state reached after executing $\ta$ satisfying the formula $\fA$'.
This idea has been of profound inspiration in the field of formal verification~\cite{stirling1991local,CGKSVWW13}.
In this work, we are interested in the propositional fragment of dynamic logic (Propositional Dynamic Logic, or \PDL)~\cite{PDLcompleteness}.

%%%%%%%%%%%%%%%%%%%%%%%%%%%%%%%%%%%%%%%%%%%%%%%%%%%%%%%%%%%%
\paragraph{\PDL and the concurrency problem}
%%%%%%%%%%%%%%%%%%%%%%%%%%%%%%%%%%%%%%%%%%%%%%%%%%%%%%%%%%%%
While \PDL has been successfully applied to the study of sequential programs, extending this approach to concurrent programs remains challenging.
In standard \PDL, a program is represented by a regular expression that describes its set of possible traces.
In other words, programs are elements of a free Kleene algebra.
This works well for sequential programs, because one obtains that the theory of equational reasoning for Kleene algebras is a complete system for reasoning about \emph{trace equivalence}~\cite{hopcroft2001introduction,K97,KBRSWZ19,SKS23}.
Trace equivalence is therefore captured by logical equivalence in \PDL:
\begin{equation}
	\mbox{$\ta$ and $\tb$ have the same traces}
	\quad
	\mbox{iff}
	\quad
	\proves[\PDL] \tbox[\ta]\fA \Leftrightarrow \tbox[\tb]\fA
	\qquad
	\mbox{for all } \fA
	\mydot
\end{equation}

However, the case of concurrent programs with an interleaving semantics is more problematic. In the presence of interleaving, one expects traces differing by interleaving to be equivalent modulo equations of the form $\ta;\tb=\tb;\ta$ (called \emph{commutations}).
Unfortunately, the word problem in a Kleene algebra enriched with an equational theory containing such commutations is known to be undecidable%
\footnote{
	This is proven for star-continuous Kleene algebras in ~\cite{kozen:und} by reducing the Post correspondence problem to the word problem by combining sequential composition, iteration, and commutations.
	The result has recently been generalized to the general case in \cite{aze:zhang:gabo:kleene}.
}, 
which makes undecidable checking whether two modalities in \PDL are the same.
For the same reason, a general treatment of concurrency is still elusive also for simpler approaches than \PDL, like equational reasoning on programs represented as terms in Kleene algebras~\cite{HMSW11,HSMSZ16,BPS17,KB0WZ20}.\footnote{%
	While the pure algebraic setting based on Kleene algebras is strictly less expressive than \PDL, the missing expressivity (propositional reasoning on states reached after performing actions in a given trace) can be recovered by considering the (much) more complex structure of Kleene modules \cite{kleene_modules}.
}

As a consequence of this problem, applications of \PDL to concurrency fall short of the expected level of expressivity from established theories, like \CCS~\cite{M80} and the $\pi$-calculus~\cite{mil:par:wal:pi}.
For example, previous works lack nested parallel composition, synchronisation, or recursion~\cite{PDL:interleaving,PDL:pi,pel:CDL,pel:concurrentschemes,pel:concurrentCom,BENEVIDES201723}.
In general, adding any new concurrency feature (e.g., a construct in the language of programs or a law defining its semantics) requires great care and effort in establishing the meta-theoretical properties of the logic.
The result: a literature of various propositional dynamic logics, all independently useful, but with different limitations and dedicated technical developments.
\footnote{
	Another line of work, started by Peleg in \cite{peleg:CPDL}, defines \emph{concurrent $\PDL$} ($\mathsf{CPDL}$) by considering an additional constructor for programs (denoted $\cap$) that, allegedly, allows for modelling parallelism.
	However, this construct is clearly not sound with respect to the general semantics of parallel, since the axiom for the $\tdia[\cap]$ is defined as follows: $\tdia[{\ta\cap \tb}]\fA \fequiv \tdia[\ta] \fA \land \tdia[\tb]\fA$. 
	This does not model a parallel execution of $\ta$ and $\tb$, but rather another for of non-deterministic choice between the two, other than the $\oplus$ (which is denoted by $\cup$ in this work), requiring that the choice satifies some sort of coherence in its continuation.
	For an example of why such construct should not as an incarnation of the parallelism, consider a formula $\fA$ which is true whenever two memory cells $c_\alpha$ and $c_\beta$ both store the value $0$, and two independent programs $\ta$ (setting $c_\alpha$ to $0$) and $\tb$ (setting $c_\beta$ to $0$).
	In this case, if $\cap$ is interpreted as the parallel composition, we clearly have that $\tdia[\ta\cap\tb]\fA$ should be valid in any possible model. However, the axiom states that $\tdia[\ta]\fA \land \tdia[\tb]\fA$ should also be true, which is clearly not the case.
	% Another paper on $\mathsf{CPDL}$ \cite{goldblatt:CPDL} also follow this interpretation. This second paper has different approach with modalities, which become non inter-definable. This unusual behavior in modal logic is justified by appealing on a report (published only many years later \cite{wij:ner:CCPDL}) in which different authors consider a `constructive' version of $\PDL$, for which the absence of duality for modalities is well-motivated (see, e.g., \cite{simpson:phd,das:mar:Imodal}).
	% Note that both papers do not provide the details of the proof of soundness and completeness with respect to the semantics.
}

%%%%%%%%%%%%%%%%%%%%%%%%%%%%%%%%%%%%%%%%%%%%%%%%%%%%%%%%%%%%
\paragraph{Main contributions and structure of the paper}
%%%%%%%%%%%%%%%%%%%%%%%%%%%%%%%%%%%%%%%%%%%%%%%%%%%%%%%%%%%%

In this work, we significantly advance the line of work on \PDL by developing \emph{operational propositional dynamic logic} ($\POL$).
The key innovation of $\POL$ is to distinguish and separate reasoning on programs from reasoning on their traces.
Thanks to this distinction, we circumvent previous limitations and finally obtain a \PDL that can be applied to established concurrency models, such as \CCS~\cite{M80} and choreographic programming~\cite{M13:phd}.
Crucially, $\POL$ is a general framework: it is parameterised on the operational semantics used to generate traces from programs, yielding a simple yet reusable approach to characterise trace reasoning.

We proceed as described next.
After recalling the axiomatization and semantics of $\PDL$ in \Cref{sec:PDL}, 
in \Cref{sec:seqPDL} we provide a proof of its soundness and completeness with respect to the sequent calculus introduced in \cite{das:gir:TCL}.
For this purpose, we provide the first cut-elimination result for this non-wellfounded calculus, by adapting the technique developed in \cite{acc:cur:gue:CSL24}.%
\footnote{
	A cut-elimination result for another sequent calculus for $\PDL$ is provided in~\cite{hill:pog:PDL}, but that calculus is fundamentally different: it employs nested sequents and contains rules with an infinite number of premises.
}
This allows us to prove our results by reasoning on the axiomatisation and the sequent system, without directly relying on semantic arguments.

Then, in \Cref{sec:DOL}, we extend $\PDL$ with an additional axiom allowing us to 
encapsulate an operational semantics for a set of programs into the trace reasoning. We call the resulting logic \emph{operational propositional dynamic logic} (or $\POL$), 
providing a general framework encompassing various previous works~\cite{PDL:interleaving,PDL:pi,PDL:flow}.

We show that the expressive power of our framework goes beyond the state of the art in \Cref{sec:conc}, by instantiating it for two use cases of archetypes of concurrent programming languages: the Calculus of Communicating Systems (\CCS)~\cite{M80}, representative of the process algebra approach, and the textbook presentation of choreographic programming~\cite{montesi:book}, representative of languages inspired by the \emph{Alice-and-Bob} notation that originates from security protocols.
These two cases are interesting because they model concurrency in completely different ways: in
$\CCS$, concurrency is obtained through an explicit parallel operator equipped with an interleaving semantics, while in choreographic programming concurrency is obtained implicitly by executing instructions out of order whenever they involve different processes.
Thus, $\POL$ advances the study of $\PDL$ with leaps in both expressivity and versatility.

We conclude and discuss future work in \Cref{sec:perp}.

Related work is discussed, where relevant, as part of our development (in addition to the works mentioned in this introduction).

%%%%%%%%%%%%%%%%%%%%%%%%%%%%%%%%%%%%%%%%%%%%%%%%%%%%%%%%%%%%
%%%%%%%%%%%%%%%%%%%%%%%%%%%%%%%%%%%%%%%%%%%%%%%%%%%%%%%%%%%%
%%%%%%%%%%%%%%%%%%%%%%%%%%%%%%%%%%%%%%%%%%%%%%%%%%%%%%%%%%%%
\section{Preliminary notions on propositional dynamic logic}\label{sec:PDL}
%%%%%%%%%%%%%%%%%%%%%%%%%%%%%%%%%%%%%%%%%%%%%%%%%%%%%%%%%%%%
%%%%%%%%%%%%%%%%%%%%%%%%%%%%%%%%%%%%%%%%%%%%%%%%%%%%%%%%%%%%
%%%%%%%%%%%%%%%%%%%%%%%%%%%%%%%%%%%%%%%%%%%%%%%%%%%%%%%%%%%%

%%%%%%%%%%%%%%%%%%%%%%%%%%%%%%%%%%%%%%%%%%%%%%%%%%%%%%%%%%%%
% FIG PDL forms
%%%%%%%%%%%%%%%%%%%%%%%%%%%%%%%%%%%%%%%%%%%%%%%%%%%%%%%%%%%%
\begin{figure}
	$$
	\begin{array}{c|c}
		\begin{array}{rlll}
			\fA,\fB	\coloneqq
				&\ftrue				&& \mbox{true}							\\
			\mid&\ffalse			&& \mbox{false}							\\
			\mid&\fP\in\atomSet		&\mbox{(with $\fP\in\atomSet$)}			&
			\mbox{atom (literal)} 											\\
			\mid&\nfP		  		&\mbox{(with $\fP\in\atomSet$)}			& 
			\mbox{negated atom (literal)} 									\\
			\mid&\fA\lor\fB			&& \mbox{disjunction}					\\
			\mid&\fA\land\fB		&& \mbox{conjunction}					\\
			\mid&\tbox[\ta]\fA		&&
			\mbox{box}														\\
			\mid&\tdia[\ta] \fA		&&
			\mbox{diamond}
		\end{array}
	&
		\begin{array}{rlll}
			\ta, \tb\coloneqq
				&\tepsi				&& \mbox{terminated program}			\\
			\mid&\tempt				&& \mbox{stuck program}				\\
			\mid&\pat\in\atomProg	&& \mbox{instruction} 					\\
			\mid&\pt\in\atomTests	&& \mbox{test} 							\\
			\mid&\ta;\tb			&& \mbox{sequential composition}		\\
			\mid&\kstar{\ta}		&& \mbox{iteration}						\\
			\mid&\ta\oplus\tb		&& \mbox{(non-deterministic) choice}
		\end{array} 
	\end{array}
	$$
	\caption{Grammar generating formulas}
	\label{fig:PDLform}
\end{figure}
%%%%%%%%%%%%%%%%%%%%%%%%%%%%%%%%%%%%%%%%%%%%%%%%%%%%%%%%%%%%
%%%%%%%%%%%%%%%%%%%%%%%%%%%%%%%%%%%%%%%%%%%%%%%%%%%%%%%%%%%%

In this section we recall standard definitions and results for $\PDL$ as presented in \cite{DLbook}.
However, we present a slightly different syntax for formulas (limiting the tests to modalities-free formulas), which is more convenient for our purposes.

%%%%%%%%%%%%%%%%%%%%%%%%%%%%%%%%%%%%%%%%%%%%%%%%%%%%%%%%%%%%
%\begin{definition}	
%	In this section w
	We consider the set $\PDLf$ of \defn{formulas}
	generated from 
	a countable set $\atomSet$ of \defn{propositional atoms},
	a set of \defn{atomic programs} $\atomProg$ and 
	a set of \defn{propositional tests} $\atomTests=\set{\tfA \mid \fA \text{ modalities-free formula in } \PDLf}$
	by the grammars in \Cref{fig:PDLform}.

	\begin{remark}\label{rem:tests}
		The standard language for \PDL considers tests to be defined for any formula, not only for modalities-free ones. 
		The reason we consider this restriction is that the presence of modalities in tests would make the proof of cut-elimination for the sequent calculus in \Cref{sec:seqPDL} much more complex, and it is not needed for our purposes: in the operational semantics of the programs we are interested, tests can be limited to `static' ones, requiring to check a value without executing any program (represented by the modalities in the formula) prior to perform the test itself.
	\end{remark}

	The \defn{(logical) implication} $A\imp B\coloneqq(\cneg A) \lor B$
	is defined by extending the negation from atoms to formulas via the \defn{De Morgan laws}:
	\begin{equation}\label{eq:deMorgan}
		\def\myskip{\hskip12pt}
		\begin{array}{r@{\;=\;}l@{\myskip}r@{\;=\;}l@{\myskip}r@{\;=\;}l@{\myskip}r@{\;=\;}l@{\mydot}}
			\cneg{\ftrue} 			& \ffalse
			&
			\cneg{(\cneg{\fA})} 	& \fA
			&
			\cneg{{\fA \land \fB}}	& (\cneg\fA) \lor (\cneg\fB)
			&
			\cneg{\tbox[\ta] \fA} 	& \tdia[\ta] \cneg\fA
		\end{array}
	\end{equation}

	We write $\proves[\PDL] \fA$ (respectively $\proves[\PDL]\fA_1,\ldots \fA_n$)
	if 
	the formula $\fA$ (respectively $\fA_1\lor\cdots \lor\fA_n$) is derivable from the axioms in \Cref{fig:axPDL}
	using the rules 
	\defn{modus ponens} ($\mprule$), 
	\defn{necessitation} ($\necrule$), 
	\defn{loop invariance} ($\looprule$), 
	from the same figure.
	The \defn{propositional dynamic logic} (or \PDL) is defined as the logic
	of formulas satisfying $\proves[\PDL] \fA$.
%%%%%%%%%%%%%%%%%%%%%%%%%%%%%%%%%%%%%%%%%%%%%%%%%%%%%%%%%%%%

%%%%%%%%%%%%%%%%%%%%%%%%%%%%%%%%%%%%%%%%%%%%%%%%%%%%%%%%%%%%
\begin{remark}\label{rem:AXind}
	The axiomatization of $\PDL$ is often presented by 
	replacing the loop invariance rule
	with 
	an additional axiom (scheme)  
	$\axA{Ind} \;\colon\; (\fA\land\tbox[\kstar\ta](\fA\imp \tbox[\ta] \fA)) \imp \tbox[\kstar\ta]\fA$
	reminding the induction axiom (scheme) in Peano arithmetic.
	Proving that the two formulations are equivalent is an exercise which can be found in \cite{DLbook}.
\end{remark}
%%%%%%%%%%%%%%%%%%%%%%%%%%%%%%%%%%%%%%%%%%%%%%%%%%%%%%%%%%%%

%%%%%%%%%%%%%%%%%%%%%%%%%%%%%%%%%%%%%%%%%%%%%%%%%%%%%%%%%%%%
% FIG ax PDL
%%%%%%%%%%%%%%%%%%%%%%%%%%%%%%%%%%%%%%%%%%%%%%%%%%%%%%%%%%%%
\begin{figure}[t]
	\adjustbox{max width=\textwidth}{$\begin{array}{c|c}
		\begin{array}{l@{\;:\;}l}
			\mathbf{PL}	& \mbox{Axiomatization of propositional classical logic}
			\\
			\mathbf{Neg}	& \tbox[\ta]\fA \fequiv \cneg{\left(\tdia[\ta] \cneg\fA\right)}						
			\\
			\mathbf{K}	& \left(\tbox[\ta](\fA\imp \fB)\right)\imp \left(\tbox[\ta]\fA\imp \tbox[\ta]\fB\right)	
			\\
			\axA\emptyset	&\tbox[\tempt]\fA
			\\
			\axA\epsilon	&\tbox[\tepsi]\fA \fequiv\fA
			\\
			\axA{?}			&\tbox[\tfB]\fA\fequiv (\cneg{\fB}\lor \fA ) 
			\\
			\axA{\oplus}	&\tbox[\ta\oplus \tb]\fA\fequiv (\tbox[\ta]\fA \land \tbox[\tb]\fA) 
			\\
			\axA{;}			&\tbox[\ta;\tb]\fA\fequiv \tbox[\ta]\tbox[\tb]\fA 
			\\
			\axA{\kstarsymb}&\tbox[\kstar\ta]\fA\fequiv (\fA\land\tbox[\ta][\kstar\ta]\fA) 
		\end{array}
	&
		\begin{array}{c}
			\vliinf{\mprule}{}{\vdash \fB}{\vdash \fA}{\vdash \fA\imp\fB}  
		\\\\
			\vlinf{\necrule}{}{\vdash \tbox[\ta] \fA}{\vdash \fA} 
		\\\\
			\vlinf{\looprule}{}{\vdash \fA \imp \tbox[\kstar\ta] \fA}{\vdash \fA\imp \tbox[\ta] \fA} 
		\end{array}
	\end{array}
	$}
	\caption{Axioms and rules for Propositional Dynamic Logic.}
	\label{fig:axPDL}
\end{figure}
%%%%%%%%%%%%%%%%%%%%%%%%%%%%%%%%%%%%%%%%%%%%%%%%%%%%%%%%%%%%
%%%%%%%%%%%%%%%%%%%%%%%%%%%%%%%%%%%%%%%%%%%%%%%%%%%%%%%%%%%%

Semantically speaking, while a \emph{model} of propositional classical logic is simply an evaluation function assigning a truth value to each formula, models for \PDL are given by \emph{Kripke frames}.
A Kripke frame for classical modal logic is given by a set of \emph{worlds}, 
an \emph{accessibility relation} between worlds, and an \emph{evaluation function} assigning to each formula the set of worlds in which it is true.
Intuitively, a classical model can be seen as a single-world Kripke frame in which the evaluation function assigns to each formula a set containing the unique world of the frame only if it the formula is evaluated as true.
We recall here the formal definition of model for \PDL.

%%%%%%%%%%%%%%%%%%%%%%%%%%%%%%%%%%%%%%%%%%%%%%%%%%%%%%%%%%%%
\begin{definition}\label{def:krip}
	A \defn{Kripke frame} (or \defn{model})  
	$\model=\tuple{W,\mean}$ is given by
	a set of \defn{worlds} $W$, 
	a \defn{meaning function} associating 
	to each atom $\fP\in \atomSet$ a set of worlds $\meanof\fP\subseteq W$ (in which $\fP$ holds), and
	to each instruction $\pat\in\atomProg$ an accessibility relation $\meanof\pat\subseteq W\times W$.
	The meaning of compound formulas and programs is defined as shown in the left of \Cref{fig:modPDL}.
	The meaning of a sequent $\fGam=A_1,\ldots, A_n$ is defined as the meaning of the conjunction of the formulas in $\fGam$, that is, $\meanof\fGam= \meanof{\fA_1 \lor \cdots \lor \fA_n}$.

	We write
	$\model, w\vDash \fGam$ if  $w\in\meanof\fGam$ (or simply $w\vDash\fGam$ if $\model$ is clear from the context),
	and
	$\model \vDash \fGam$ if $\model,w \vDash \fGam$ holds for any world $w$ of $\model$.
	Finally we write $\vDash \fGam$ if $\model \vDash \fGam$ holds for any possible model $\model$.
\end{definition}
%%%%%%%%%%%%%%%%%%%%%%%%%%%%%%%%%%%%%%%%%%%%%%%%%%%%%%%%%%%%

%%%%%%%%%%%%%%%%%%%%%%%%%%%%%%%%%%%%%%%%%%%%%%%%%%%%%%%%%%%%
% FIG  Kripke frames
%%%%%%%%%%%%%%%%%%%%%%%%%%%%%%%%%%%%%%%%%%%%%%%%%%%%%%%%%%%%
\begin{figure}[t]
	\adjustbox{max width=\textwidth}{$\begin{array}{l@{\;=\;}l|l@{\;=\;}l}
			\meanof{\ftrue}			& W									
			&
			\meanof{\tepsi}			& \set{(v,v)\mid v\in W}			
			\\
			\meanof{\ffalse}		& \emptyset							
			&
			\meanof{\tempt}			& \emptyset							
			\\
			\meanof{\cneg \fA}  	& W\setminus \meanof\fA				
			&
			\meanof{\tfA}			& \set{(v,v)\mid v\in \meanof\fA}	
			\\
			\meanof{\fA\lor\fB}		& \meanof\fA\cup\meanof\fB			
			&
			\meanof{\ta;\tb}		& \set{(u,w)\!\mid\! \mbox{exists } v \mbox{ s.t. } (u,v)\!\in\!\meanof\ta \mbox{ and } (v,w)\!\in\!\meanof\tb} 
			\\
			\meanof{\fA\land\fB}	& \meanof\fA\cap\meanof\fB			
			&
			\meanof{\ta\oplus\tb}	& \meanof\ta\cup\meanof\tb
			\\
			\meanof{\tbox[\ta]\fA}	& \set{v\mid w\in\meanof\fA  \mbox{ for all }w \mbox{ s.t. }(v,w)\in\meanof\ta}										
			&
			\meanof{\kstar\ta}		& \bigcup_{n\geq0} \meanof{\ta^n} 	\hfill\mbox{(where $\ta^0=\tepsi$)}\hfill
			\\
			\meanof{\tdia[\ta] \fA}	& \set{v\mid w\in\meanof\fA  \mbox{ for a }w \mbox{ s.t. }(v,w)\in\meanof\ta}
		\end{array}$}
	\caption{
		Inductive definition of the meaning of compound formulas and programs in the Kripke frames.
	}
	\label{fig:modPDL}
\end{figure}
%%%%%%%%%%%%%%%%%%%%%%%%%%%%%%%%%%%%%%%%%%%%%%%%%%%%%%%%%%%%
%%%%%%%%%%%%%%%%%%%%%%%%%%%%%%%%%%%%%%%%%%%%%%%%%%%%%%%%%%%%

%%%%%%%%%%%%%%%%%%%%%%%%%%%%%%%%%%%%%%%%%%%%%%%%%%%%%%%%%%%%
The proof of soundness and completeness result of the axioms of \PDL with respect to the semantics can be found in \cite{DLbook}.
In particular, completeness is shown by constructing a model $\model_\fA$ for each consistent formula $\fA$ such that the formula $\fA$ holds in at least a world (i.e. $\meanof{\fA}\neq\emptyset$ in $\model_\fA$).
\begin{theorem}\label{thm:PDLax:sc}
	Let $\fA$ be a \PDLform. 
	Then 
	$\proves[\PDL] \fA$  iff $\vDash_\PDL \fA$.
\end{theorem}
%%%%%%%%%%%%%%%%%%%%%%%%%%%%%%%%%%%%%%%%%%%%%%%%%%%%%%%%%%%%

We conclude by showing the following result, which allows us to interpreted the modality $\tbox[\kstar\ta]$ as a fixpoint for the modality $\tbox[\ta]$.
%%%%%%%%%%%%%%%%%%%%%%%%%%%%%%%%%%%%%%%%%%%%%%%%%%%%%%%%%%%%
\begin{lemma}\label{lem:PDL:starImpN}
	If 
	$\not\proves[\PDL]\fA \imp \tbox[\kstar\ta]\fB$, 
	then there is $n\in\N$ such that $\not\proves[\PDL]\fA \imp \tbox[\ta^n]\fB$.
\end{lemma}
\begin{proof}
	By \Cref{thm:PDLax:sc},
	if $\not\proves[\PDL]\fA \imp \tbox[\kstar\ta]\fB$,
	then $\not\vDash\fA \imp \tbox[\kstar\ta]\fB$.
	By definition, this means that there is a model $\model=\tuple{W,\mean}$ such that 
	$\meanof{\fA \imp \tbox[\kstar\ta]\fB}\neq W$.
	Now assume that $\proves[\PDL]\fA \imp \tbox[\ta^k]\fB$ for any $k\in \N$.
	By \Cref{thm:PDLax:sc} this implies
	$\vDash\fA \imp \tbox[\ta^k]\fB$ for any $k\in\N$;
	therefore that in any model $\model=\tuple{W,\mean}$ 
	we have $W=\meanof{\fA \imp \tbox[\ta^k]\fB}$ for all $k\in\N$.
	This would be absurd since 
	\begin{align*}
	W
	&=
	\bigcap_{k\geq 0}W
	=
	\bigcap_{k\geq 0}\left(\meanof{\fA \imp \tbox[\ta^k]\fB}\right)
	=
	\bigcap_{k\geq 0}\left(\meanof{\cneg\fA}\cup \meanof{\tbox[\ta^k]\fB}\right)
	=
	\\
	&=
	\meanof{\cneg\fA}\cup \left(\bigcap_{k\geq 0}\meanof{\tbox[\ta^k]\fB}\right)
	=
	\meanof{\cneg\fA}\cup \meanof{\tbox[\kstar\ta]\fB}
	=
	\meanof{\fA \imp \tbox[\kstar\ta]\fB}
	\neq W
	\qedhere
	\end{align*}
\end{proof}
%%%%%%%%%%%%%%%%%%%%%%%%%%%%%%%%%%%%%%%%%%%%%%%%%%%%%%%%%%%%

%%%%%%%%%%%%%%%%%%%%%%%%%%%%%%%%%%%%%%%%%%%%%%%%%%%%%%%%%%%%
%%%%%%%%%%%%%%%%%%%%%%%%%%%%%%%%%%%%%%%%%%%%%%%%%%%%%%%%%%%%
%%%%%%%%%%%%%%%%%%%%%%%%%%%%%%%%%%%%%%%%%%%%%%%%%%%%%%%%%%%%
\section{Sequent calculus for \PDL}\label{sec:seqPDL}
%%%%%%%%%%%%%%%%%%%%%%%%%%%%%%%%%%%%%%%%%%%%%%%%%%%%%%%%%%%%
%%%%%%%%%%%%%%%%%%%%%%%%%%%%%%%%%%%%%%%%%%%%%%%%%%%%%%%%%%%%
%%%%%%%%%%%%%%%%%%%%%%%%%%%%%%%%%%%%%%%%%%%%%%%%%%%%%%%%%%%%

In this section we recall the definition for (possibly infinite) derivations in a sequent system.
We then consider the sequent system $\Spdl$ given by the rules in \Cref{fig:seqPDL}
introduced in \cite{das:gir:TCL,das:gir:TCLext,das:gir:journal}
(for the fragment of $\PDL$ without the programs $\tepsi$ and $\tempt$), 
as an adaptation of the sequent calculus for the modal $\mu$-calculus given in \cite{sequent:modalMu}.

To prove soundness and completeness of the sequent system $\Spdl$ with respect to the axiomatization of $\PDL$, we rely on the subformula property following from the admissibility of the $\cutr$ rule.
To prove admissibility of $\cutr$, we provide the first cut-elimination result for $\Spdl$ by adapting the technique developed in \cite{acc:cur:gue:CSL24,acc:cur:gue:CSL24ext}.
Note that the system $\Spdl$ from \cite{das:gir:TCL} 
(as well as the system studied in \cite{sequent:modalMu}, and the labeled cyclic proof system from \cite{doc:row:PDL})
does not contain the rule $\cutr$, and its soundness and completeness is not proven with respect to axiomatization but with respect to the semantics. More precisely, these adequacy results are proven by translating the winning conditions of the \emph{provability games} for $\PDL$ and for modal $\mu$-calculus (respectively defined in \cite{games:PDL} and \cite{games:modal-mu}) into correctness criteria for non-wellfounded derivations%
\footnote{
	As explicitly shown in \cite{acc:cat:games} in the case of intuitionistic logic,
	translations of winning games correspond to derivations in \emph{focused} sequent systems, that is, a systems in which the order of rules in proof search is subject to specific restriction.
}%
.

The main novelty in our proof of cut-elimination respect to the technique developed in \cite{acc:cur:gue:CSL24} is way we deal with the non-confluence of cut-elimination: since $\Spdl$ is a conservative extension of classical logic sequent calculus, for which cut-elimination is not confluent, our definitions of reduction strategy require to be adapted to deal with this characteristic -- see \Cref{rem:notCurryHoward}.

%%%%%%%%%%%%%%%%%%%%%%%%%%%%%%%%%%%%%%%%%%%%%%%%%%%%%%%%%%%%
% FIG sequebt calculus PDL
%%%%%%%%%%%%%%%%%%%%%%%%%%%%%%%%%%%%%%%%%%%%%%%%%%%%%%%%%%%%
\begin{figure}[t]
	\def\myskip{\hskip2em}
	\adjustbox{max width=\textwidth}{$\begin{array}{c}
		\begin{array}{c}		
			\vlinf{\Trule}{}{\vdash\ftrue}{}
			\myskip
			\vlinf{\AXrule}{}{\vdash\fA, \nfA}{}
			\myskip
			\vlinf{\Wrule}{}{\vdash\fGam, \fA}{\vdash\fGam}
			\myskip
			\vlinf{\lor}{}{\vdash\fGam, \fA\lor \fB}{\vdash\fGam, \fA,\fB}
			\myskip
			\vliinf{\land}{}{\vdash\fGam, \fA\land \fB}{\vdash\fGam, \fA}{\vdash\fGam,\fB}
		\\\\
			\begin{array}{c|c}
				\vlinf{\Krule}{
					\text{\scriptsize{$\ta\neq \tempt$}}
				}{\vdash\tdia[\ta]\fGam, \tbox[\ta]\fA}{\vdash\fGam,\fA}
				\myskip&\myskip
				\vliinf{\cutr}{}{\vdash \fGam}{\vdash \fGam, \fA}{\vdash \fGam,\nfA} 
			\end{array}
		\end{array}
		\\\\\hline\\
		\begin{array}{c@{\myskip}c@{\myskip}c}
			\vlinf{\tbox[\tepsi]}{}{\vdash\fGam,\tbox[\tepsi]\fA}{\vdash\fGam,\fA}
			&
			\vlinf{\tbox[\tempt]}{}{\vdash\tbox[\tempt]\fA}{}
			&
			\vlinf{\btestr}{}{\vdash\fGam,\tbox[\tfA]\fB}{\vdash\fGam,\cneg\fA\lor \fB}
			\\\\
			\vliinf{\bplusr}{}{
				\vdash\fGam,\tbox[\ta\oplus\tb]\fA
			}{
				\vdash\fGam,\tbox[\ta]\fA
			}{
				\vdash\fGam,\tbox[\tb]\fA
			}
			&
			\vlinf{\bseqr}{}{\vdash\fGam,\tbox[\ta;\tb]\fA}{\vdash\fGam,\tbox[\ta]\tbox[\tb]\fA}
			&
			\vliinf{\bstarr}{}{
				\vdash\fGam ,\tbox[\kstar\ta]\fA
			}{
				\vdash\fGam,\fA
			}{
				\vdash\fGam,\tbox[\ta;\kstar\ta]\fA
			}
			\\\\
			\vlinf{\tdia[\tepsi]}{}{\vdash\fGam,\tdia[\tepsi]\fA}{\vdash\fGam,\fA}
			&
			\vlinf{\tdia[\tempt]}{}{\vdash\fGam,\fB,\tdia[\tempt]\fA}{\vdash\fGam,\fB}
			&
			\vlinf{\dtestr}{}{\vdash\fGam,\tdia[\tfA]\fB}{\vdash\fGam,\fA\land\fB}
			\\\\
			\vlinf{\dplusr}{}{
				\vdash\fGam,\tdia[\ta\oplus\tb]\fA
			}{
				\vdash\fGam,\tdia[\ta]\fA, \tdia[\tb]\fA
			}
			&
			\vlinf{\dseqr}{}{\vdash\fGam,\tdia[\ta;\tb]\fA}{\vdash\fGam,\tdia[\ta]\tdia[\tb]\fA}
			&
			\vlinf{\dstarr}{}{
				\vdash\fGam ,\tdia[\kstar\ta]\fA
			}{
				\vdash\fGam,\fA,\tdia[\ta;\kstar\ta]\fA
			}
		\end{array}
	\end{array}$}
	\caption{Sequent calculus rules of the sequent system $\Spdlcut=\Spdl\cup\set\cutr$.}
	\label{fig:seqPDL}
\end{figure}
%%%%%%%%%%%%%%%%%%%%%%%%%%%%%%%%%%%%%%%%%%%%%%%%%%%%%%%%%%%%
%%%%%%%%%%%%%%%%%%%%%%%%%%%%%%%%%%%%%%%%%%%%%%%%%%%%%%%%%%%%

%%%%%%%%%%%%%%%%%%%%%%%%%%%%%%%%%%%%%%%%%%%%%%%%%%%%%%%%%%%%
%%%%%%%%%%%%%%%%%%%%%%%%%%%%%%%%%%%%%%%%%%%%%%%%%%%%%%%%%%%%
\subsection{Definitions and Notations for Derivations}
%%%%%%%%%%%%%%%%%%%%%%%%%%%%%%%%%%%%%%%%%%%%%%%%%%%%%%%%%%%%
%%%%%%%%%%%%%%%%%%%%%%%%%%%%%%%%%%%%%%%%%%%%%%%%%%%%%%%%%%%%

We assume the reader to be familiar with the terminology of sequent calculus (see, e.g., \cite{troelstra_schwichtenberg_2000}) and non-wellfounded sequent calculi (see, e.g., \cite{bae:dou:sau:16,acc:cur:gue:CSL24}).
We recall here the formalism we adopt in this paper.

A \defn{sequent} is a set of formulas $\Gamma = A_1,\ldots, A_n$.
A \defn{sequent system} is given by a set of \defn{rules} of the form
$\vlinf{\rrule}{}{\vdash \Gamma}{}$
or
$\vlinf{\rrule}{}{\vdash \Gamma}{\vdash \Gamma_1}$
or
$\vliinf{\rrule}{}{\vdash \Gamma}{\vdash \Gamma_1}{\vdash \Gamma_2}$,
where the sequents $\Gamma_1$ and $\Gamma_2$ are called \defn{premises} and 
the sequent $\Gamma$ is called \defn{conclusion} of the rule $\rrule$.
A formula is \defn{active} (resp. \defn{principal}) for a rule if 
it occurs in a premise but not in the conclusion
(resp. it occurs in the conclusion but in none of its premises).
A rule $\rrule$ is \defn{admissible} in $\Sys$ if its conclusion is derivable in $\Sys$ whenever its premises are.

%%%%%%%%%%%%%%%%%%%%%%%%%%%%%%%%%%%%%%%%%%%%%%%%%%%%%%%%%%%%
\begin{definition}\label{def:derivation}
	A \defn{tree} $\tree$ is a prefix-closed set of words over the alphabet $\intset1n\subset \N$ such that
	if $vk\in \tree$, then $vm\in \tree$ for all $m<k$.
	The elements of $\tree$ are called \defn{nodes}, the empty word $\epsilon\in\tree$ is called \defn{root}.
	A node $v\in \tree$ is \defn{below} $w\in \tree$ if $w=vv'$ with $v'\neq\epsilon$.
	The \defn{height} of a node is the number of nodes below it.
	A \defn{child} of $v\in\tree$ is a node of the form $vk\in \tree$ with $k\in\intset1n$.
	A \defn{branch} is a prefix-closed totally ordered (w.r.t to the prefix order) set of nodes.

	A \defn{derivation} (resp. \defn{open derivation}) is a labeling $\dD$ of a tree $\tree_\dD$ with nodes labeled by sequents in such a way for each node $v$  (resp. for each non-leaf node $v$) the sequent $\dD(v)$ is the conclusion sequent of a rule with premises the sequents $\dD(v_1),\ldots, \dD(v_n)$ where $v_1,\ldots, v_n$ are the children of $v$.
	The sequent $\dD(\epsilon)$ is called the \defn{conclusion} of $\dD$ and a leaf $v$ such that $\dD(v)$ is not the conclusion of a rule is called an \defn{open premise}.
	We identify (an occurrence of) a \defn{rule} $\rrule$ in $\dD$ with the nodes corresponding to its conclusion and premises. A node is \defn{below} a rule if it is its conclusion or any node below,
	and we may refer to a node of $\dD$ as a node in the underlying tree $\tree_\dD$.
\end{definition}
%%%%%%%%%%%%%%%%%%%%%%%%%%%%%%%%%%%%%%%%%%%%%%%%%%%%%%%%%%%%

\begin{nota}
	We may denote a derivation with conclusion $\Gamma$ 
	(resp. an open derivation with open premise $\Delta$ and conclusion $\Gamma$) 
	by $\vlderivation{\vlpr{\dD}{}{\vdash\Gamma}}$  
	$\left(\mbox{resp. }\vlderivation{\vlde{\dD}{}{\vdash\Gamma}{\vlhy{\vdash\Delta}}}\right)$ .
	
	A regular derivation can be represented as a finite (directed) graph of sequents, 
	by identifying nodes of its tree which are conclusions of two identical sub-derivations. 
	In this case we label the bottom-most rules of identical derivations by the same symbol (see the derivation on the left of \Cref{ex:unsoundDer} for an example).
\end{nota}

%%%%%%%%%%%%%%%%%%%%%%%%%%%%%%%%%%%%%%%%%%%%%%%%%%%%%%%%%%%%
\begin{definition}
	A \defn{sub-derivation} of $\dD$ is a derivation $\dD'$ such that $\dD'(w)=\dD(vw)$ for a $v\in\tree_\dD$.
	A derivation is \defn{regular} if it has finitely many distinct sub-derivations.
	An open derivation $\dD'$ is an \defn{approximation}\footnote{
		The name is meant to suggest that infinite derivations can be seen as the limit of their approximations. See \Cref{def:Scott}.
	}
	of $\dD$ 
	(denoted $\dD'\preceq \dD$) 
	if $\tree_{\dD'}\subseteq\tree_{\dD}$ and $\dD'(v)=\dD(v)$ for any $v\in\tree_{\dD'}$.
	If $\dD$ is a derivation, then we denote by $\apx[\dD]$ the set of its approximations.
	
	If $\X$ is a set of derivations, then we say that a sequent $\fGam$ is \defn{provable} in $\X$ (denoted $\proves[\X] \fGam$)
	if there is a derivation of $\fGam$ in $\X$.
	For this purpose, we may identify a sequent system $\Sys$ with the set of derivations over $\Sys$.
\end{definition}
%%%%%%%%%%%%%%%%%%%%%%%%%%%%%%%%%%%%%%%%%%%%%%%%%%%%%%%%%%%%

%%%%%%%%%%%%%%%%%%%%%%%%%%%%%%%%%%%%%%%%%%%%%%%%%%%%%%%%%%%%
%%%%%%%%%%%%%%%%%%%%%%%%%%%%%%%%%%%%%%%%%%%%%%%%%%%%%%%%%%%%
\subsection{A Sequent System for $\PDL$}
%%%%%%%%%%%%%%%%%%%%%%%%%%%%%%%%%%%%%%%%%%%%%%%%%%%%%%%%%%%%
%%%%%%%%%%%%%%%%%%%%%%%%%%%%%%%%%%%%%%%%%%%%%%%%%%%%%%%%%%%%

The sequent systems $\Spdlcut$ is defined by the set of rules in \Cref{fig:seqPDL}, while $\Spdl$ is the sub-system without the rule $\cutr$.

As standard, by dropping the condition that derivations are finite trees, we can easily define unsound (infinite) derivations, as the one below.
\begin{equation}\label{ex:unsoundDer}
	\vlderivation{
		\vliin{\cutr}{\star}{\vdash\fA}{
			\vlin{\AXrule}{}{\vdash \fA ,\cneg\fA}{\vlhy{}}
		}{
			\vlin{\cutr}{\star}{\vdash\fA}{\vlhy{}}
		}
	}
	\quad\coloneqq\quad
	\vlderivation{
		\vliin{\cutr}{}{\vdash\fA}{
			\vlin{\AXrule}{}{\vdash \fA ,\cneg\fA}{\vlhy{}}
		}{
			\vliin{\cutr}{}{\vdash\fA}{
				\vlin{\AXrule}{}{\vdash \fA ,\cneg\fA}{\vlhy{}}
			}{
				\vlin{\cutr}{}{\vdash\fA}{
					\vlhy{\vdots}
				}
			}
		}
	}
% \hskip6em
% 	\vlderivation{
% 		\vlin{\Crule}{}{\vdash\fA}{
% 			\vlin{\Crule}{}{\vdash\fA,\fA}{\vlhy{\vdots}}
% 		}
% 	}
\end{equation}
In the leftmost derivation, we use the standard syntax for representing infinite proofs. In this syntax, repeated applications of the same rule may be annotated with a symbol (in this case, $\star$) to indicate that they are roots of isomorphic sub-derivations.
In other words, when a leaf is annotated with a symbol, it signifies that, to fully unfold the derivation, this leaf should be recursively replaced by the root of the sub-derivation whose conclusion is a rule with the same label. 
Importantly, for each label, there should be a unique non-leaf occurrence of a rule annotated with that symbol.
%%%%%%%%%%%%%%%%%%%%%%%%%%%%%%%%%%%%%%%%%%%%%%%%%%%%%%%%%%%%

%%%%%%%%%%%%%%%%%%%%%%%%%%%%%%%%%%%%%%%%%%%%%%%%%%%%%%%%%%%%
In order to recover correctness, we introduce the following \emph{progressiveness} criterion.
\begin{definition}\label{def:prog}
	Let $\fA$ be a formula occurring in a sequent of a $\dD\in\Spdl$ conclusion of a rule $\rrule$.
	We say that a formula $\fB$ occurring in a premise of $\rrule$
	is an \defn{immediate ancestor} of $\fA$ 
	whenever
	one of the following holds:
	\begin{itemize}
		\item 
		$\fB$ is an active formula of $\rrule\neq\Krule$
		with principal formula $\fA$;
%		In this case $\fB$ is a subformula of $\fA$;
		
		\item 
		$\fB$ is an active formula of $\Krule$-rule
		and 
		$\fA\in\set{\tbox[\ta]\fB,\tdia[\ta]\fB}$;
		
		\item
		$\fB$ is the unique occurrence of $\fA$ in the sequent.
		
	\end{itemize}
	A \defn{thread} in a derivation $\dD$
	is a maximal sequence of formulas occurring in sequents of $\dD$ 
	totally ordered with respect to the immediate ancestor relation.
	Its first element is called \defn{starting point}.
	A thread is \defn{progressing} if its starting point is a formula $\fA=\tbox[\kstar{\ta}]\fB$ (also called the \defn{principal formula} of the thread) which occurs as active formula of $\Krule$-rules\footnote{
		It is easy to show that the principal formula of a progressing thread which 
		is active for $\Krule$-rules
		infinitely often, 
		is also principal for $\bstarr$-rules infinitely often.
		More precisely, occurrence of these rules are interleaved
		and we can easily show that our 
		progress condition is equivalent to the one in \cite{das:gir:TCL} formulated by means of $\bstarr$-rules.
	}
	infinitely often.
	A derivation is \defn{progressing} if each infinite branch contains a progressing thread.
	We denote by $\pSpdl$ the set of progressing derivations.
\end{definition}
%%%%%%%%%%%%%%%%%%%%%%%%%%%%%%%%%%%%%%%%%%%%%%%%%%%%%%%%%%%%

%%%%%%%%%%%%%%%%%%%%%%%%%%%%%%%%%%%%%%%%%%%%%%%%%%%%%%%%%%%%
\begin{lemma}\label{lem:pSpdl:decomp}
	Each $\dD\in\pSpdl$ is of the following shape
	$$
	\toks0={0.2}
	\vlderivation{
		\vltrf{\dD_0}{\vdash \fGam}{
			\vlpr{\dD_1}{}{\vdash \fGam_1,\tbox[\kstar{\ta_1}]\fA_1}
		}{\vlhy{\quad\cdots\quad}}{
			\vlpr{\dD_n}{}{\vdash \fGam_n,\tbox[\kstar{\ta_n}]\fA_n}
		}{\the\toks0}
	}
	$$
	where
	$\dD_0$ is a finite open derivation
	with $n$ open premises 
	of the form 
	$\fGam_i,\tbox[\kstar{\ta_i}]\fA_i$
	and
	such that
	$\tbox[\kstar{\ta_i}]\fA_i$ is 
	the starting point of a progressing thread in $\dD$ for all $i\in\intset1n$.
\end{lemma}
\begin{proof}
	By definition, each infinite branch of $\dD$ contains a progressing thread, 
	which must have a starting point.
%	 therefore a bottom-most sequent in which the principal formula of such a thread is the principal formula of a $\bstarr$.
	We conclude by letting $\dD_0$ be the approximation of $\dD$ with open premises all such nodes.
\end{proof}
%%%%%%%%%%%%%%%%%%%%%%%%%%%%%%%%%%%%%%%%%%%%%%%%%%%%%%%%%%%%

We can easily prove that if a formula is valid in $\PDL$, then it is derivable in $\pSpdlcut$.

%%%%%%%%%%%%%%%%%%%%%%%%%%%%%%%%%%%%%%%%%%%%%%%%%%%%%%%%%%%%
\begin{lemma}\label{lem:pSpdlcutSound}
	The set of derivations $\pSpdl\cup\cutr$ is complete for $\PDL$.
\end{lemma}
\begin{proof}
	Each axiom in  \Cref{fig:axPDL} is derivable in $\pSpdl$, that is, there is a derivation with conclusion the axiom formula defined as follows (see also \Cref{fig:axDerivations}):
	\begin{itemize}
		\item $\mathbf{PL}$: the sub-system 
		$\set{\Trule,\AXrule,\Wrule,\land,\lor}$ is a well-known sound and complete sequent system for classical logic (see \cite{troelstra_schwichtenberg_2000}, where the system is called the system $\mathsf G_{3}p$);
		
		\item $\mathbf{Neg}$ is immediate by definition of the negation;
		
		\item $\axA\emptyset$ is proven using a single instance of $\tbox[\tempt]$;
		
		\item $\axA\epsilon$, $\axA{;}$ and $\mathbf{K}$ are straightforward using rule $\AXrule$, $\land$, $\lor$ and $\Krule$;

		\item $\axA{?}$ is also straightforward using rule $\AXrule$, $\land$, $\lor$ and $\btestr$ and $\dtestr$;
		
		\item $\axA{\oplus}$ (resp. $\axA{\star}$) require the use of rules
		rule $\AXrule$, $\land$, $\lor$, and both $\bplusr$ and $\dplusr$ (resp. $\bstarr$ and $\dstarr$)
		plus the rule $\Wrule$.
	\end{itemize}
	
	Moreover, each rule in  \Cref{fig:axPDL} is derivable in $\pSpdl$, that is, there is an open derivation in $\pSpdl$ with the same conclusion of the rule and with a single open premise which is the same of the premise of the rule.
	Rules $\mprule$ and $\necrule$ are derivable as shown in \Cref{eq:rulesDer} below, while the loop-invariance rule ($\looprule$) can be simulated by the progressive infinite derivation shown in \Cref{fig:loopInv}.
	Note that right premise of the $\cutr$-rule at the bottom of the derivation in \Cref{fig:loopInv} is  the axiom $\axA{Ind}$ mentioned in \Cref{rem:AXind}.
	
	\begin{equation}\label{eq:rulesDer}
		\begin{array}{@{\hskip4em}c@{\hskip4em}|@{\hskip4em}c@{\hskip4em}}
			\mprule&\necrule
			\\\hline
			\vlderivation{
				\vliin{\cutr}{}{\vdash \fB}{
					\vlpr{}{}{\vdash \fA}
				}{
					\vliin{\cutr}{}{
						\vdash \nfA,\fB
					}{
						\vlpr{}{}{\vdash \nfA \lor \fB}
					}{
						\vliin{\land}{}{
							\vdash \fA \land \nfB, \nfA,\fB
						}{
							\vlin{\AXrule}{}{\vdash \fA ,\nfA}{\vlhy{}}
						}{
							\vlin{\AXrule}{}{\vdash \nfB,\fB}{\vlhy{}}
						}
					}
				}
			}
		&
			\vlinf{\Krule}{}{\vdash \tbox[\ta] \fA}{\vdash \fA} 
		\end{array}
		\qquad\qquad
	\end{equation}

	This allows us to conclude because $\proves[\PDL]\fA$ iff there is a derivation in the Hilbert system made of the axioms and rules in \Cref{fig:axPDL}, and each of such a derivation can be translated into a derivation in $\pSpdlcut$ by replacing each axiom and rule with the ones provided.
\end{proof}
%%%%%%%%%%%%%%%%%%%%%%%%%%%%%%%%%%%%%%%%%%%%%%%%%%%%%%%%%%%%

%%%%%%%%%%%%%%%%%%%%%%%%%%%%%%%%%%%%%%%%%%%%%%%%%%%%%%%%%%%%
%%%%%%%%%%%%%%%%%%%%%%%%%%%%%%%%%%%%%%%%%%%%%%%%%%%%%%%%%%%%
\begin{figure}
	\adjustbox{max width=\textwidth}{$\begin{array}{c}
		\mathsf K:\;
		\vlderivation{
			\vliq{2\times\lor}{}{
				\vdash\tdia[\ta](\fA\land \nfB)\lor \left(\tdia[\ta]\nfA\lor \tbox[\ta]\fB\right)	
			}{
				\vlin{\Krule}{}{
					\vdash\tdia[\ta](\fA\land \nfB), \tdia[\ta]\nfA, \tbox[\ta]\fB
				}{
					\vliin{\land}{}{
						\vdash\fA\land \nfB,\nfA,\fB
					}{
						\vlin{\AXrule}{}{\vdash\fA,\nfA}{\vlhy{}}
					}{
						\vlin{\AXrule}{}{\vdash\fB,\nfB}{\vlhy{}}
					}
				}
			}
		}
	\qquad\qquad
		\axA\epsilon:\;
		\vlderivation{
			\vliin{\land}{}{
				\vdash
				(\tdia[\tepsi]\nfA \lor \fA)
				\land
				(\tbox[\tepsi]\fA \lor \nfA)
			}{
				\vlin{\lor}{}{
					\vdash\tdia[\tepsi]\nfA \lor \fA
				}{
					\vlin{\tdia[\tepsi]}{}{
						\vdash\tdia[\tepsi]\nfA , \fA
					}{
						\vlin{\AXrule}{}{\vdash\nfA,\fA}{\vlhy{}}
					}
				}
			}{
				\vlin{\lor}{}{
					\vdash\tbox[\tepsi]\fA \lor \nfA
				}{
					\vlin{\tbox[\tepsi]}{}{
						\vdash\tbox[\tepsi]\fA , \nfA
					}{
						\vlin{\AXrule}{}{\vdash\nfA,\fA}{\vlhy{}}
					}
				}
			}
		}
	\\
		\axA{?}:\;
		\vlderivation{
			\vliin{\land}{}{
				\vdash
				(\tdia[\tfB]\nfA\lor (\cneg{\fB}\lor \fA ))
				\land
				(\tbox[\tfB]\fA\lor ({\fB}\land \nfA ))
			}{
				\vlin{\lor}{}{
					\vdash\tdia[\tfB]\nfA\lor (\cneg{\fB}\lor \fA )
				}{
					\vlin{\dtestr}{}{
						\vdash\tdia[\tfB]\nfA, (\cneg{\fB}\lor \fA )
					}{
						\vlin{\AXrule}{}{\vdash \fB \land \nfA, \cneg{\fB}\lor \fA }{\vlhy{}}
					}
				}
			}{
				\vlin{\lor}{}{
					\vdash\tbox[\tfB]\fA\lor ({\fB}\land \nfA )
				}{
					\vlin{\btestr}{}{
						\vdash\tbox[\tfB]\fA, (\fB\land \nfA )
					}{
						\vlin{\AXrule}{}{\vdash\nfB\lor \fA, (\fB\land \nfA )}{\vlhy{}}
					}
				}
			}
		}
	\qquad
		\axA{;}:\;
		\vlderivation{
			\vliin{\land}{}{
				\vdash
				(\tdia[\ta;\tb]\nfA\lor \tbox[\ta]\tbox[\tb]\fA)
				\land
				(\tbox[\ta;\tb]\fA\lor \tdia[\ta]\tdia[\tb]\nfA)
			}{
				\vlin{\lor}{}{
					\vdash\tdia[\ta;\tb]\nfA\lor \tbox[\ta]\tbox[\tb]\fA
				}{
					\vlin{\dtestr}{}{
						\vdash\tdia[\ta;\tb]\nfA, \tbox[\ta]\tbox[\tb]\fA
					}{
						\vlin{\AXrule}{}{\vdash\tdia[\ta]\tdia[\tb]\nfA, \tbox[\ta]\tbox[\tb]\fA }{\vlhy{}}
					}
				}
			}{
				\vlin{\lor}{}{
					\vdash\tbox[\ta;\tb]\fA\lor \tdia[\ta]\tdia[\tb]\nfA
				}{
					\vlin{\btestr}{}{
						\vdash \tbox[\ta;\tb]\fA, \tdia[\ta]\tdia[\tb]\nfA
					}{
						\vlin{\AXrule}{}{\vdash\tbox[\ta]\tbox[\tb]\fA, \tdia[\ta]\tdia[\tb]\nfA}{\vlhy{}}
					}
				}
			}
		}
	\\
		\axA{\oplus}:\;
		\vlderivation{
			\vliin{\land}{}{
				\left(\tdia[\ta\oplus \tb]\nfA\lor (\tbox[\ta]\fA \land \tbox[\tb]\fA) \right)
				\land
				\left(\tbox[\ta\oplus \tb]\fA\lor (\tdia[\ta]\nfA \lor \tdia[\tb]\nfA)\right)
			}{
				\vlin{\lor}{}{\vdash
					\tdia[\ta\oplus \tb]\nfA\lor (\tbox[\ta]\fA \land \tbox[\tb]\fA) 
				}{
					\vlin{\dplusr}{}{\vdash
						\tdia[\ta\oplus \tb]\nfA , \tbox[\ta]\fA \land \tbox[\tb]\fA 
					}{
						\vliin{\land}{}{\vdash
							\tdia[\ta]\nfA,\tdia[\tb]\nfA , \tbox[\ta]\fA \land \tbox[\tb]\fA 
						}{
							\vlin{\Wrule}{}{\vdash
								\tdia[\ta]\nfA,\tdia[\tb]\nfA, \tbox[\ta]\fA 
							}{\vlin{\AXrule}{}{\vdash \tdia[\ta]\nfA \tbox[\ta]\fA }{\vlhy{}}} 
						}{
							\vlin{\Wrule}{}{\vdash
								\tdia[\ta]\nfA,\tdia[\tb]\nfA, \tbox[\ta]\fA 
							}{\vlin{\AXrule}{}{\vdash\tdia[\tb]\nfA, \tbox[\ta]\fA }{\vlhy{}}} 
						}
					}
				}
			}{\vliq{2\times\lor}{}{\vdash
					\tbox[\ta\oplus \tb]\fA\lor (\tdia[\ta]\nfA \lor \tdia[\tb]\nfA)
				}{
					\vliin{\bplusr}{}{\vdash
						\tbox[\ta\oplus \tb]\fA ,\tdia[\ta]\nfA , \tdia[\tb]\nfA
					}{
						\vlin{\Wrule}{}{\vdash
							\tbox[\ta]\fA ,\tdia[\ta]\nfA , \tdia[\tb]\nfA
						}{\vlin{\AXrule}{}{\vdash \tbox[\ta]\fA ,\tdia[\ta]\nfA}{\vlhy{}}}
					}{
						\vlin{\Wrule}{}{\vdash
							\tbox[\tb]\fA ,\tdia[\ta]\nfA , \tdia[\tb]\nfA
						}{\vlin{\AXrule}{}{\vdash \tbox[\tb]\fA, \tdia[\tb]\nfA}{\vlhy{}}}
					}
				}
			}
		}
	\\
		\axA{\kstarsymb}:\;
		\vlderivation{
			\vliin{\land}{}{
				\left(\tdia[\kstar\ta]\nfA \lor (\fA\land\tbox[\ta][\kstar\ta]\fA)\right)
				\land
				\left(\tbox[\kstar\ta]\fA\lor  (\nfA\lor\tdia[\ta]\tdia[\kstar\ta]\nfA)\right)
			}{
				\vlin{\lor}{}{\vdash
					\tdia[\kstar\ta]\nfA \lor (\fA\land\tbox[\ta][\kstar\ta]\fA)
				}{
					\vlin{\dstarr}{}{\vdash
						\tdia[\kstar\ta]\nfA ,  (\fA\land\tbox[\ta][\kstar\ta]\fA) 
					}{
						\vliin{}{}{\vdash
							\nfA, \tdia[\ta]\tdia[\kstar\ta]\nfA,  (\fA\land\tbox[\ta][\kstar\ta]\fA) 
						}{
							\vlin{\Wrule}{}{\vdash
								\nfA, \tdia[\ta]\tdia[\kstar\ta]\nfA,  \fA
							}{\vlin{\AXrule}{}{\vdash \nfA,\fA}{\vlhy{}}}
						}{
							\vlin{\Wrule}{}{\vdash
								\nfA, \tdia[\ta]\tdia[\kstar\ta]\nfA, \tbox[\ta][\kstar\ta]\fA 
							}{\vlin{\AXrule}{}{\vdash \tdia[\ta]\tdia[\kstar\ta]\nfA, \tbox[\ta][\kstar\ta]\fA}{\vlhy{}}} 
						}
					}
				}
			}{
				\vlin{2\times\lor}{}{\vdash
					\tbox[\kstar\ta]\fA\lor  (\nfA\lor\tdia[\ta]\tdia[\kstar\ta]\nfA)
				}{
					\vliin{\bstarr}{}{\vdash
						\tbox[\kstar\ta]\fA , \nfA , \tdia[\ta]\tdia[\kstar\ta]\nfA 
					}{
						\vlin{\Wrule}{}{\vdash
							\fA, \nfA , \tdia[\ta]\tdia[\kstar\ta]\nfA
						}{\vlin{\AXrule}{}{\vdash \fA ,\nfA}{\vlhy{}}}
					}{
						\vlin{\Wrule}{}{\vdash
							\tbox[\ta]\tbox[\kstar\ta]\fA, \nfA , \tdia[\ta]\tdia[\kstar\ta]\nfA 
						}{\vlin{\AXrule}{}{\vdash \tbox[\ta]\tbox[\kstar\ta]\fA, \tdia[\ta]\tdia[\kstar\ta]\nfA}{\vlhy{}}} 
					}
				}
			}
		}
	\end{array}$}
	\caption{Non-trivial derivations of the axioms of $\PDL$ in $\Spdl$.}
	\label{fig:axDerivations}
\end{figure}
%%%%%%%%%%%%%%%%%%%%%%%%%%%%%%%%%%%%%%%%%%%%%%%%%%%%%%%%%%%%
%%%%%%%%%%%%%%%%%%%%%%%%%%%%%%%%%%%%%%%%%%%%%%%%%%%%%%%%%%%%

%%%%%%%%%%%%%%%%%%%%%%%%%%%%%%%%%%%%%%%%%%%%%%%%%%%%%%%%%%%%
% FIG loop invariance
%%%%%%%%%%%%%%%%%%%%%%%%%%%%%%%%%%%%%%%%%%%%%%%%%%%%%%%%%%%%
\begin{figure*}
	\adjustbox{max width=\textwidth}{$
		\vlderivation{
			\vliin{\cutr}{}{
				\vdash \nfA,\tbox[\kstar\ta]\fA
			}{
				\vlin{\Krule}{}{
					\vdash \tbox[\kstar\ta] (\nfA\lor\tbox[\ta]\fA)
				}{
					\vlin{\lor}{}{
						\vdash \nfA\lor\tbox[\ta]\fA
					}{
						\vlhy{\vdash \nfA ,\tbox[\ta]\fA}
					}
				}
			}{
				\vliin{\bseqr}{\star}{
					\vdash \nfA,\tdia[\kstar\ta](\fA \land \tdia[\ta]\nfA),\tbox[\kstar\ta]\fA
				}{
					\vlin{\Wrule}{}{\vdash \nfA,\tdia[\kstar\ta](\fA \land \tdia[\ta]\nfA),\fA }{
						\vlin{\AXrule}{}{\vdash \nfA,\fA}{\vlhy{}}
					}
				}{
					\vlin{\dstarr}{}{
						\vdash \nfA,\tdia[\kstar\ta](\fA \land \tdia[\ta]\nfA),\tbox[\ta]\tbox[\kstar\ta]\fA
					}{
						\vliin{\land}{}{
							\vdash 
							\nfA,\fA \land \tdia[\ta]\nfA,
							\tdia[\ta]\tdia[\kstar\ta](\fA \land \tdia[\ta]\nfA),
							\tbox[\ta]\tbox[\kstar\ta]\fA
						}{
							\vliq{2\times\Wrule}{}{
								\vdash \nfA,\fA, \tdia[\ta]\tdia[\kstar\ta](\fA \land \tdia[\ta]\nfA),
								\tbox[\ta]\tbox[\kstar\ta]\fA
							}{
								\vlin{\AXrule}{}{\vdash \nfA,\fA}{\vlhy{}}
							}
						}{
							\vlin{\Wrule}{}{
								\vdash 
								\nfA,\tdia[\ta]\nfA,
								\tdia[\ta]\tdia[\kstar\ta](\fA \land \tdia[\ta]\nfA),
								\tbox[\ta]\tbox[\kstar\ta]\fA
							}{
								\vlin{\Krule}{}{
									\vdash 
									\tdia[\ta]\nfA,
									\tdia[\ta]\tdia[\kstar\ta](\fA \land \tdia[\ta]\nfA),
									\tbox[\ta]\tbox[\kstar\ta]\fA
								}{
									\vlin{\bstarr}{\star}{
										\vdash 
										\nfA,
										\tdia[\kstar\ta](\fA \land \tdia[\ta]\nfA),
										\nfA,\tbox[\kstar\ta]\fA
									}{
										\vlhy{}
									}
								}
							}
						}
					}
				}
			}
		}
		$}
	\caption{Derivability of the loop invariance rule in $\Spdlcut$.}
	\label{fig:loopInv}
\end{figure*}
%%%%%%%%%%%%%%%%%%%%%%%%%%%%%%%%%%%%%%%%%%%%%%%%%%%%%%%%%%%%
%%%%%%%%%%%%%%%%%%%%%%%%%%%%%%%%%%%%%%%%%%%%%%%%%%%%%%%%%%%%

%%%%%%%%%%%%%%%%%%%%%%%%%%%%%%%%%%%%%%%%%%%%%%%%%%%%%%%%%%%%
%%%%%%%%%%%%%%%%%%%%%%%%%%%%%%%%%%%%%%%%%%%%%%%%%%%%%%%%%%%%
\subsection{Cut-Elimination in $\Spdlcut$}\label{subse:PDLcutElim}
%%%%%%%%%%%%%%%%%%%%%%%%%%%%%%%%%%%%%%%%%%%%%%%%%%%%%%%%%%%%
%%%%%%%%%%%%%%%%%%%%%%%%%%%%%%%%%%%%%%%%%%%%%%%%%%%%%%%%%%%%

%%%%%%%%%%%%%%%%%%%%%%%%%%%%%%%%%%%%%%%%%%%%%%%%%%%%%%%%%%%%
% FIG sequebt calculus PDL
%%%%%%%%%%%%%%%%%%%%%%%%%%%%%%%%%%%%%%%%%%%%%%%%%%%%%%%%%%%%
\begin{figure}[t]
	\adjustbox{max width=\textwidth}{$\begin{array}{c}
		\begin{array}{cc}
			\vlderivation{
				\vliin{\cutr}{}{\vdash \fGam}{
					\vlhy{\vdash \fGam,\fA}
				}{
					\vlin{\Wrule}{}{\vdash \fGam,\nfA}{\vlhy{\vdash \fGam}}
				}
			}
			\quad\cutelim\quad
			\vdash \fGam
		&
			\vlderivation{
				\vliin{\cutr}{}{\vdash \fGam,\fA}{\vlhy{\vdash \fGam,\fA}}{\vlin{\AXrule}{}{\vdash \nfA,\fA}{\vlhy{}}}
			}
		\quad\cutelim\quad
			\vdash \fGam,\fA
		\\
		\multicolumn{2}{c}{
			\vlderivation{
				\vliin{\cutr}{}{\vdash \fGam}{
					\vliin{\land}{}{
						\vdash \fGam, \fA\land\fB
					}{
						\vlhy{\vdash \fGam,\fA}
					}{
						\vlhy{\vdash \fGam,\fB}
					}
				}{
					\vlin{\lor}{}{\vdash\fGam, \nfA\lor\nfB}{
						\vlhy{\vdash\fGam, \nfA,\nfB}
					}
				}
			}
			\quad\cutelim\quad
			\vlderivation{
				\vliin{\cutr}{}{\vdash \fGam}{
					\vlhy{\vdash \fGam,\fA}
				}{
					\vliin{\cutr}{}{
						\vdash \fGam, \nfA
					}{
						\vlhy{\vdash \fGam,\fB}
					}{
						\vlhy{\vdash \fGam,\nfA,\nfB}
					}
				}
			}
		}
		\end{array}
		\\
		\begin{array}{r@{\quad\cutelim\quad}l}
			\vlderivation{
				\vliin{\cutr}{}{\vdash \fGam}{
					\vlhy{\vdash \fGam, \tbox[\tempt]\fA}
				}{
					\vlin{\tbox[\tempt]}{}{\vdash \fGam, \tdia[\tempt]\nfA}{\vlhy{\vdash \fGam}}
				}
			}
		&
			\vdash\fGam
		\\
			\vlderivation{
				\vliin{\cutr}{}{\vdash \fGam}{
					\vlin{\tbox[\tepsi]}{}{\vdash \fGam, \tbox[\tepsi]\fA}{\vlhy{\vdash \fGam, \fA}}
				}{
					\vlin{\tdia[\tepsi]}{}{\vdash \fGam, \tdia[\tepsi]\nfA}{\vlhy{\vdash \fGam, \nfA}}
				}
			}
		&%	\quad\cutelim\quad
			\vlderivation{
				\vliin{\cutr}{}{\vdash \fGam}{
					\vlhy{\vdash \fGam, \fA}
				}{
					\vlhy{\vdash \fGam, \nfA}
				}
			}
		\\
			\vlderivation{
				\vliin{\cutr}{}{\vdash \tdia[\ta]\fGam,\tbox[\ta]\fB}{
					\vlin{\Krule}{}{
						\vdash \tdia[\ta]\fGam, \tbox[\ta]\fA
					}{
						\vlhy{\vdash \fGam,\fA}
					}
				}{
					\vlin{\Krule}{}{
						\vdash \tdia[\ta]\nfA,\tdia[\ta]\fGam, \tbox[\ta]\fB
					}{
						\vlhy{\vdash\nfA,\fGam, \fB}
					}
				}
			}
		&
			\vlderivation{
				\vlin{\Krule}{}{\vdash \tdia[\ta]\fGam,\tbox[\ta]\fB}{
					\vliin{\cutr}{}{\vdash \fGam, \fB}{
						\vlhy{\vdash \fGam,\fA}
					}{
						\vlhy{\vdash\nfA,\fGam, \fB}
					}
				}
			}
		\\
			\vlderivation{
				\vliin{\cutr}{}{\vdash \fGam}{
					\vlin{\btestr}{}{\vdash \fGam, \tbox[\tfB]\fA}{\vlhy{\vdash \fGam, \fA\lor\cneg\fB}}
				}{
					\vlin{\dtestr}{}{\vdash \fGam, \tdia[\tfB]\nfA}{\vlhy{\vdash \fGam, \nfA\land\fB}}
				}
			}
			&
			\vlderivation{
				\vliin{\cutr}{}{\vdash \fGam}{
					\vlhy{\vdash \fGam, \fA\lor\cneg\fB}
				}{
					\vlhy{\vdash \fGam, \nfA\land\fB}
				}
			}
		\\
			\vlderivation{
				\vliin{\cutr}{}{\vdash \fGam}{
					\vliin{\bplusr}{}{
						\vdash \fGam, \tbox[\ta\oplus\tb]\fA
					}{
						\vlhy{\vdash \fGam, \tbox[\ta]\fA}
					}{
						\vlhy{\vdash \fGam, \tbox[\tb]\fA}
					}
				}{
					\vlin{\dplusr}{}{\vdash\fGam, \tdia[\ta\oplus\tb]\nfA}{
						\vlhy{\vdash\fGam, \tdia[\ta]\nfA,\tdia[\tb]\nfA}
					}
				}
			}
		&
			\vlderivation{
				\vliin{\cutr}{}{\vdash \fGam}{
					\vlhy{\vdash \fGam, \tbox[\ta]\fA}
				}{
					\vliin{\cutr}{}{
						\vdash \fGam, \tdia[\ta]\nfA
					}{
						\vlhy{\vdash \fGam, \tbox[\tb]\fA}
					}{
						\vlhy{\vdash\fGam, \tdia[\ta]\nfA,\tdia[\tb]\nfA}
					}
				}
			}
		\\
			\vlderivation{
				\vliin{\cutr}{}{\vdash \fGam}{
					\vlin{\bseqr}{}{\vdash \fGam, \tbox[\ta;\tb]\fA}{\vlhy{\vdash \fGam, \tbox[\ta]\tbox[\tb]\fA}}
				}{
					\vlin{\dseqr}{}{\vdash \fGam, \tdia[\ta;\tb]\nfA}{\vlhy{\vdash \fGam, \tdia[\ta]\tdia[\tb]\nfA}}
				}
			}
		&
			\vlderivation{
				\vliin{\cutr}{}{\vdash \fGam}{
					\vlhy{\vdash \fGam, \tbox[\ta]\tbox[\tb]\fA}
				}{
					\vlhy{\vdash \fGam, \tdia[\ta]\tdia[\tb]\nfA}
				}
			}
		\\
			\vlderivation{
				\vliin{\cutr}{}{\vdash \fGam}{
					\vliin{\bstarr}{}{\vdash \fGam, \tbox[\kstar\ta]\fA}{
						\vlhy{\vdash \fGam, \fA}
					}{
						\vlhy{\vdash \fGam, \tbox[\ta;\kstar\ta]\fA}
					}
				}{
					\vlin{\dstarr}{}{\vdash\fGam, \tdia[\kstar\ta]\nfA}{
						\vlhy{\vdash\fGam, \nfA,\tdia[\ta;\kstar\ta]\nfA}
					}
				}
			}
		&
			\vlderivation{
				\vliin{\cutr}{}{\vdash \fGam}{
					\vlhy{\vdash \fGam, \fA}
				}{
					\vliin{\cutr}{}{
						\vdash \fGam, \nfA
					}{
						\vlhy{\vdash \fGam, \tbox[\ta;\kstar\ta]\fA}
					}{
						\vlhy{\vdash\fGam, \nfA,\tdia[\ta;\kstar\ta]\nfA}
					}
				}
			}
		\end{array}
		\\\hline\\
		\begin{array}{cc}
			\vlderivation{
				\vliin{\cutr}{}{\vdash\fGam, \fDel}{
					\vlin{\rrule_1}{}{\vdash\fGam, \fA}{
						\vlhy{\vdash  \fGam_1, \fA}
					}
				}{
					\vlhy{\vdash\nfA, \fDel}
				}
			}
		\quad\cutelim\quad
			\vlderivation{
				\vlin{\rrule_1}{}{\vdash\fGam, \fDel}{
					\vliin{\cutr}{}{\vdash\fGam_1, \fDel
					}{
						\vlhy{\vdash\fGam_1, \fA }
					}{
						\vlhy{\vdash \nfA, \fDel}
					}
				}
			}
			&
			\vlderivation{
				\vliin{\cutr}{}{\vdash\fGam, \fDel}{
					\vliin{\rrule_2}{}{\vdash\fGam, \fA}{\vlhy{\vdash\fGam_1,\fA}}{\vlhy{\vdash\fGam_2}}
				}{
					\vlhy{\vdash\fDel, \nfA}
				}
			}
		\quad\cutelim\quad
			\vlderivation{
				\vliin{\rrule_2}{}{\vdash\fGam,\fDel} {
					\vliin{\cutr}{}{\vdash\fGam_1,\fDel}{\vlhy{\vdash\fGam_1, \fA}}{\vlhy{\vdash\nfA,\fDel}}
				}{
					\vlhy{\vdash\fGam_2}
				}
			}
		\\
			\mbox{with }\rrule_1 \mbox{ unary rule}
			&
			\mbox{with }\rrule_2\in\Set{\land,\bplusr,\bstarr}
		\end{array}
	\end{array}$}
	\caption{
		Cut-elimination steps in $\Spdl$ (with $\Krule$ restricted on $\ta\in\atomProg$).
		The steps in the bottom-most row are called \emph{commutative}.
%		 (with $\rrule_1$ any unary rule and $\rrule_2$ any binary rule).
	}
	\label{fig:cutElim}
\end{figure}
%%%%%%%%%%%%%%%%%%%%%%%%%%%%%%%%%%%%%%%%%%%%%%%%%%%%%%%%%%%%

In order to prove cut-elimination in $\Spdl$, we adapt the proof in \cite{acc:cur:gue:CSL24,acc:cur:gue:CSL24ext} to define an infinitary rewriting 
defined from the cut-elimination steps in \Cref{fig:cutElim}
able to remove all $\cutr$-rules from progressing derivations in $\Spdlcut$.

%%%%%%%%%%%%%%%%%%%%%%%%%%%%%%%%%%%%%%%%%%%%%%%%%%%%%%%%%%%%
\begin{remark}\label{rem:atomiK}
	To reduce the cases taken into account in \Cref{fig:cutElim}, 
	we restrain the rule $\Krule$ to atomic programs $\pa\in\atomProg$.
	The general instance of this rule is derivable 
	reasoning by induction on the structure of $\ta\in\progSet$ using this atomic version of the rule $\Krule$ and rules $\bplusr,\dplusr,\bseqr,\dseqr,\bstarr,\dstarr$.
\end{remark}
%%%%%%%%%%%%%%%%%%%%%%%%%%%%%%%%%%%%%%%%%%%%%%%%%%%%%%%%%%%%

The proof of cut-elimination can be summarized as follows:
\begin{itemize}
	\item 
	we prove that the set of approximations of derivations of a same sequent $\fGam$ is a Scott domain;
	
	\item 
	we then define a strategy of cut-elimination over $\Spdlcut$ as a set of maximal sequences of open derivations obtained by applying cut-elimination steps to open derivations over $\Spdlcut$.
	In such strategy enforce determinism, and we require a \emph{coherence} condition ensuring that each sequence in the strategy starting from an open derivation $\dD$ can be projected \resp{lifted} to a view in the strategy over an open derivation $\dD'$ such that $\dD$ is an approximation of $\dD'$ \resp{$\dD'$ is an approximation of $\dD$}.

	\item 
	we prove that each sequence $\sigma$ in a cut-elimination strategy $\hbh$, where cut-elimination steps are applied \emph{bottom-up},
	defines a Scott-continuous function $f_\sigma$, which associates to each derivation $\dD$ the derivation which is the limit of succession of the greatest $\cutr$-free approximations the derivations in view $\sigma$;
	
	\item 
	we conclude by showing that each $f_\sigma(\dD)$ is a well-defined and progressing derivations.
	
\end{itemize}

\begin{remark}\label{rem:notCurryHoward}
	Note that in \cite{acc:cur:gue:CSL24} the authors rely on the confluence of cut-elimination over finite approximations.
	This property is due to the fact that the system is inspired by the parsimonious linear logic \cite{mazza:15,maz:ter:15},
	a variant of linear logic \cite{gir:ll} following the tradition of light and soft linear logic \cite{girard:98,laf:soft,maz:LLandP}.

	However, such a desirable feature is not possible in $\Spdl$ since this system is an extension of a sequent calculus for classical logic.
	For a simple example of non-confluence, consider the following critical pair of cut-elimination steps, known as a \emph{Lafont critical pair}:
	$$	
	\vlderivation{
		\vliin{\cutr}{}{\vdash\fGam}{
			\vlin{\Wrule}{}{\vdash\fGam,\fA}{\vlpr{\dD'}{}{\vdash\fGam}}
		}{
			\vlin{\Wrule}{}{\vdash\fGam,\nfA}{\vlpr{\dD''}{}{\vdash\fGam}}
		}
	}
	\qquad
	\mbox{with $\dD'$ and $\dD''$ distinct cut-free derivations}
	$$

	Therefore, we define cut-elimination strategies to avoid this issue, by forcing the strategy to be deterministic, fixing an arbitrary choice of cut-elimination steps whenever there are more than one possible step to apply.
%	
	% Note that the lack of confluence does not jeopardize our results because we are interested in proving cut-elimination, not in studying a Curry-Howard correspondence for $\PDL$ -- which would require a different approach because the denotational semantics of $\Spdl$ should extend the one of $\LK$ (see, e.g., \cite{par:lambdamu}).
\end{remark}

We first recall standard definitions on Scott domains and Scott-continuous functions.
%%%%%%%%%%%%%%%%%%%%%%%%%%%%%%%%%%%%%%%%%%%%%%%%%%%%%%%%%%%%
\begin{definition}\label{def:Scott}
	Let $S$ be a set, $S'$ be a subset of $S$, and let $\leq$ be a partial order over $S$.
	We say that $S$ is a \defn{directed set} if for all $x,y\in S$ there is $z\in S$ such that $x\leq z$ and 
	$y\leq z$.
	An \defn{upper bound} of $S'$ is an element $x\in S$ such that $y\leq x$ for all $y\in S'$;
	a  \defn{supremum} of $S'$ (also denoted $\mysup S'$) is an upper bound $x$ such that $x\leq y$ for any $y$ upper bound of $S'$.
	A $c\in S$ is \defn{compact} if for all direct subset $S'\subseteq S$ such that 
	if $\mysup S'$ is defined and $c\leq \mysup S'$, then $c\leq x$ for a $x\in S$.
	
	A \defn{Scott domain} is a pair $\dom=\tuple{S,\leq}$ such that:
	\begin{itemize}
		\item $\dom$ is \defn{directed complete}: every directed subset of $S$ has a supremum;

		\item $\dom$ is \defn{bounded complete}: every subset which has an upper bound has a supremum;
		
		\item $\dom$ is \defn{algebraic}: every element in $\dom$ can be seen as the supremum of a directed set of compact elements of $\dD$.
	\end{itemize}
	
	A function $f$ over a Scott domain $\dom$ is \defn{Scott-continuous} if it is continuous, and it preserves suprema, that is, if $f(\mysup S')=\mysup (f(S'))$.
\end{definition}
%%%%%%%%%%%%%%%%%%%%%%%%%%%%%%%%%%%%%%%%%%%%%%%%%%%%%%%%%%%%

%%%%%%%%%%%%%%%%%%%%%%%%%%%%%%%%%%%%%%%%%%%%%%%%%%%%%%%%%%%%
\begin{proposition}
	The set of open derivations in $\oSpdl$ with conclusion $\fGam$ is a Scott domain (w.r.t. $\prec$)
	with compact elements the open derivations with conclusion $\fGam$.
\end{proposition}
\begin{proof}
	Directed and bounded completeness follows by definition of $\preceq$.
	Algebricity follows by the remark that each open derivation $\dD$ can be seen as the supremum of the set $\apx[\dD]$.
	%	It follows that each finite approximation of a derivation with conclusion $\fGam$ is a compact element.
\end{proof}
%%%%%%%%%%%%%%%%%%%%%%%%%%%%%%%%%%%%%%%%%%%%%%%%%%%%%%%%%%%%

We define maximal cut-elimination strategies as sequences of open derivations obtained by applying the cut-elimination steps to open derivations.
The name \emph{strategy} is intended to evoke the idea of both a reduction strategy for a rewriting system.

%%%%%%%%%%%%%%%%%%%%%%%%%%%%%%%%%%%%%%%%%%%%%%%%%%%%%%%%%%%%
\begin{definition}
	A \defn{cut-elimination view} for $\dD\in\oSpdl$ is a 
	countable sequence $\sigma=\sigma_0\cdot\sigma_1\cdots\sigma_n\cdots$
	(with length $\lengthof{\sigma}\in\N\cup\set{\infty}$)
	of open derivations in $\oSpdl$
	with $\sigma_0=\dD$ 
	and
	such that
	$\sigma_{i+1}$ is obtained by applying a cut-elimination step to $\sigma_i$.
	It is \defn{bottom-up} if each derivation $\sigma_{i+1}$ is obtained applying a cut-elimination step to a bottom-most reducible $\cutr$-rule in $\sigma_i$.
	
	A family of views $\strat$ is a \defn{deterministic maximal cut-elimination strategy} (or simply \defn{\mces}) if:
	\begin{itemize}
		\item 
		$\strat$ is \defn{fully deterministic}: 
		$\strat^\dD\coloneqq\set{\sigma\in\strat\mid \sigma_0=\dD}$ contains a unique view $\sigma^\dD$ for each $\dD\in\oSpdl$;
		
		\item 
		$\strat$ contains \defn{maximal views} only:
		no cut-elimination step can be applied to $\sigma_{\lengthof{\sigma}}$
		for any $\sigma\in\strat$;
		
		% \item 
		% $\strat$ is \defn{complete}: 
		% if $\dD$ can reduce to a $\dD'$ via a cut-elimination step, then $\dD \cdot \sigma\in\strat_\dD$ for all $\sigma\in\strat_{\dD'}$.

		% \item 
		% $\strat$ is \defn{deterministic}: 
		% % if $\dD \cdot \dD' \cdot \sigma' \in \strat_\dD$
		% % such that $\dD'$ is obtained by applying a cut-elimination step to an instance $\rrule$ of the $\cutr$-rule (in $\dD$), then 
		% % $\dD \cdot \dD'' \cdot \sigma'' \notin \strat_\dD$
		% % for any derivation $\dD''\neq \dD'$ and $\sigma'' \notin \strat_{\dD''}$ such that $\dD''$ can be obtained by applying a cut-elimination step to the same instance $\rrule$ of $\cutr$-rule in $\dD$;
		% if $\dD$ contains a reducible $\cutr$-rule, then there is a unique $\dD'$ obtained by applying a cut-elimination step to $\dD$ such that $\dD\cdot \dD'\cdot \sigma' \in \strat$.
		
		\item
		$\strat$ is \defn{coherent} over approximations:
		if $\dD\preceq \dD'$ and 
		$\rrule$ is a reducible $\cutr$-rule both in $\dD$ and in $\dD'$, and $\dD_1$ \resp{$\dD_1'$} is the derivation obtained from $\dD$ \resp{$\dD'$} by applying a cut-elimination step to $\rrule$ , 
		then 
		$\dD\cdot \dD_1\cdot\sigma_1\in\strat^\dD$
		iff 
		$\dD'\cdot \dD'_1\cdot\sigma'_1\in\strat^{\dD'}$
		
	\end{itemize}
	It is \defn{bottom-up} if all views in $\strat$ are.

	Given a bottom-up \mces $\strat$, we define a function $f_\strat$  associating to each $\dD\in\oSpdl$ the greatest (w.r.t. $\preceq$)  $\cutr$-free approximation of $\dD$ obtained by applying the cut-elimination steps to $\dD$ according to the (unique) view in $\strat_\dD$ 
	$$
	f_\strat(\dD)= 
	\mysup_{\dK\in\apx[\dD]}
	\cf{\sigma^\dK_{\lengthof{\sigma^\dK}}}
	$$
\end{definition}
%%%%%%%%%%%%%%%%%%%%%%%%%%%%%%%%%%%%%%%%%%%%%%%%%%%%%%%%%%%%

%%%%%%%%%%%%%%%%%%%%%%%%%%%%%%%%%%%%%%%%%%%%%%%%%%%%%%%%%%%%%
\def\appi#1#2{\mathsf{k}_{#2}({#1})}
\begin{proposition}
	If $\hbh$ is a bottom-up \mces, 
	then $f_\hbh$  is Scott-continuous.
\end{proposition}
\begin{proof}
	Let $\dD\in\oSpdl$ and $\sigma\in\hbh$ such that $\sigma_0=\dD$.
	% Since $\hbh$ is a bottom-up \mces, the derivation $\sigma_i$ is obtained by applying $i\in\N$ cut-elimination steps to the bottom-most $\cutr$-rule.
	% 
	We let $\appi\dD i$ be defined as the greatest approximation of $\dD$ containing all nodes of $\dD$ which will end being below all $\cutr$-rules in $\sigma_i$, that is:
	\begin{itemize}
		\item all nodes which are in the greatest approximation of $\dD$ containing the greatest cut-free approximation of $\dD$, \item all nodes in $\dD$ which will interact with some cut-elimination step applied to reach $\sigma_i$.
	\end{itemize} 

	For any finite approximation $\dD'\in\apx[\dD]$, we can find a $i\in\N$ such that $\dD' \preceq \appi\dD i$ because $\hbh$ is bottom-up and it will eventually push all cuts in $\dD'$ away from the bottom of $\dD$.
	Then, if $\sigma^{\appi\dD i}$ is the view  defined by coherently restricting $\sigma$ on the finite approximation $\appi\dD i$ of $\dD=\sigma_0$ (i.e., a compact in $\apx[\dD]$), 
	we must have $\cf{\sigma_i}\preceq\cf{\sigma^{\appi\dD i}_i} = f_{\hbh}(\appi\dD i)$.
	Thus, 
	$
	f_\strat(\dD)=
	\mysup_{\dK\in\apx[\dD]}
	\cf{\sigma^\dK_{\lengthof{\sigma^\dK}}} 
	= 
	\mysup_{i\geq 0} f_{\hbh}(\appi\dD i)$
	and conclude thanks to the coherence condition.
\end{proof}
%%%%%%%%%%%%%%%%%%%%%%%%%%%%%%%%%%%%%%%%%%%%%%%%%%%%%%%%%%%%%

%%%%%%%%%%%%%%%%%%%%%%%%%%%%%%%%%%%%%%%%%%%%%%%%%%%%%%%%%%%%
\begin{restatable}{theorem}{cutElimPDL}\label{thm:pSpdl:cutElim}
	The rule $\cutr$ is admissible in $\pSpdl$.
\end{restatable}
\begin{proof}
	We can always define a bottom-up \mces, e.g., by always applying a cut-elimination step to the right-most bottom-most $\cutr$-rule in a derivation, and fixing a reductum for each possible critical pair of cut-elimination steps.

	Once we fix a bottom-up \mces $\hbh$, to conclude it suffices to prove that if $\dD\in\pSpdl$, then $f_{\sigma}(\dD)$ is a well-defined progressing derivation for any $\sigma\in\hbh_{\dD}$.
	For this purpose, we prove that each branch $\branch$ in $f_{\sigma}(\dD)$ does not end with an open premise and, if infinite, it contains a progressing thread.
	
	If there is a $k\in\N$ such that $\branch$ occurs in $\sigma_k$ and in all $\sigma_i$ with $i\geq k$,
	then,
	either $\branch$ is finite, ending with a $\AXrule$-rule (since $\dD$ has no open branches), or infinite, in which case $\branch$ contains a progressing thread since progressing threads are preserved by finitely many cut-elimination steps.
	
	Otherwise, we define the \emph{open branch} $\branch_i$ as the set of nodes in $\sigma_i$ containing the nodes in $\branch$ (seen as set of nodes) strictly below any $\cutr$-rule in $\sigma_i$. 
	Note that each $\branch_i$ can be seen as a finite subset of nodes in $\branch$,
	and that the sequence $(\branch_i)_{i\geq 0}$ is well-ordered with supremum $\branch$ by definition.
	The existence of a progressing thread $\rho$ in the infinite (therefore not ending with an open premise) branch $\branch$ is proven by remarking that each cut-elimination steps either do not interact with $\Krule$-rules, or it is a cut-elimination step of the form $\Krule$-vs-$\Krule$.
	In the latter case, the sequence of nodes in $\branch_{i+1}$ strictly longer than $\branch_i$ and 
	it contains an additional `progressing point' of the progressive thread of $\branch$ with respect to $\branch_i$, that is, 
	the number of $\Krule$-rules in $\branch_{i+1}$ with principal formula the one of the progressive thread of $\branch$ is one more than the one in $\branch_i$. This ensures progressiveness of $\branch$, which is the supremum of $(\branch_i)_{i\geq 0}$.
	
	Details can easily obtained by adapting the technique developed in \cite{acc:cur:gue:CSL24ext}, where the modality $\oc$ (resp.~$\wn$) can be considered as a box (resp.~a diamond), and the rule $\mathsf{c!p}$ can plays the same role of $\Krule$ (and $\bstarr$) to define the progressing condition.
\end{proof}
%%%%%%%%%%%%%%%%%%%%%%%%%%%%%%%%%%%%%%%%%%%%%%%%%%%%%%%%%%%%

%%%%%%%%%%%%%%%%%%%%%%%%%%%%%%%%%%%%%%%%%%%%%%%%%%%%%%%%%%%%
%%%%%%%%%%%%%%%%%%%%%%%%%%%%%%%%%%%%%%%%%%%%%%%%%%%%%%%%%%%%
\subsection{Soundness and Completeness of $\pSpdl$}
%%%%%%%%%%%%%%%%%%%%%%%%%%%%%%%%%%%%%%%%%%%%%%%%%%%%%%%%%%%%
%%%%%%%%%%%%%%%%%%%%%%%%%%%%%%%%%%%%%%%%%%%%%%%%%%%%%%%%%%%%

We conclude by proving soundness and completeness of $\pSpdl$ with respect to $\PDL$ relying on the cut-elimination result.
% The omitted details of the proofs provided in \Cref{app:PDL}.

%%%%%%%%%%%%%%%%%%%%%%%%%%%%%%%%%%%%%%%%%%%%%%%%%%%%%%%%%%%%
\begin{lemma}\label{lem:pSpdl:starImpN}
	If $\proves[\pSpdl]\fGam,\tbox[\kstar\ta]\fA$, then $\proves[\pSpdl]\fGam,\tbox[\ta^n]\fA$ for any $n\in\N$.
\end{lemma}
\begin{proof}
	It follows from \Cref{thm:pSpdl:cutElim} since we have a derivation $\dD_n$ 
	defined as in \Cref{fig:starImpN}.
\end{proof}
\begin{figure*}[t]
	\adjustbox{max width=\textwidth}{$\begin{array}{c}
		\vlderivation{
			\vliin{\cutr}{}{\vdash \fGam,\tbox[\ta^n]\fA }{
				\vliq{\lor}{}{\vdash \left(\bigvee\fGam\right) \lor \tbox[\kstar\ta]\fA}{
					\vlpr{}{\IH}{\vdash \fGam,\tbox[\kstar\ta]\fA}
				}
			}{
				\vliin{\land}{}{
					\vdash 
					(\bigvee \cneg\fGam) \land \tdia[\kstar\ta]\nfA,
					\fGam,\tbox[\ta^n]\fA
				}{
					\vlin{\Wrule}{}{\vdash \bigwedge \cneg\fGam,\fGam,\tbox[\ta^n]\fA}{
						\vliiin{\land}{}{\vdash \bigwedge \cneg\fGam,\fGam}{
							\vlin{\AXrule}{}{\vdash \cneg\fGam_1,\fGam_1}{\vlhy{}}
						}{\vlhy{\cdots}}{
							\vlin{\AXrule}{}{\vdash \cneg\fGam_{\sizeof{\fGam}},\fGam_{\sizeof{\fGam}}}{\vlhy{}}
						}
					}
				}{
					\vliq{\Wrule}{%\set{\Wrule,\dstarr}
					}{
						\vdash \tdia[\kstar\ta]\nfA,\fGam,\tbox[\ta^n]\fA
					}{
						\vlpr{\dD_n'}{%\set{\Krule,\bseqr,\tbox[\tepsi]}
						}{
							\vdash \tdia[\kstar\ta]\nfA,\tbox[\ta^n]\fA
							%					}{
							%						\vlin{\AXrule}{}{\vdash \nfA,\fA}{\vlhy{}}
						}
					}
				}
			}
		}
		\\\\\mbox{where}\\
		\begin{array}{c@{\qquad\vrule\qquad}c@{\qquad\vrule\qquad}c}
			\dD_0'
			&
			\dD_1'
			&
			\dD_{n+1}'
		\\\hline
			\vlderivation{
				\vlin{\dstarr}{}{
					\vdash \tdia[\kstar\ta]\nfA, \tbox[\tepsi]\fA
				}{
					\vlin{\Wrule}{}{
						\vdash \nfA,\tdia[\ta;\kstar\ta]\nfA, \tbox[\tepsi]\fA
					}{
						\vlin{\tbox[\tepsi]}{}{
							\vdash \nfA, \tbox[\tepsi]\fA
						}{
							\vlin{\AXrule}{}{\vdash \nfA, \fA}{\vlhy{}}
						}
					}
				}
			}
			&
			\vlderivation{
				\vlin{\dstarr}{}{
					\vdash \tdia[\kstar\ta]\nfA, \tbox[\ta]\fA
				}{
					\vlin{\Wrule}{}{
						\vdash \tdia[\ta]\nfA,\tdia[\ta;\kstar\ta]\nfA, \tbox[\ta]\fA 
					}{
						\vlin{\AXrule}{}{
							\vdash \tdia[\ta]\nfA, \tbox[\ta]\fA
						}{\vlhy{}}
					}
				}
			}
			&
			\vlderivation{
				\vlin{\Wrule+\dstarr}{}{
					\vdash  \tdia[\kstar\ta]\nfA, \tbox[\ta^n]\fA
				}{
					\vlin{\dseqr+\bseqr}{}{
						\vdash  
						\tdia[\ta;\kstar\ta]\nfA, \tbox[\ta^n]\fA
					}{
						\vlin{\Krule}{}{
							\vdash  
							\tdia[\ta]\tdia[\kstar\ta]\nfA, \tbox[\ta]\tbox[\ta^{n-1}]\fA
						}{
							\vlpr{\dD'_{n-1}}{}{
								\vdash 
								\tdia[\kstar\ta]\nfA, \tbox[\ta^{n-1}]\fA 
							}
						}
					}
				}
			}
		\end{array}
	\end{array}$}
	\caption{Derivations proving that if $\proves[\pSpdl]\fGam,\tbox[\kstar\ta]\fA$, then $\proves[\pSpdl]\fGam,\tbox[\ta^n]\fA$ for any $n\in\N$.}
	\label{fig:starImpN}
\end{figure*}
%%%%%%%%%%%%%%%%%%%%%%%%%%%%%%%%%%%%%%%%%%%%%%%%%%%%%%%%%%%%

%%%%%%%%%%%%%%%%%%%%%%%%%%%%%%%%%%%%%%%%%%%%%%%%%%%%%%%%%%%%
To prove soundness and completeness of $\pSpdl$ with respect to $\pSpdl$, we use the notion of \defn{Fischer-Ladner closure} of a formula $\fA$.
This is  defined as the smallest set of formulas $\FL\fA$ containing $\fA$ and such that the conditions in \Cref{eq:FL} hold.
\begin{equation}\label{eq:FL}\adjustbox{max width=\textwidth}{$
	\begin{array}{ll}
		\FL{\fP}=\set{\fP}
		&
		\FL{\nfP}=\set{\nfP}
		\\
		\FL{\fA\land\fB}\supset (\FL{\fA}\cup\FL{\fB})
		&
		\FL{\fA\lor\fB}\supset (\FL{\fA}\cup\FL{\fB})
		\\
		\FL{\tbox[\ta]\fA}\supset \FL{\fA}
		&
		\FL{\tdia[\ta]\fA}\supset \FL{\fA}
		\\
		\FL{\tbox[\tfA]\fB}\supset (\FL{\fA}\cup\FL{\fB})
		&
		\FL{\tdia[\tfA]\fB}\supset (\FL{\fA}\cup\FL{\fB})
		\\
		\FL{\tbox[\ta\oplus\tb]\fA}\supset (\FL{\tbox[\ta]\fA}\cup \FL{\tbox[\tb]\fA})
		&
		\FL{\tdia[\ta\oplus\tb]\fA}\supset (\FL{\tdia[\ta]\fA}\cup \FL{\tdia[\tb]\fA})
		\\
		\FL{\tbox[\ta;\tb]\fA}\supset (\FL{\tbox[\ta]\tbox[\tb]\fA})
		&
		\FL{\tdia[\ta;\tb]\fA}\supset (\FL{\tdia[\ta]\tdia[\tb]\fA})
		\\
		\FL{\tbox[\kstar\ta]\fA}\supset (\FL{\tbox[\ta]\tbox[\kstar\ta]\fA})
		&
		\FL{\tdia[\kstar\ta]\fA}\supset (\FL{\tdia[\ta]\tdia[\kstar\ta]\fA})
	\end{array}
$}\end{equation}
If $\fGam$ is a  sequent, then $\FL\fGam\coloneqq\bigcup_{\fA\in\fGam}\FL\fA$.

\begin{remark}[Fisher-Ladner Analyticity]\label{rem:FLanal}
	By rules inspection, 
	each sequent occurring in a derivation $\dD\in\Spdl$ with conclusion $\fGam$
	is a subset of $\FL\fGam=\bigcup_{\fA\in\fGam}\FL\fA$.
	More precisely, if $\fDel$ is a premise of a rule with conclusion $\fGam$, then $\FL\fDel\subseteq\FL\fGam$.
\end{remark}
%%%%%%%%%%%%%%%%%%%%%%%%%%%%%%%%%%%%%%%%%%%%%%%%%%%%%%%%%%%%

%%%%%%%%%%%%%%%%%%%%%%%%%%%%%%%%%%%%%%%%%%%%%%%%%%%%%%%%%%%%
\begin{theorem}\label{thm:PDL:SC}
	Let $\fGam$ be a sequent.
	Then $\proves[\pSpdl]\fGam$ iff $\proves[\PDL]\fGam$.
\end{theorem}
\begin{proof}
	Completeness of $\pSpdl$ with respect to $\PDL$ is a consequence of \Cref{lem:pSpdlcutSound} and $\cutr$-admissibility (\Cref{thm:pSpdl:cutElim}).

	To prove that $\pSpdl$ is sound for $\PDL$, 
	we first observe that each rule in $\pSpdl$ is locally sound, 
	that is, 
	if each premise of a rule is valid in $\PDL$, then its conclusion is.
	As a consequence, if a sequent $\fGam$ is a conclusion of a derivation $\dD$ in $\pSpdl$ is not valid in $\PDL$, we deduce that $\dD$ must be infinite.
	Then, by \Cref{lem:pSpdl:decomp}, $\dD$ can be written as a finite open derivation with open premises of the form $\fGam,\tbox[\kstar\ta]\fA$ which are derivable in $\pSpdl$.
	We deduce that if the conclusion of $\dD$ is not valid in $\PDL$, then
	there must exist a sequent of the form $\fGam,\tbox[\kstar\ta]\fA$ which is derivable in $\pSpdl$ (via an infinite derivation) but whose conclusion is not valid in $\PDL$.
	Therefore to prove soundness it suffices to prove the the statement for infinite derivations in $\pSpdl$ with conclusion a sequent of the form $\fGam,\tbox[\kstar\ta]\fA$.

	Let $\fGam,\tbox[\kstar\ta]\fA$ 
	and such that
	$\proves[\pSpdl]\fGam,\tbox[\kstar\ta]\fA$ but $\not\proves[\PDL]\fGam,\tbox[\kstar\ta]\fA$.
	We can assume $\fGam,\tbox[\kstar\ta]\fA$ to be minimal with respect to the well-founded partial order over sequents defined by the inclusion of the Fisher-Lander closures of the sequents (see \Cref{rem:FLanal}).
	By \Cref{lem:PDL:starImpN} there is a $n\in \N$ minimal
	such that $\not\proves[\PDL] \fGam, \tbox[\ta^n]\fA$,
	while by 
	\Cref{lem:pSpdl:starImpN}
	we have that 
	if $\proves[\pSpdl]\fGam,\tbox[\kstar\ta]\fA$;
	then we must have that 
	$\proves[\pSpdl]\fGam,\tbox[\kstar\ta]\fA$
	but $\not\proves[\PDL]\fGam,\tbox[\kstar\ta]\fA$.
	This would only be possible if 
	$\fGam,\tbox[\kstar\ta]\fA$
	is not minimal since all rules in $\pSpdl$ are analytic as intended in \Cref{rem:FLanal}. Absurd.
\end{proof}
%%%%%%%%%%%%%%%%%%%%%%%%%%%%%%%%%%%%%%%%%%%%%%%%%%%%%%%%%%%%

%%%%%%%%%%%%%%%%%%%%%%%%%%%%%%%%%%%%%%%%%%%%%%%%%%%%%%%%%%%%
%%%%%%%%%%%%%%%%%%%%%%%%%%%%%%%%%%%%%%%%%%%%%%%%%%%%%%%%%%%%
%%%%%%%%%%%%%%%%%%%%%%%%%%%%%%%%%%%%%%%%%%%%%%%%%%%%%%%%%%%%
\section{Embedding Operational Semantics in Propositional Dynamic Logic}\label{sec:DOL}
%%%%%%%%%%%%%%%%%%%%%%%%%%%%%%%%%%%%%%%%%%%%%%%%%%%%%%%%%%%%
%%%%%%%%%%%%%%%%%%%%%%%%%%%%%%%%%%%%%%%%%%%%%%%%%%%%%%%%%%%%
%%%%%%%%%%%%%%%%%%%%%%%%%%%%%%%%%%%%%%%%%%%%%%%%%%%%%%%%%%%%
\def\OS{\mathcal O}
\def\axO{\mathbf{O}_{\mathsf{Atom}}}
\def\axOind{\mathbf{O}_{\mathsf{Ind}}}
\def\osKlee{\OS_{\mathbb K}}
\def\Pf{\mathcal F_\progSet}

We consider a new set of formulas defined as \PDL-formulas
where the programs in $\progSet$ are provided with an \emph{operational semantics}.

\def\treq{\sim_{\mathsf{Tr}}}
\def\tracesof#1{\mathsf{Tr}(#1)}
%%%%%%%%%%%%%%%%%%%%%%%%%%%%%%%%%%%%%%%%%%%%%%%%%%%%%%%%%%%%%
\begin{definition}
	Let $\progSet$ be a set of programs possibly containing a set of \defn{tests} $\atomTests$.
	An \defn{operational semantics} for a set of programs $\progSet$
	with labels in $\labSet$\footnote{
		Unless specified otherwise, we can assume $\labSet\subseteq\progSet$.
	}
	is
	a labeled binary relation between programs 
	$\OS\subset\progSet\times\labSet\times\progSet$ 
	whose elements are called (labeled) transitions and may be written as $\pa\sosto{\tb}\pc$ instead of $(\pa,\tb,\pc)$.
	We assume $\progSet$ contains a distinguished program $\pepsi$ (called \defn{terminated program}) such that 
	$(\pepsi,\tb,\pc)\notin\OS$ for any 
	$\tb\in\labSet$ and $\pc\in\progSet$.
	An operational semantics is \defn{finitely branching}
	if 
	the set of $\{(\tb,\pc) \mid \pa\sosto{\tb}\pc\}$
	is finite 
	for all 
	$\pa\in\progSet$~\cite[Def.~2.2]{AFV01}.
	
	A \defn{trace} is a finite sequence of labels
	\footnote{
		We may use the color \tcol{green} for traces whenever we want to distinguish them from general programs (in \pcol{red}).
	}%
	.
	A trace $\ta$ is \defn{valid} for a program $\pb$ if 
	there is 
	a trace $\ta'$ 
	such that 
	$\ta=\tat;\ta'$
	which is 
	valid for a $\pb'$
	such that $\pb\sosto{\tat}\pb'$.
	We denote by $\tracesof\pa$ the set of traces valid for $\pa$.
	Two programs are \defn{trace equivalent} (denoted $\pa\treq\pb$)
	if $\tracesof\pa=\tracesof\pb$.
	
\end{definition}
\begin{definition}
	The \defn{(operational) Fisher-Ladner closure} of a formula $\fA$ is defined as the smallest set of formulas closed with respect to conditions given for the Fisher-Ladner closure in \Cref{eq:FL} plus the following:
	\begin{equation*}\adjustbox{max width=.48\textwidth}{$
		\begin{array}{ll}
			\FL{\tbox[\pa]\fA}\supseteq (\FL{\tbox[\tb]\tbox[\pc]\fA})
			&
			\mbox{for all $\tb$ such that $\pa\sosto\tb\pc$}
		\\
			\FL{\tdia[\pa]\fA}\supseteq (\FL{\tdia[\tb]\tdia[\pc]\fA})
		&
			\mbox{for all $\tb$ such that $\pa\sosto\tb\pc$}
		\end{array}
	$}\end{equation*}
\end{definition}
%%%%%%%%%%%%%%%%%%%%%%%%%%%%%%%%%%%%%%%%%%%%%%%%%%%%%%%%%%%%%

%%%%%%%%%%%%%%%%%%%%%%%%%%%%%%%%%%%%%%%%%%%%%%%%%%%%%%%%%%%%%
\begin{definition}[Dynamic Operational  Logic]
	\label{def:OPDL}
	Let $\OS$ be an finitely branching operational semantics for a set of programs $\progSet$.
	The set $\Pf$ of \defn{$\progSet$-formulas} (or simply \defn{formulas} when clear)
	is defined by the same grammar in \Cref{fig:PDLform} 
	by letting 
	$\atomProg=\labSet\cup \left(\progSet\setminus\atomTests\right)$
	and
	by assuming 
	that the set of propositional atoms $\atomSet$ is such that
	$\atomTests\subseteq\set{\tfA \mid \fA \in \Pf}$.%
	\footnote{
		The condition on the set of propositional atoms ensures us that any evaluation of a conditional or a guard required in the operational semantics 
		can be evaluated in the logic itself without the need of an external language.
		See \Cref{sec:chor}.
	}
	
	We write $\proves[\POL] \fA$ 
	if $\fA$ is derivable 
	using rules and axioms of $\PDL$ (\Cref{fig:axPDL})
	plus
	the 
	following axiom
	\begin{equation}\label{eq:OSax}
		\axA{\OS}%^\ta	
		\;:\;
		\tbox[\pa]\fA \fequiv 
		\left(
		\bigwedge\limits_{\pa\sosto\tb\pc}%^{\tb\text{ atomic}}
		\tbox[\tb]\tbox[\pc]\fA \right)
	\end{equation}
	The \defn{operational propositional dynamic logic} of $\OS$ 
	(denoted $\POL[\OS]$, or simply $\POL$ if $\OS$ is clear) 
	is the set of formulas such that $\proves[\POL] \fA$.
\end{definition}
%%%%%%%%%%%%%%%%%%%%%%%%%%%%%%%%%%%%%%%%%%%%%%%%%%%%%%%%%%%%%

%%%%%%%%%%%%%%%%%%%%%%%%%%%%%%%%%%%%%%%%%%%%%%%%%%%%%%%%%%%%%
\begin{remark}\label{rem:finiteBranching}
	In this paper, we only consider finitely-branching operational semantics. At the syntactic level, this guarantees that the axiom $\axA{\OS}$ is a finite formula.
	Note that this condition does not guarantee the so-called \emph{small world property} for models of $\POL$, nor that the Fisher-Ladner closure of a sequent is finite.
\end{remark}
%%%%%%%%%%%%%%%%%%%%%%%%%%%%%%%%%%%%%%%%%%%%%%%%%%%%%%%%%%%%%

%%%%%%%%%%%%%%%%%%%%%%%%%%%%%%%%%%%%%%%%%%%%%%%%%%%%%%%%%%%%
\begin{example}\label{ex:DOL:PDL}
	The standard $\PDL$ can be recovered as the $\POL$ where the set of programs $\progSet$ is the set of regular programs generated from a set on instructions $\atomProg$ and a set of tests $\atomTests$ (i.e. a Kleene algebra with tests)
	provided with the following operational semantics $\mathcal{K}$:
	\begin{equation}\label{eq:osKlee}
		\begin{array}{r@{\;}c@{\;}l@{\qquad}r@{\;}c@{\;}l@{\qquad}r@{\;}c@{\;}l}
			\pat;\pb 		&\sosto{\tat}	& \pb		&
			\pa\oplus\pb 	&\sosto{\tepsi}	& \pa 		&
			\pa\oplus\pb 	&\sosto{\tepsi}	& \pb		
			\\
			\tfA;\pb 		&\sosto{\tfA}	& \pb		&
			\kstar\pa 		&\sosto{\tepsi}	& \pepsi	&
			\kstar\pa	 	&\sosto{\tepsi}	& \pa;\kstar\pa
		\end{array}
	\end{equation}
	Note that
	if we identify the 
	sequential composition, choice and iteration in $\progSet$ with the analogous operations we use to generate regular languages, 
	then 
	each instance of the axiom $\axA\OS$ in $\POL[\mathcal K]$
	is derivable using the axioms 
	$\axA{;}$, $\axA{\kstarsymb}$ and $\axA{\oplus}$.
\end{example}

\paragraph{Recovering previous approaches}
We give a few examples of how previous approaches can be recovered in our framework, for the reader familiar with the literature on $\PDL$ and its applications to concurrency.

\begin{example}\label{ex:DOL:CCSstar}
	In \cite{PDL:interleaving,BENEVIDES201723} the authors study different versions of \PDL for fragments of $\CCS^\kstarsymb$ \cite{busi:gab:zav:RRI}, the restriction of Milner's \CCS~\cite{M80} that replaces recursion with iteration \`a la Kleene star. These versions can be recovered in $\POL$ by instantiating it with the operational semantics considered in those papers.
	In particular, in \cite{PDL:interleaving} the parallel constructor (see \Cref{fig:CCS}) is restricted to be at the root of the syntactic tree of the terms representing processes,
	while in \cite{BENEVIDES201723} the parallel is entirely removed.
\end{example}
\begin{example}\label{ex:DOL:pi}
	The logic $\pi\DL$ in \cite{PDL:pi} is the $\POL$ of the unconventional version of the $\pi$-calculus where the \emph{replication} constructor (usually denoted by $!$) is replaced by the \emph{iteration} (denoted by $\kstarsymb$).
	That is, 
	the version of the $\pi$-calculus they consider 
	stays at
	the standard $\pi$-calculus 
	as 
	$\CCS^\kstarsymb$ 
	stays at 
	the standard $\CCS$.
\end{example}
\begin{example}\label{ex:DOL:A}
	The logic $\mathsf{APDL}$ in \cite{PDL:flow} (based on an idea proposed in \cite{pratt:PDL})
	reminds a test-free fragment of the $\POL$ where a program $\pa^i_j$ is a path over a finite state automaton from a state $i$ to a state $j$.
	However, the induction axiom in $\mathsf{APDL}$,
	which can be reformulated in the following way
	$$
	\axA{5}\colon\quad
	\left(\bigwedge_{\pa^i_j\sosto{\pa^i_k}\ta^k_l}^{\ta^k_l \text{ transition}}
	\tbox[\pa^i_k]\left(\fA_k\imp \tbox[\ta^k_l] \fA_l\right)\right) \imp \left(\fA_i \imp \tbox[\pa^i_j]\fA_j\right)
	$$
	would be derivable in 
	a $\POL$ with programs defined as paths over a finite state automaton only if $\fA_i=\fA_j=\fA_k=\fA_l$.
	Note that $\axA{5}$ require that $\ta^k_l$ is a transition in the automaton
	but there is no such condition on $\pa^i_k$ (nor that $\pa^i_k$ is an elementary path).
	This may be problematic in an automaton containing loops.
\end{example}
%%%%%%%%%%%%%%%%%%%%%%%%%%%%%%%%%%%%%%%%%%%%%%%%%%%%%%%%%%%%

%%%%%%%%%%%%%%%%%%%%%%%%%%%%%%%%%%%%%%%%%%%%%%%%%%%%%%%%%%%%
%%%%%%%%%%%%%%%%%%%%%%%%%%%%%%%%%%%%%%%%%%%%%%%%%%%%%%%%%%%%
\subsection{Soundness and Completeness of $\POL$}
%%%%%%%%%%%%%%%%%%%%%%%%%%%%%%%%%%%%%%%%%%%%%%%%%%%%%%%%%%%%
%%%%%%%%%%%%%%%%%%%%%%%%%%%%%%%%%%%%%%%%%%%%%%%%%%%%%%%%%%%%

We conclude this section by proving soundness and completeness of the axiomaxization of $\POL$ with respect to its semantics.
We then consider a sequent system extending $\Spdl$, and we prove its soundness and completeness with respect to $\POL$ by lifting the method used in the previous section.

%%%%%%%%%%%%%%%%%%%%%%%%%%%%%%%%%%%%%%%%%%%%%%%%%%%%%%%%%%%%
\begin{definition}\label{def:kripOS}
	A \defn{model (for $\POL$)} is a Kripke frame defined similarly to \Cref{def:krip}
	with the following differences:
	\begin{itemize}
		\item the meaning of each label $\labSet$ is defined by an accessibility relation $\meanof\pat\subseteq W\times W$;
		
		\item the meaning of a (non-atomic) program $\pa\in\progSet$ is defined as
		\begin{equation}\label{eq:meaning}
			\meanof{\pa}= \bigcup\limits_{\pa\sosto\tb\pc}%^{\tb\text{ atomic}}
			\meanof{\tb;\pc}
			\mydot
		\end{equation}
	\end{itemize}
	The satisfability relation $\vDash$ is defined analogously to \Cref{def:krip} by considering $\POL$ models.
\end{definition}
%%%%%%%%%%%%%%%%%%%%%%%%%%%%%%%%%%%%%%%%%%%%%%%%%%%%%%%%%%%%

%%%%%%%%%%%%%%%%%%%%%%%%%%%%%%%%%%%%%%%%%%%%%%%%%%%%%%%%%%%%
In order to prove soundness and completeness, 
we recall that a formula $\fA$ is said \defn{refutable} if $\vdash \cneg\fA$,
and \defn{consistent} if not refutable, that is, $\not\vdash\cneg\fA$.
Similarly, a (possibly infinite) set of formulas $A$ is \emph{consistent} 
if there is no refutable finite subset $\set{\fA_1,\ldots,\fA_n}$ of $A$,
that is, if $\not\vdash\cneg{\fA_1\land\cdots\land\fA_n}$.
\begin{theorem}\label{thm:DOLax:sc}
	Let $\fA$ be a formula. 
	Then 
	$\proves[\POL] \fA$  iff $\vDash_{\POL} \fA$.
\end{theorem}
\begin{proof}
	Soundness of $\proves[\POL]$ with respect to $\vDash_{\POL}$ follows by definition of models.
	To prove the completeness, 
	we adapt the proof of completeness for $\PDL$ in \cite{PDLcompleteness}.
	In particular, 
	proving completeness is equivalent to prove that 
	if $\fA$ is consistent, then $\not\vDash \cneg\fA$.
	By definition $\not\vDash \cneg\fA$ holds iff there is a model $\model$ and a world $w$ of $\model$ such that $\model, w \not\vDash\cneg\fA$,
	therefore $\model, w \vDash\fA$.
	For each formula $\fA$ construct such a model $\model_\fA$ as follows:
	\begin{itemize}
		\item 
		the model $\model_\fA$ has a world $w_A$
		for each 
		maximal consistent sets of formulas $A$
		such that,
		for each $\fB\in\FL\fA$,
		either $\fB\in A$ or $\cneg\fB\in A$;
		
		\item 
		for each $\tat\in\labSet$ we have
		$(w_A,w_B)\in\meanof{\tat}$
		iff 
		$A \cup \tdia[\tat]B$ is a consistent set of formulas.
	\end{itemize}
	In $\model_\fA$, the following properties hold:
	\begin{itemize}
		\item 
		if $A \cup \tdia[\pa]B$ is consistent, 
		then $(w_A,w_B)\in\meanof{\pa}$;
		
		\item 
		if $\tdia[\pa]\fB\in\FL\fA$
		and
		$w_A\in\meanof{\tdia[\pa]\fB}$, 
		then
		there is a $w_B$ such that $(w_A,w_B)\in\meanof{\pa}$;
		
		\item 
		for any $\fB\in\FL\fA$, 
		we have that 
		$w_A\in\meanof\fB$ iff $\fB\in A$;
	\end{itemize}

	We conclude that 
	if $\fA$ is consistent, 
	then $\meanof{\fA}\neq\emptyset$;
	therefore there is $w_A$ such that $\model_\fA,w_A\vDash \fA$.
\end{proof}
%%%%%%%%%%%%%%%%%%%%%%%%%%%%%%%%%%%%%%%%%%%%%%%%%%%%%%%%%%%%

%%%%%%%%%%%%%%%%%%%%%%%%%%%%%%%%%%%%%%%%%%%%%%%%%%%%%%%%%%%%
\begin{definition}
	We define the sequent system $\Sdol= \Spdl\cup \Set{\tbox[\OS],\tdia[\OS]}$
	defined by the rules in \Cref{fig:seqPDL} plus the two following rules capturing the axiom $\axA\OS$,
	
	\begin{equation}\label{fig:seqDOL}
		\begin{array}{c}
			\vliiinf{\bosr}{\dagger}{
				\vdash\fGam,\tbox[\pa]\fA
			}{
				\vdash\fGam,\tbox[\tb_1]\tbox[\pc_1]\fA
			}{\cdots}{
				\vdash\fGam,\tbox[\tb_n]\tbox[\pc_n]\fA
			}
			\qquad
			\vlinf{\dosr}{\dagger\S}{
				\vdash\fGam,\tdia[\pa]\fA
			}{
				\vdash\fGam,\tdia[\tb_1]\tdia[\pc_1]\fA,\ldots, \tdia[\tb_n]\tdia[\pc_n]\fA
			}
		\end{array}
	\end{equation}
	where the side condition $\dagger$ requires that the set $\set{(\tb_i,\pc_i)\mid i\in\intset1n}=\set{(\tb,\pc)\mid\pa\sosto\tb\pc}$ is finite, and the side condition $\S$ requires that $\Gamma$ is nonempty if $n=0$.
	
	The definition of progressive derivation in $\Sdol$ is obtained by replacing any occurrence of $\Spdl$ in \Cref{def:prog} with $\Sdol$.
	We denote by $\pSdol$ the set of progressing derivations in $\Sdol$.
\end{definition}
%%%%%%%%%%%%%%%%%%%%%%%%%%%%%%%%%%%%%%%%%%%%%%%%%%%%%%%%%%%%

%%%%%%%%%%%%%%%%%%%%%%%%%%%%%%%%%%%%%%%%%%%%%%%%%%%%%%%%%%%%
\begin{remark}
	The restriction on the operational semantics discussed in \Cref{rem:finiteBranching} also guarantees that the sequent rules in \Cref{fig:seqDOL} capturing the axiom $\axA\OS$ have a finite number of premises (in the case of the rule $\tbox[\OS]$)
	and finite-sequent premise (in the case of the rule $\tdia[\OS]$).
\end{remark}
%%%%%%%%%%%%%%%%%%%%%%%%%%%%%%%%%%%%%%%%%%%%%%%%%%%%%%%%%%%%

%%%%%%%%%%%%%%%%%%%%%%%%%%%%%%%%%%%%%%%%%%%%%%%%%%%%%%%%%%%%
\begin{theorem}\label{thm:cutElimDOL}
	The rule $\cutr$ is admissible in $\Sdol$.
\end{theorem}
\begin{proof}
	The proof of cut-elimination for progressing derivations in $\Sdol\cup\set\cutr$ is similar to the one provided in \Cref{subse:PDLcutElim} and it only requires to consider the additional cut-elimination step in \Cref{fig:seqDOLcutElim}. Note that, in the degenerate case of $n=0$, the step simply rewrites the left-hand side into the premise of the application of rule $\tdia[\OS]$.
	Also, this step does not affect any of the reasoning on threads, which are crucial to guaranteeing the preservation of the progressing condition in \Cref{subse:PDLcutElim}.
\end{proof}
%%%%%%%%%%%%%%%%%%%%%%%%%%%%%%%%%%%%%%%%%%%%%%%%%%%%%%%%%%%%

%%%%%%%%%%%%%%%%%%%%%%%%%%%%%%%%%%%%%%%%%%%%%%%%%%%%%%%%%%%%
% FIG proofs axioms O
%%%%%%%%%%%%%%%%%%%%%%%%%%%%%%%%%%%%%%%%%%%%%%%%%%%%%%%%%%%%
\begin{figure}[t]
	$$
	\vlderivation{
		\vliin{\land}{}{\vdash 
			\tbox[\pa]\fA 
			\fequiv
			\left(\bigwedge_{i=1}^n\tbox[\tb_i]\tbox[\pc]\fA \right)
		}{\vlin{\lor}{}{
				\vdash \tdia[\pa]\nfA \lor
				\left(\bigwedge_{i=1}^n\tbox[\tb_i]\tbox[\pc]\fA \right)
			}{
				\vlin{\land}{}{\vdash \tdia[\pa]\nfA ,
					\left(\bigwedge_{i=1}^n\tbox[\tb_i]\tbox[\pc]\fA \right)
				}{
					\vlhy{\left\{\vlderivation{
							\vlin{\tdia[\OS]}{}{
								\vdash \tdia[\pa]\nfA ,\tbox[\tb_i]\tbox[\pc]\fA
							}{
								\vliq{(n-1)\times\Wrule}{}{
									\vdash \tdia[\tb_1][\pc]\nfA, 
									\ldots ,
									\tdia[\tb_n][\pc]\nfA ,
									\tbox[\tb_i]\tbox[\pc]\fA
								}{
									\vlin{\AXrule}{}{\vdash \tdia[\tb_i]\tdia[\pc]\nfA ,\tbox[\tb_i]\tbox[\pc]\fA}{\vlhy{}}
								}
							}
						}\right\}_{\beta\in X}}
				}
			}
		}{
			\vlin{\lor}{}{
				\vdash \tbox[\pa]\fA \lor
				\left(\bigvee_{i=1}^n\tdia[\tb_i]\tdia[\pc]\nfA \right)
			}{
				\vlin{\tbox[\OS]}{}{\vdash \tbox[\pa]\fA ,
					\left(\bigvee_{i=1}^n\tdia[\tb_i]\tdia[\pc]\nfA \right)
				}{
					\vlhy{\left\{\vlderivation{
							\vliq{(n-1)\times\lor}{}{
								\vdash \tbox[\tb_i]\tbox[\pc]\fA,
								\left(\bigvee_{i=1}^n\tdia[\tb_i]\tdia[\pc]\nfA \right)
							}{
								\vliq{(n-1)\times\Wrule}{}{
									\vdash \tbox[\tb_i]\tbox[\pc]\fA,
									\tdia[\tb_1]\tdia[\pc]\nfA,\ldots,
									\tdia[\tb_n]\tdia[\pc]\nfA
								}{
									\vlin{\AXrule}{}{\vdash \tbox[\tb_i]\tbox[\pc]\fA,
										\tdia[\tb_i]\tdia[\pc]\nfA
									}{\vlhy{}}
								}
							}
						}\right\}_{\tb\in X}}
				}
			}			
		}
	}
	$$
	\caption{
		Derivation in $\pSdol$ of the axiom $\axO$, where we are assuming $\set{\tb_1,\ldots,\tb_n}=\set{\tb\in\labSet\mid \pa\sosto\tb\pc}$.
		The case $n=0$ is trivial and it is omitted.
	}
	\label{fig:proofAxO}
\end{figure}
%%%%%%%%%%%%%%%%%%%%%%%%%%%%%%%%%%%%%%%%%%%%%%%%%%%%%%%%%%%%
%%%%%%%%%%%%%%%%%%%%%%%%%%%%%%%%%%%%%%%%%%%%%%%%%%%%%%%%%%%%

%%%%%%%%%%%%%%%%%%%%%%%%%%%%%%%%%%%%%%%%%%%%%%%%%%%%%%%%%%%%
\begin{restatable}{theorem}{thmDOLsc}\label{scPOL}
	Let $\fGam$ be a sequent.
	Then 
	$\proves[\pSdol]\fGam$ iff $\proves[\POL]\fGam$.
\end{restatable}
\begin{proof}
	We conclude by \Cref{thm:cutElimDOL}
	since the axiom $\axO$ is derivable in $\Sdol$
	(see \Cref{fig:proofAxO}).
	Note that contrary to what happens in $\PDL$,  
	the Fisher-Ladner closure of a formula in $\POL$ may be infinite.
	However, the partial order over sequents in \Cref{rem:FLanal} used in the proof of \Cref{thm:PDL:SC} is still well-founded.
\end{proof}
%%%%%%%%%%%%%%%%%%%%%%%%%%%%%%%%%%%%%%%%%%%%%%%%%%%%%%%%%%%%%

%%%%%%%%%%%%%%%%%%%%%%%%%%%%%%%%%%%%%%%%%%%%%%%%%%%%%%%%%%%%%
\begin{theorem}\label{thm:DOL:traceEq}
	Let $\pa,\pb\in\progSet$.
	Then 
	$
	\proves[\POL] \!\tbox[\pa]\fA\!\fequiv\!\tbox[\pb]\fA
	\mbox{ iff }
	\pa\treq\pb
%	\mydot
	$.
\end{theorem}
\begin{proof}
	By definition of $\treq$
	we have that
	$\pa\sosto{\ta_1}\pa'$
	iff
	$\pb\sosto{\ta_1}\pb'$.
	We conclude by induction on 
	size length of the (finite) prefixes of traces in 
	$\tracesof{\pa}=\tracesof{\pb}$.
\end{proof}
%%%%%%%%%%%%%%%%%%%%%%%%%%%%%%%%%%%%%%%%%%%%%%%%%%%%%%%%%%%%%

%%%%%%%%%%%%%%%%%%%%%%%%%%%%%%%%%%%%%%%%%%%%%%%%%%%%%%%%%%%%%
\begin{remark}
	As written, the rule $\tbox[\OS]$ introduces (bottom-up) a branching during proof search which corresponds to the branching in the label transition system of the program execution.
	However,
	it would be desirable to refine such a rule
	in order to distinguish 
	the branching due to interleaving concurrency from 
	the branching due to internal choices of the system.
	More precisely, using the terminology from \cite{andreoli1992logic,hemer2002don,liang:hal-03457379}, 
	interleaving concurrency is a form of `don't care' non-determinism, depending on inessential choices introduced by the syntax because of its limitations in handling concurrency,
	while internal choices cause a `don't know' non-determinism, requiring us to take into account all possible evolution of the system in order to overcome this lack of knowledge about the next state of a computation.
	In proof theory, 
	the `don't care' non-determinism is considered inessential in defining a notion of equivalence for proofs, and it is usually captured by simple independent rule permutations (see \Cref{fig:indRules}) 
	while the `don't know' non-determinism is the responsible of having different proofs.
	
	For this purpose, it would suffice to define a notion of \emph{concurrency} between two elements $(\pa,\tb_1,\pc_1)$ and $(\pa,\tb_2,\pc_2)$ in $\OS$
	by requiring the existence of a 
	program $\pc'$ such that 
	$$
		\begin{array}{c@{\qquad}c@{\qquad}c}
			&\vemod1{\pc_1}				\\
			\vemod0{\pa}&&\vemod3{\pc'}	\\
			&\vemod2{\pc_2}
		\end{array}
		\OSledges{emod0/emod1/{\tb_1}}
		\OSledges{emod0/emod2/{\tb_2}}
		\OSledges{emod1/emod3/{\tb_2}}
		\OSledges{emod2/emod3/{\tb_1}}
	\qquad
	\mbox{with}
	\qquad
		(\pc_1,\tb_2,\pc')	,
		(\pc_2,\tb_1,\pc')
		\in\OS
	$$
	and restrict the side condition $\dagger$ of the rules $\tbox[\OS]$ and $\tdia[\OS]$ to sets of pairs such that 
	$(\pa,\tb_1,\pc_1)$ and $(\pa,\tb_2,\pc_2)$ are concurrent for each 
	for each $(\tb_1,\pc_1)$ and $(\tb_2,\pc_2)$ in $\dagger$.
	The adequacy result for the calculus with such a restricted rule is proven by showing, modulo rule permutations, that the general and restricted version of the rules are inter-definable.
\end{remark}
%%%%%%%%%%%%%%%%%%%%%%%%%%%%%%%%%%%%%%%%%%%%%%%%%%%%%%%%%%%%%

%%%%%%%%%%%%%%%%%%%%%%%%%%%%%%%%%%%%%%%%%%%%%%%%%%%%%%%%%%%%%
% FIG independet rule perm
%%%%%%%%%%%%%%%%%%%%%%%%%%%%%%%%%%%%%%%%%%%%%%%%%%%%%%%%%%%%%
\begin{figure}[t]
	$$\begin{array}{c}
	\vlderivation{
		\vlin{\rrule_1}{}{\vdash \fGam, \fDel,\fSig}{
			\vlin{\rrule_2}{}{\vdash \fGam, \fDel', \fSig}{\vlhy{\vdash \fGam,\fDel',\fSig'}}
		}
	}
	\equiv
		\vlderivation{
		\vlin{\rrule_2}{}{\vdash \fGam, \fDel,\fSig}{
			\vlin{\rrule_1}{}{\vdash \fGam, \fDel,\fSig'}{
				\vlhy{\vdash \fGam,\fDel',\fSig'}
			}
		}
	}
	\qquad
	\vlderivation{
		\vliin{\rrule_1}{}{\fGam, \fDel,\fSig}{
			\vlin{\rrule_2}{}{\fGam,\fA}{\vlhy{\fGam', \fA}}
		}{
				\vlhy{\fDel,\fB}
		}
	}
	\equiv
	\vlderivation{
		\vlin{\rrule_2}{}{\fGam, \fDel,\fSig}{
			\vliin{\rrule_1}{}{\fGam', \fDel,\fSig}{
				\vlhy{\fGam', \fA }
			}{
				\vlhy{\fDel,\fB}
			}
		}
	}
	\qquad
	\vlderivation{
		\vliin{\rrule_1}{}{\fGam_1,\fGam_2,\fGam_3, \fDel,\fSig}{
			\vliin{\rrule_2}{}{\fGam_1,\fGam_2,\fDel}{
				\vlhy{\fGam_1,\fA}
			}{
				\vlhy{\fGam_2,\fB}
			}
		}{
			\vlhy{\fGam_3, \fC}
		}
	}
	\equiv
	\vlderivation{
		\vliin{\rrule_2}{}{\fGam_1,\fGam_2,\fGam_3, \fDel,\fSig}{
			\vlhy{\fGam_1,\fA}
		}{
			\vliin{\rrule_1}{}{\fGam_2,\fGam_3,\fSig}{
				\vlhy{\fGam_2,\fB}
			}{
				\vlhy{\fGam_3, \fC}
			}
		}
	}
	\end{array}$$
	\caption{Independent rule permutations.}
	\label{fig:indRules}
\end{figure}
%%%%%%%%%%%%%%%%%%%%%%%%%%%%%%%%%%%%%%%%%%%%%%%%%%%%%%%%%%%%%
%%%%%%%%%%%%%%%%%%%%%%%%%%%%%%%%%%%%%%%%%%%%%%%%%%%%%%%%%%%%%

%%%%%%%%%%%%%%%%%%%%%%%%%%%%%%%%%%%%%%%%%%%%%%%%%%%%%%%%%%%%%
%% FIG additional sequent rules of DOL
%%%%%%%%%%%%%%%%%%%%%%%%%%%%%%%%%%%%%%%%%%%%%%%%%%%%%%%%%%%%%
%\begin{figure}[t]
%	$$\begin{array}{c}
%		\vliiinf{\bosr}{\dagger}{
%			\vdash\fGam,\tbox[\ta]\fA
%		}{
%			\vdash\fGam,\tbox[\tb_1]\tbox[\pc_1]\fA
%		}{\cdots}{
%			\vdash\fGam,\tbox[\tb_n]\tbox[\pc_n]\fA
%		}
%	%
%	\qquad
%	%
%		\vlinf{\dosr}{\dagger}{
%			\vdash\fGam,\tdia[\pa]\fA
%		}{
%			\vdash\fGam,\tdia[\tb_1]\tdia[\pc_1]\fA,\ldots, \tdia[\tb_n]\tdia[\pc_n]\fA
%		}
%	\end{array}$$
%	\caption{
%		Additional sequent calculus rules for $\Sdol$ with side condition $\dagger\coloneqq \set{(\tb_i,\pc_i)\mid i\in\intset1n}=\set{(\tb,\pc)\mid\pa\sosto\tb\pc}$.
%	}
%	\label{fig:seqDOL}
%\end{figure}
%%%%%%%%%%%%%%%%%%%%%%%%%%%%%%%%%%%%%%%%%%%%%%%%%%%%%%%%%%%%%%
%%%%%%%%%%%%%%%%%%%%%%%%%%%%%%%%%%%%%%%%%%%%%%%%%%%%%%%%%%%%%%

%%%%%%%%%%%%%%%%%%%%%%%%%%%%%%%%%%%%%%%%%%%%%%%%%%%%%%%%%%%%%
% FIG cut-elim
%%%%%%%%%%%%%%%%%%%%%%%%%%%%%%%%%%%%%%%%%%%%%%%%%%%%%%%%%%%%%
\begin{figure*}[t]
	\adjustbox{max width=\textwidth}{$\begin{array}{c}
		\vlderivation{
			\vliin{\cutr}{}{\vdash \fGam}{
				\vlin{\tbox[\OS]}{}{
					\vdash\fGam,\tbox[\ta]\nfA
				}{
					\vlhy{
						\vdash\fGam,\tbox[\tb_1]\tbox[\pc_1]\nfA 
						\quad\cdots\quad
						\vdash\fGam,\tbox[\tb_n]\tbox[\pc_n]\nfA 
					}
				}
			}{
				\vlin{\tdia[\OS]}{}{
					\vdash\fGam,\tdia[\pa]\fA
				}{\vlhy{
						\vdash\fGam,\tdia[\tb_1]\tdia[\pc_1]\fA,\ldots, \tdia[\tb_n]\tdia[\pc_n]\fA
				}}
			}
		}
	\\
	\rotatebox{-90}{$\rightsquigarrow$}
	\\\\
		\vlderivation{
			\vliin{\cutr}{}{\vdash \fGam}{
				\vlhy{\vdash\fGam,\tbox[\tb_1]\tbox[\pc_1]\nfA}
			}{
				\vliin{\cutr}{}{
					\vdash\fGam,\tdia[\tb_1]\tdia[\pc_1]\fA
				}{
					\vlhy{\vdash\fGam,\tdia[\tb_2]\tdia[\pc_2]\fA}
				}{
					\vliin{\cutr}{}{
						\quad\vdots\quad
					}{
						\vlhy{\vdash\fGam,\tbox[\tb_n]\tbox[\pc_n]\nfA}
					}{
						\vlhy{\vdash\fGam,\tdia[\tb_1]\tdia[\pc_1]\fA,\ldots,\tdia[\tb_n]\tdia[\pc_n]\fA} 
					}
				}
			}
		}
	\end{array}$}
	\caption{Additional cut-elimination step in $\Sdol$.}
	\label{fig:seqDOLcutElim}
\end{figure*}
%%%%%%%%%%%%%%%%%%%%%%%%%%%%%%%%%%%%%%%%%%%%%%%%%%%%%%%%%%%%
%%%%%%%%%%%%%%%%%%%%%%%%%%%%%%%%%%%%%%%%%%%%%%%%%%%%%%%%%%%%

%%%%%%%%%%%%%%%%%%%%%%%%%%%%%%%%%%%%%%%%%%%%%%%%%%%%%%%%%%%%
%%%%%%%%%%%%%%%%%%%%%%%%%%%%%%%%%%%%%%%%%%%%%%%%%%%%%%%%%%%%
%%%%%%%%%%%%%%%%%%%%%%%%%%%%%%%%%%%%%%%%%%%%%%%%%%%%%%%%%%%%
\section{Concurrency Theory meets \PDL}\label{sec:conc}
%%%%%%%%%%%%%%%%%%%%%%%%%%%%%%%%%%%%%%%%%%%%%%%%%%%%%%%%%%%%
%%%%%%%%%%%%%%%%%%%%%%%%%%%%%%%%%%%%%%%%%%%%%%%%%%%%%%%%%%%%
%%%%%%%%%%%%%%%%%%%%%%%%%%%%%%%%%%%%%%%%%%%%%%%%%%%%%%%%%%%%
\def\OSCCS{{\OS_{\CCS}}}

In this section we provide two case studies of languages for concurrent systems: Milner's \emph{Calculus of Communicating Systems} (\CCS) \cite{M80},
and a theory of \emph{Choreographic Programming} \cite{montesi:book}.
The first provides an archetypal case of concurrency via parallel composition of processes and the second an illustrative example of concurrency via out-of-order execution of non-interfering actions. 

%%%%%%%%%%%%%%%%%%%%%%%%%%%%%%%%%%%%%%%%%%%%%%%%%%%%%%%%%%%%
%%%%%%%%%%%%%%%%%%%%%%%%%%%%%%%%%%%%%%%%%%%%%%%%%%%%%%%%%%%%
\subsection{Concurrency via parallel composition}
%%%%%%%%%%%%%%%%%%%%%%%%%%%%%%%%%%%%%%%%%%%%%%%%%%%%%%%%%%%%
%%%%%%%%%%%%%%%%%%%%%%%%%%%%%%%%%%%%%%%%%%%%%%%%%%%%%%%%%%%%

%%%%%%%%%%%%%%%%%%%%%%%%%%%%%%%%%%%%%%%%%%%%%%%%%%%%%%%%%%%%
% FIG CCS
%%%%%%%%%%%%%%%%%%%%%%%%%%%%%%%%%%%%%%%%%%%%%%%%%%%%%%%%%%%%
\begin{figure}
	\adjustbox{max width=\textwidth}{$\begin{array}{c|c|c}
		\mbox{\defn{Processes}}
		&
		\mbox{\defn{Labels}}
		&
		\mbox{\defn{Reduction rules}}
	\\
		\begin{array}{rll}
			\pP,\pQ  
			\coloneqq & \pnil	&	\mbox{inactive process}
		\\
			\mid&\pl.\pP		&	\mbox{action prefix}
		\\
			\mid&\pP\ppar\pQ	&	\mbox{parallel composition}
		\\
			\mid&\pP + \pQ		&	\mbox{choice}
		\\
			\mid&\pres a \pP	&	\mbox{action restriction}
		\\
			\mid&\pX			&	\mbox{process name}
		\end{array} 
	&
		\begin{array}{rll}
			\lambda \coloneqq& \tat  & \mbox{actions } (\tat \in \Act) 
		\\
			\mid & \ntat & \mbox{co-actions }  (\tat \in \Act)
		\\
			\mid & \ttau & \mbox{silent}
		\end{array} 
	&
		\begin{array}{lrcll}
			\prer	&\pl.\pP		&\sosto{\tl}	& \pP 
			&
			\\
			\parr_1	&\pP\ppar\pQ	&\sosto{\tl}& \pP'\ppar \pQ &\mbox{ if } \pP\sosto\tl\pP'
			\\
			\parr_2	&\pP\ppar\pQ	&\sosto{\tl}& \pP\ppar \pQ' &\mbox{ if } \pQ\sosto\tl\pQ'
			\\
			\comr	&\pP\ppar\pQ	&\sosto{\ttau}	& \pP'\ppar\pQ'	
			&\mbox{ if } \pP\sosto\tat\pP' \mbox{ and } \pQ\sosto\ntat\pQ'
			\\
			\sumr_1	&\pP+\pQ		&\sosto{\tl}& \pP'			&\mbox{ if } \pP\sosto\tl\pP'
			\\
			\sumr_2	&\pP+\pQ		&\sosto{\tl}& \pQ'			&\mbox{ if } \pQ\sosto\tl\pQ'
			\\
			\resr	&\pres a \pP		&\sosto{\tl}& \pres a {\pP'}	
			&\mbox{ if } \pP\sosto\tl\pP' \mbox{ and } \tl\notin\set{\pat,\npat}
			\\
			\recr	&X				&\sosto{\tl}& \pP'
			&\mbox{ if } X \defeq \pP \mbox{ and } \pP\sosto\tl\pP'
		\end{array}	
	\end{array}$}
	\caption{Syntax and operational semantics of $\CCS$.}%, and the structural congruence.}
	\label{fig:CCS}
\end{figure}
%%%%%%%%%%%%%%%%%%%%%%%%%%%%%%%%%%%%%%%%%%%%%%%%%%%%%%%%%%%%
%%%%%%%%%%%%%%%%%%%%%%%%%%%%%%%%%%%%%%%%%%%%%%%%%%%%%%%%%%%%

$\CCS$ is a process calculus where processes interact via synchronisations where two parties perform complementary actions (often thought of as sending and receiving).
Concurrency is achieved via explicit parallel composition of processes 
equipped with interleaving semantics.

\defn{Processes} in \CCS (with recursion) are described by the terms 
generated by the grammar in \Cref{fig:CCS}. 
The definition is parametrised in a countable set $\Act$ of symbols denoting the synchronisation \defn{actions} that processes can perform.
The set is equipped with an involution $(\cneg{\cdot})$ mapping each action its complementary action, or \defn{co-action} for short.
The definition is also parametrised in a set of \defn{process definitions} (objects of the form $\pX \defeq \pP$) which are used to express infinite behaviours via recursion.
The semantics of processes is given as the labelled transition system (or LTS) with processes as states and as transition relation the smallest relation closed under the derivation rules reported in \Cref{fig:CCS}.
Both syntax and semantics are standard and we briefly discuss them below. 

The term $\pnil$ denotes the \defn{terminated process} and has no transitions.
A term $\pl.\pP$ denotes a process ready to perform the action $\pl$ before continuing as $\pP$ as specified by rule $\prer$.
A term $\pP \ppar \pQ$ denotes the \defn{parallel} composition of processes $\pP$ and $\pQ$ which are executed by interleaving (rules $\parr_1$ and $\parr_2$) or synchronising their actions (rule $\comr$).
Rule $\parr_1$ allows $\pP \ppar \pQ$ to perform a transition where $\pP$ performs an action (evolving into $\pP'$) independently from $\pQ$ and symmetrically for rule $\parr_2$.
Rule $\comr$ describes transitions where $\pP$ and $\pQ$ synchronise by performing matching actions.
To model that synchronisations are binary, transitions derived with this rule are given the label $\ttau$ which is separate from $\Act$ and is traditionally used in process algebras to denote steps that do not interact with the context of a process (hence named \defn{silent} or \defn{internal}).
A term $\pP + \pQ$ denotes a \defn{choice} between actions performed by $\pP$ and $\pQ$ where performing an action from one process disregards the other as specified by rules $\sumr_1$ and $\sumr_2$.
A term $\pres{\pat}{\pP}$ denotes a process where synchronisations using the action $\pat$ are \defn{restricted} to remain within $\pP$ (cannot be used to interact with the context but only with other parts of $\pP$) as prescribed by rule $\resr$ which requires $\pl$ to be neither $\pat$ nor $\npat$.
A term $\pX$ denotes the process $\pP$ associated to the \defn{process name} $\pX$ by the process definition $\pX \defeq \pP$ and has the same semantics as $\pP$ (rule $\recr$). 
To ensure that the resulting LTS is finitely branching, we assume, 
as common practice  (see, e.g., \cite{GV92,AFV01}), that process definitions are \emph{guarded} meaning that every process name occurring in the body of a process definition occurs under an action prefix.

We denote by $\OSCCS$ the operational semantics over the set of 
\CCS processes (the set of tests is empty)
with
labels defined as in \Cref{fig:CCS} (by letting $\ttau=\tepsi$)
defined as in \Cref{fig:CCS}.
Then, we obtain as an instance of \Cref{thm:DOL:traceEq}, that logical equivalence in ${\POL[{\OSCCS}]}$ captures trace equivalence in \CCS.

\begin{corollary}
	\label{thm:CCS:traceEq}
	Let $\pP$ and $\pQ$ be process in \CCS.
	Then,
	$$
	\pP\treq\pQ
	\qquad\mbox{ iff }\qquad
	\proves[{\POL[{\OSCCS}]}] \tbox[\pP]\fA\fequiv\tbox[\pQ]\fA
	\mbox{ for any formula $\fA$}
	\mydot
	$$
\end{corollary}

We now illustrate the expressiveness and interpretive power of ${\POL[{\OSCCS}]}$ through a series of representative examples.
These examples demonstrate how ${\POL[{\OSCCS}]}$ accommodates standard program constructs and equational reasoning patterns typical of process calculi, with a particular focus on their operational semantics.

The next example revisits a textbook example of the different discriminating power of bisimilarity and trace equivalence with a recursion twist to illustrate how circular proofs in ${\POL[{\OSCCS}]}$ allows to reason about recursive processes.
\begin{example}
\begin{figure*}[t]
	$$\adjustbox{max width=\textwidth}{$\begin{array}{c}
		\vlderivation{
			\vliin{\land}{}{
				\vdash \tbox[\ppi_2]\fA \fequiv \tbox[\ppi_1]\fA
			}{
				\vlin{\lor}{\star_1}{\vdash \tdia[\ppi_2]\nfA \lor \tbox[\ppi_1]\fA}{
					\vlin{\bosr}{}{\vdash \tdia[\ppi_2]\nfA, \tbox[\ppi_1]\fA}{
						\vlin{\dosr}{}{\vdash \tdia[\ppi_2]\nfA, \tbox[\ta]\tbox[\pb.\ppi_1 + \pc]\fA}{
							\vlin{\Krule}{}{
								\vdash \tdia[\ta]\tdia[\pb.\ppi_2]\nfA, \tdia[\ta]\tdia[\tc]\nfA,\tbox[\ta]\tbox[\pb.\ppi_1 + \pc]\fA
							}{
								\vliin{\bosr}{}{
									\vdash \tdia[\pb.\ppi_2]\nfA,\tdia[\tc]\nfA,\tbox[\pb.\ppi_1 + \pc]\fA
								}{
									\vlin{\Wrule}{}{
										\vdash \tdia[\pb\ppi_2]\nfA,\tdia[\tc]\nfA,\tbox[\tb]\tbox[\ppi_1]\fA
									}{\vlin{\dosr}{}{
											\vdash \tdia[\tb]\tdia[\ppi_2]\nfA,\tbox[\tb]\tbox[\ppi_1]\fA
										}{\vlin{\Krule}{\star_1}{
												\vdash \tdia[\tb]\tdia[\ppi_2]\nfA,\tbox[\tb]\tbox[\ppi_1]\fA
											}{
											\vlhy{}
											}
										}
									}
								}{
									\vlin{\Wrule}{}{
										\vdash \tdia[\pb.\ppi_2]\nfA,\tdia[\tc]\nfA,\tbox[\tc]\fA
									}{\vlin{\AXrule}{}{
											\vdash \tdia[\tc]\nfA,\tbox[\tc]\fA
										}{\vlhy{}}
									}
								}
							}
						}
					}
				}
			}{
				\vlin{\lor}{\star_2}{\vdash \tdia[\ppi_1]\nfA\lor \tbox[\ppi_2]\fA}{
					\vliin{\bosr}{}{\vdash \tdia[\ppi_1]\nfA, \tbox[\ppi_2]\fA}{
						\vlin{\dosr}{}{
							\vdash \tdia[\ppi_1]\nfA, \tbox[\ta]\tbox[\pb.\ppi_2]\fA
						}{
							\vlin{\Wrule}{}{
								\vdash \tdia[\ta]\tdia[\pb.\ppi_1]\nfA, \tdia[\ta]\tdia[\tc]\nfA,\tbox[\ta]\tbox[\pb.\ppi_2]\fA
							}{
								\vlin{\Krule}{}{
									\vdash \tdia[\ta]\tdia[\pb.\ppi_1]\nfA,\tbox[\ta]\tbox[\pb.\ppi_2]\fA
								}{\vliq{\bosr+\dosr}{}{
										\vdash \tdia[\pb.\ppi_1]\nfA,\tbox[\pb.\ppi_2]\fA
									}{\vlin{\Krule}{\star_2}{
										\vdash \tdia[\tb]\tdia[\ppi_1]\nfA,\tbox[\tb]\tbox[\ppi_2]\fA
										}{\vlhy{}}}
								}
							}
						}
					}{
						\vlin{\dosr}{}{\vdash \tdia[\ppi_1]\nfA, \tbox[\ta]\tbox[\tc]\fA}{
							\vlin{\Wrule}{}{
								\vdash \tdia[\ta]\tdia[\pb.\ppi_1]\nfA, \tdia[\ta]\tdia[\tc]\nfA,\tbox[\ta]\tbox[\tc]\fA
							}{
								\vlin{\AXrule}{}{
									\vdash \tdia[\ta]\tdia[\tc]\nfA,\tbox[\ta]\tbox[\tc]\fA
								}{
									\vlhy{}
								}
							}
						}
					}
				}
			}
		}
	\\\\
		\text{where}
	\qquad
		\begin{array}{c@{\qquand}c}
			\begin{array}{c}
				\scriptsize \ppi_1=(\pa.\pb.\ppi_1) + (\pa.\pc)
			\\
				\begin{array}{ccc}
					&\vbul1
					\\[15pt]
					\vbul2 && \vbul4
					\\[20pt]
					&& \vbul5
				\end{array}
				\tikz[overlay,remember picture,draw,fill,opacity=1] \draw[-stealth,draw,thick] (bul1) to [bend right = 20] node[midway,fill=white,inner sep=1pt]{\scriptsize$\ta$}  (bul2);
				\tikz[overlay,remember picture,draw,fill,opacity=1] \draw[-stealth,draw,thick] (bul2) to [bend right = 20] node[midway,fill=white,inner sep=1pt]{\scriptsize$\tb$}  (bul1);
				\tikz[overlay,remember picture,draw,fill,opacity=1] \draw[-stealth,draw,thick] (bul1) to node[midway,fill=white,inner sep=1pt]{\scriptsize$\ta$}  (bul4);
				\tikz[overlay,remember picture,draw,fill,opacity=1] \draw[-stealth,draw,thick] (bul4) to node[midway,fill=white,inner sep=1pt]{\scriptsize$\tc$}  (bul5);
			\end{array}
		&
			\begin{array}{c}
				\scriptsize \ppi_2=\pa.(\pb.\ppi_2 + \pc)
			\\
				\begin{array}{c}
					\vbul1
					\\[20pt]
					\vbul2 
					\\[15pt]
					\vbul5
				\end{array}
				\tikz[overlay,remember picture,draw,fill,opacity=1] \draw[-stealth,draw,thick] (bul1) to [bend right = 20] node[midway,fill=white,inner sep=1pt]{\scriptsize$\ta$}  (bul2);
				\tikz[overlay,remember picture,draw,fill,opacity=1] \draw[-stealth,draw,thick] (bul2) to [bend right = 20] node[midway,fill=white,inner sep=1pt]{\scriptsize$\tb$}  (bul1);
				\tikz[overlay,remember picture,draw,fill,opacity=1] \draw[-stealth,draw,thick] (bul2) to node[midway,fill=white,inner sep=1pt]{\scriptsize$\tc$}  (bul5);
			\end{array}
		\end{array}
	\end{array}$}$$
	\caption{The derivation proving $(\pa.\pb.\ppi_1) + (\pa.\pc)\treq\pa.(\pb.\ppi_2 + \pc)$ in $\OSCCS$.}
	\label{fig:CCSex}
\end{figure*}
Consider the recursive processes $\ppi_1=(\pa.\pb.\ppi_1) + (\pa.\pc)$ and $\ppi_2=\pa.(\pb.\ppi_2 + \pc)$ depicted in \Cref{fig:CCSex} on the left.
These processes are trace equivalent but not bisimilar.
The derivation on the right of the figure proves, by \Cref{thm:CCS:traceEq}, that $\ppi_1$ and $\ppi_2$ are indeed trace equivalent. 
\end{example}

%The next example captures the semantics of processes that have no transition. 
\begin{example}
	If a process $\pP$ has no transitions then reasoning about its behaviour coincides with reasoning about the empty program. Formally, $\tbox[\pP]\fA \fequiv \tbox[\tempt]\fA$ holds whenever $\{ \pQ \mid \pP \sosto{\tl} \pQ \} = \emptyset$ (cf.~the application of rule $\bosr$ in the derivation below).
	\begin{equation*}
	\vlderivation
	{\vliin{\land}{}
		{\vdash \tbox[\pP]\fA \fequiv \tbox[\tempt]\fA}
		{\vlin{\lor}{}
			{\vdash \tdia[\pP]\nfA \lor \tbox[\tempt]\fA}
			{\vlin{\Wrule}{}
				{\vdash \tdia[\pP]\nfA,\tbox[\tempt]\fA}
				{\vlin{\tbox[\tempt]}{}
					{\vdash\tbox[\tempt]\fA}
					{\vlhy{}}}}}
		{\vlin{\lor}{}
			{\vdash \tdia[\tempt]\nfA \lor \tbox[\pP]\fA}
			{\vlin{\Wrule}{}
				{\vdash \tdia[\tempt]\nfA,\tbox[\pP]\fA}
				{\vlin{\bosr}{}
					{\vdash\tbox[\pP]\fA}
					{\vlhy{}}}}}}
	\end{equation*}
\end{example}

\begin{example}
Procedure names are equivalent to their definition i.e., $\tbox[X]\fA \fequiv \tbox[P]\fA$ where $X \defeq \pP$. Below we report the derivation for one direction of the derivation, the other is similar. 
	\begin{gather*}
		% \vliiin{\land}{}{\vdash \bigwedge \cneg\fGam,\fGam}{
		% 					\vlin{\AXrule}{}{\vdash \cneg\fGam_1,\fGam_1}{\vlhy{}}
		% 				}{\vlhy{\cdots}}{
		% 					\vlin{\AXrule}{}{\vdash \cneg\fGam_{\sizeof{\fGam}},\fGam_{\sizeof{\fGam}}}{\vlhy{}}
		% 				}
	\vlderivation{\vlin{\dosr}{}
		{\vdash \tdia[X]\nfA,\tbox[\pP]\fA}
		{\vliiin{\bosr}{}
			{\vdash \tdia[\ta_1]\tdia[\pP_1]\nfA,\dots,\tdia[\ta_n]\tdia[\pP_n]\nfA,\tbox[\pP]\fA}
			{\vliq{\Wrule}{}
				{\vdash \tdia[\ta_1]\tdia[\pP_1]\nfA,\dots,\tdia[\ta_n]\tdia[\pP_n]\nfA,\tbox[\pa_1]\tbox[\pP_1]\fA}
				{\vlin{\AXrule}{}
					{\vdash \tdia[\ta_1]\tdia[\pP_1]\nfA,\tbox[\pa_1]\tbox[\pP_1]\fA}
					{\vlhy{}}}}
			{\vlhy{\cdots}}
			{\vliq{\Wrule}{}
					{\vdash \tdia[\ta_1]\tdia[\pP_1]\nfA,\dots,\tdia[\ta_n]\tdia[\pP_n]\nfA,\tbox[\pa_i]\tbox[\pP_i]\fA}
					{\vlin{\AXrule}{}
						{\vdash \tdia[\ta_n]\tdia[\pP_n]\nfA,\tbox[\pa_n]\tbox[\pP_n]\fA}
						{\vlhy{}}}}}}
					% {\vlhy{\left\{\vlderivation
					% 	{\vliq{\Wrule}{}
					% 		{\vdash \tdia[\ta_1]\tdia[\pP_1]\nfA,\dots,\tdia[\ta_n]\tdia[\pP_n]\nfA,\tbox[\pa_i]\tbox[\pP_i]\fA}
					% 		{\vlin{\AXrule}{}
					% 			{\vdash \tdia[\ta_i]\tdia[\pP_i]\nfA,\tbox[\pa_i]\tbox[\pP_i]\fA}
					% 			{\vlhy{}}}}
					% 	\right\}_{i\in\{1,\dots,n\}}}}}}
	\end{gather*}
\end{example}

\begin{example}
	Nondeterministic choice ($+$) is often replaced with \emph{guarded choice}, a more restrictive form of nondeterminism ($\sum_{i = 0}^{n}\pa_i.\pP_i$) where each branch ($\pP_i$) is guarded by a prefixing action ($\pa_i$), its guard.
	This construct is recovered directly in terms of `$\oplus$' and `$;$' thanks to the equivalence 
	$\tbox[\sum_{i = 0}^{n}\pa_i.\pP_i]\fA \fequiv \tbox[\bigoplus_{i = 0}^{n}\pa_i;\pP_i]\fA$.
	$$
		\vlderivation{
			\vliin{\land}{}
			{\vdash \tbox[\sum_{i = 0}^{n}\pa_i.\pP_i]\fA \fequiv \tbox[\bigoplus_{i = 0}^{n}\pa_i;\pP_i]\fA}
			{\vlin{\lor}{}
				{\vdash \tdia[\sum_{i = 0}^{n}\pa_i.\pP_i]\nfA \lor \tbox[\bigoplus_{i = 0}^{n}\pa_i;\pP_i]\fA}
				{\vlin{\bplusr}{}
					{\vdash \tdia[\sum_{i = 0}^{n}\pa_i.\pP_i]\nfA, \tbox[\bigoplus_{i = 0}^{n}\pa_i;\pP_i]\fA}
					{\vlhy{\left\{\vlderivation{\vlpr{\dD_j}{}{\vdash \tdia[\sum_{i = 0}^{n}\pa_i.\pP_i]\nfA, \tbox[\pa_j;\pP_j]\fA}}\right\}_{j\in\{0,\dots,n\}}}}}}
			{\vlin{\lor}{}
				{\vdash \tdia[\bigoplus_{i = 0}^{n}\pa_i;\pP_i]\nfA \lor \tbox[\sum_{i = 0}^{n}\pa_i.\pP_i]\fA}
				{\vlin{\bosr}{}
					{\vdash \tdia[\bigoplus_{i = 0}^{n}\pa_i;\pP_i]\nfA, \tbox[\sum_{i = 0}^{n}\pa_i.\pP_i]\fA}
					{\vlhy{\left\{\vlderivation{\vlpr{\dD'_j}{}{\vdash \tdia[\bigoplus_{i = 0}^{n}\pa_i;\pP_i]\nfA, \tbox[\pa_j]\tbox[\pP_j]\fA}}\right\}_{j\in\{0,\dots,n\}}}}}}
		}
	$$
	where $\dD_j$ and $\dD'_j$ are as follows:
	$$
		\vlderivation{
			\vlin{\dosr}{}
			{\vdash \tdia[\sum_{i = 0}^{n}\pa_i.\pP_i]\nfA, \tbox[\pa_j;\pP_j]\fA}
			{\vliq{\Wrule}{}
				{\vdash \tdia[\pa_0]\tdia[\pP_0]\nfA,\dots,\tdia[\pa_n]\tdia[\pP_n]\nfA, \tbox[\pa_j;\pP_j]\fA}
				{\vlin{\bseqr}{}
					{\vdash \tdia[\pa_j]\tdia[\pP_j]\nfA,\tbox[\pa_j;\pP_j]\fA}
					{\vlin{\AXrule}{}
						{\vdash \tdia[\pa_j]\tdia[\pP_j]\nfA,\tbox[\pa_j]\tbox[\pP_j]\fA}
						{\vlhy{}}}}}
		}
	\qquad
		\vlderivation{
			\vlin{\dplusr}{}{
				\vdash \tdia[\oplus_{i = 0}^{n}\pa_i;\pP_i]\nfA, \tbox[\pa_j]\tbox[\pP_j]\fA
			}{
				\vliq{\Wrule}{}{
					\vdash \tdia[\pa_0;\pP_0]\nfA,\dots,\tdia[\pa_n;\pP_n]\nfA, \tbox[\pa_j]\tbox[\pP_j]\fA
				}{
					\vlin{\dseqr}{}{
						\vdash \tdia[\pa_j;\pP_j]\nfA,\tbox[\pa_j]\tbox[\pP_j]\fA
					}{\vlin{\AXrule}{}
						{\vdash \tdia[\pa_j]\tdia[\pP_j]\nfA,\tbox[\pa_j]\tbox[\pP_j]\fA}
						{\vlhy{}}
					}
				}
			}
		}
	$$
\end{example}

\begin{example}
	By composing the last two examples, we can establish a correspondence between `+' and `$\oplus$' allowing us to reason about processes with nondeterministic choice using the rules for `$\oplus$': $\tbox[\pP+\pQ]\fA \fequiv \tbox[\pP\oplus\pQ]\fA$. Below, we report the core part of the derivation for one direction of the equivalence, the other is similar. Note that when $\pP$ and $\pQ$ have no transitions $n$ and $m$ are $0$ and the corresponding sequence of $\pP_i$ and $\pQ_i$ is empty.
	\begin{gather*}
	\vlderivation{\vlin{\dosr}{}
		{\vdash \tdia[\pP+\pQ]\nfA,\tbox[\pP\oplus\pQ]\fA}
		{\vliin{\bplusr}{}
			{\vdash \tdia[\ta_1]\tdia[\pP_1]\nfA,\dots,\tdia[\ta_n]\tdia[\pP_n]\nfA,\tdia[\tb_1]\tdia[\pQ_1]\nfA,\dots,\tdia[\tb_m]\tdia[\pQ_m]\nfA,\tbox[\pP\oplus\pQ]\fA}
			{\vliq{\Wrule}{}
				{\vdash \tdia[\ta_1]\tdia[\pP_1]\nfA,\dots,\tdia[\tb_m]\tdia[\pQ_m]\nfA,\tbox[\pP]\fA}
				{\vlin{\bosr}{}
					{\vdash \tdia[\ta_1]\tdia[\pP_1]\nfA,\dots,\tdia[\ta_n]\tdia[\pP_n]\nfA,\tbox[\pP]\fA}
					{\vlhy{\left\{\vlderivation
						{\vliq{\Wrule}{}
							{\vdash \tdia[\ta_1]\tdia[\pP_1]\nfA,\dots,\tdia[\ta_n]\tdia[\pP_n]\nfA,\tbox[\pa_i]\tbox[\pP_i]\fA}
							{\vlin{\AXrule}{}
								{\vdash \tdia[\ta_i]\tdia[\pP_i]\nfA,\tbox[\pa_i]\tbox[\pP_i]\fA}
								{\vlhy{}}}}
						\right\}_{i\in\{1,\dots,n\}}}}}}
			{\vliq{\Wrule}{}
				{\vdash \tdia[\ta_1]\tdia[\pP_1]\nfA,\dots,\tdia[\tb_m]\tdia[\pQ_m]\nfA,\tbox[\pQ]\fA}
				{\vlin{\bosr}{}
					{\vdash \tdia[\tb_1]\tdia[\pQ_1]\nfA,\dots,\tdia[\tb_m]\tdia[\pQ_m]\nfA,\tbox[\pQ]\fA}
					{\vlhy{\left\{\vlderivation
						{\vliq{\Wrule}{}
							{\vdash \tdia[\tb_1]\tdia[\pQ_1]\nfA,\dots,\tdia[\tb_m]\tdia[\pQ_m]\nfA,\tbox[\tb_i]\tbox[\pQ_i]\fA}
							{\vlin{\AXrule}{}
								{\vdash \tdia[\tb_i]\tdia[\pQ_i]\nfA,\tbox[\pb_i]\tbox[\pQ_i]\fA}
								{\vlhy{}}}}
						\right\}_{i\in\{1,\dots,m\}}}}}}}}
	\end{gather*}
Therefore, the Abelian monoid laws available for $\oplus$ and $\tempt$ can be applied to $+$ and $\pnil$ -- as expected.
\end{example}

\begin{example}
	When reasoning in terms of trace equivalence, action prefixes distribute over nondeterministic choices $
	\ta.(\pP+\pQ) \treq \ta.\pP+\ta.\pQ$.
	One can prove this equivalence by invoking \cref{thm:CCS:traceEq} and showing that $\tbox[\ta.(\pP+\pQ)]\fA \fequiv \tbox[\ta.\pP+\ta.\pQ]\fA$ holds.
	This can be done directly with a short derivation or, even more concisely, using the correspondences established by the previous examples to invoke the distribution of `$;$' over  `$\oplus$' i.e.,~$\tbox[\ta;(\pP\oplus\pQ)]\fA \fequiv \tbox[\ta;\pP\oplus\ta;\pQ]\fA$.
\end{example}

\begin{example}
	Classical results for compositional reasoning like the fact that trace inclusion and equivalence are preserved by the constructors of the calculus can be recovered in ${\POL[{\OSCCS}]}$.
	For instance, from $\tbox[\pP]\fA \imp \tbox[\pP']\fA$ we can conclude that $\tbox[\pa.\pP]\fA \imp \tbox[\ta.\pP']\fA$ and $\tbox[\pP+\pQ]\fA \imp \tbox[\pP'+\pQ]\fA$.
	We can show that these hold directly
	\begin{equation*}
		\vlderivation{
			\vlin{\bosr+\dosr}{}
				{\vdash\tdia[\pa.\pP]\nfA,\tbox[\pa.\pP']\fA}
				{\vlin{\Krule}{}
					{\vdash \tdia[\pa]\tdia[\pP]\nfA,\tbox[\pa]\tbox[\pP']\fA}
					{\vlhy{\vdash \tdia[\pP]\nfA,\tbox[\pP']\fA}}}}
	\end{equation*}
	or build on the previous examples and derive the equivalent formulas $\tbox[\pa;\pP]\fA \imp \tbox[\ta;\pP']\fA$ and $\tbox[\pP\oplus\pQ]\fA \imp \tbox[\pP'\oplus\pQ]\fA$.
	\begin{equation*}
		\vlderivation{
			\vlin{\bseqr+\dseqr}{}
				{\vdash\tdia[\pa;\pP]\nfA,\tbox[\pa;\pP']\fA}
				{\vlin{\Krule}{}
					{\vdash \tdia[\pa]\tdia[\pP]\nfA,\tbox[\pa]\tbox[\pP']\fA}
					{\vlhy{\vdash \tdia[\pP]\nfA,\tbox[\pP']\fA}}}}
		\qquad
		\vlderivation{
			\vliin{\bplusr+\dplusr}{}
				{\vdash\tdia[\pP\oplus\pQ]\nfA,\tbox[\pP'\oplus\pQ]\fA}
				{\vlin{\Wrule}{}
					{\vdash \tdia[\pP]\nfA,\tdia[\pQ]\nfA,\tbox[\pP']\fA}
					{\vlhy{\vdash \tdia[\pP]\nfA,\tbox[\pP']\fA}}}
				{\vlin{\Wrule}{}
					{\vdash \tdia[\pP]\nfA,\tdia[\pQ]\nfA,\tbox[\pQ]\fA}
					{\vlin{\AXrule}{}
						{\vdash \tdia[\pQ]\nfA,\tbox[\pQ]\fA}
						{\vlhy{}}}}}
	\end{equation*}
\end{example}

\begin{example}
	The operator for parallel composition of \CCS does not have a direct counterpart in the logic and thus one has to rely on its operational semantics.
	$$
	\tbox[\pP\ppar\pQ]\fA \fequiv \left(
		\bigwedge_{(\pP\ppar\pQ)\sosto\tl\pR}
		\tbox[\tl]\tbox[\pR]\fA \right)
	$$
	Nonetheless, we can derive a number of useful implications to reason about it and recover known laws of trace inclusion and equivalence on concrete instances.
	For example, the derivations below show trace inclusion when action prefixes are distributed over parallel composition ($\tbox[(\pa.\pP)\ppar\pQ]\fA \imp \tbox[\pa.(\pP\ppar\pQ)]\fA$) and trace equivalence for synchronising actions under a restriction ($\tbox[\pres{\tat}{(\tat.\pP\ppar\ntat.\pQ)}]\fA \fequiv \tbox[\pres{\tat}{(\pP\ppar\pQ)}]\fA$).
	\begin{equation*}
	\vlderivation{
		\vlin{\bosr+\dosr+\Wrule}{}{
			\vdash \tdia[\pa.\pP\ppar\pQ]\nfA,\tbox[\pa.(\pP\ppar\pQ)]\fA
		}{
			% \vliq{\Wrule}{}{
			% 	\vdash \tdia[\pa]\tdia[\pP\ppar\pQ]\nfA,\dots,\tbox[\pa]\tbox[\pP\ppar\pQ]\fA
			% }{
				\vlin{\AXrule}{}{
					\vdash \tdia[\pa]\tdia[\pP\ppar\pQ]\nfA,\tbox[\pa]\tbox[\pP\ppar\pQ]\fA
				}{
					\vlhy{}
				}
			% }
		}
	}
	\quad
	\vlderivation{
		\vlin{\bosr}{}
			{\vdash \tdia[\pres{\tat}{(\pP\ppar\pQ)}]\nfA, \tbox[\pres{\tat}{(\tat.\pP\ppar\ntat.\pQ)}]\fA}
			{\vlin{\tbox[\tepsi]}{}
				{\vdash \tdia[\pres{\tat}{(\pP\ppar\pQ)}]\nfA, \tbox[\tepsi]\tbox[\pres{\tat}{(\pP\ppar\pQ)}]\fA}
				{\vlin{\AXrule}{}
				{\vdash \tdia[\pres{\tat}{(\pP\ppar\pQ)}]\nfA, \tbox[\pres{\tat}{(\pP\ppar\pQ)}]\fA}
					{\vlhy{}}}}}
	\quad
	\vlderivation{
		\vlin{\dosr}{}
			{\vdash \tdia[\pres{\tat}{(\tat.\pP\ppar\ntat.\pQ)}]\nfA, \tbox[\pres{\tat}{(\pP\ppar\pQ)}]\fA}
			{\vlin{\tdia[\tepsi]}{}
				{\vdash \tdia[\tepsi]\tdia[\pres{\tat}{(\pP\ppar\pQ)}]\nfA, \tbox[\pres{\tat}{(\pP\ppar\pQ)}]\fA}
				{\vlin{\AXrule}{}
				{\vdash \tdia[\pres{\tat}{(\pP\ppar\pQ)}]\nfA, \tbox[\pres{\tat}{(\pP\ppar\pQ)}]\fA}
					{\vlhy{}}}}}
	% \vlderivation{
	% 	\vlin{\dosr}{}
	% 		{\vdash \tdia[\pres{\tat}{(\tat.\pP\ppar\ntat.\pQ)}]\nfA, \tbox[\pres{\tat}{(\pP\ppar\pQ)}]\fA}
	% 		{\vlin{\tdox[\tepsi]}{}
	% 			{\vdash \tdia[\tepsi]\tdia[\pres{\tat}{(\pP\ppar\pQ)}]\nfA,\tbox[\pres{\tat}{(\pP\ppar\pQ)}]\fA}
	% 			{\vlin{\AXrule}{}
	% 				{\vdash \tdia[\pres{\tat}{(\pP\ppar\pQ)}]\nfA, \tbox[\pres{\tat}{(\pP\ppar\pQ)}]\fA}
	% 				{\vlhy{}}}}}
	\end{equation*}
%\end{example} 
% \begin{example}
Proving these implications did not require a deep analysis of the processes and could be fully derived invoking the operational semantics for one step only. 
This means, that implications like these can be added to the system and be derivable rules. 
On the other hand, an implication like $\tbox[\pP\ppar\pQ]\fA \imp \tbox[\pP]\fA$ requires deeper inspection. 
The procedure for building a proof is similar to how one would check trace inclusion: invoke the semantics of $\pP$ following a breadth-first strategy, mimic the step in $\pP\ppar\pQ$ without running $\pQ$. 
\begin{equation*}\adjustbox{max width=\textwidth}{$
	\vlderivation{
		% \vlin{\bosr+\dosr}{}{
		% 	\vdash \tdia[\pP\ppar\pQ]\nfA, \tbox[\pP]\fA
		% }{
			\vliiin{\bosr+\dosr}{}{
				\vdash \tdia[\pP\ppar\pQ]\nfA, \tbox[\pP]\fA
			}{
				\vliq{\Wrule}{}{
					\vdash \tdia[\ta_1]\tdia[\pP_1\ppar\pQ]\nfA,\dots,\tdia[\ta_n]\tdia[\pP_n\ppar\pQ]\nfA, \tbox[\ta_1]\tbox[\pP_1]\fA
				}{
					\vlin{\Krulep{\ta_1}}{}{
						\vdash \tdia[\ta_1]\tdia[\pP_1\ppar\pQ]\nfA, \tbox[\ta_1]\tbox[\pP_1]\fA
					}{\vlpr{\dD_1}{}{
						\vdash \tdia[\pP_1\ppar\pQ]\nfA, \tbox[\pP_1]\fA
					}}
				}
			}{\vlhy{\qquad\cdots\qquad}}{
				\vliq{\Wrule}{}{
					\vdash \tdia[\ta_1]\tdia[\pP_1\ppar\pQ]\nfA,\dots,\tdia[\ta_n]\tdia[\pP_n\ppar\pQ]\nfA, \tbox[\ta_n]\tbox[\pP_n]\fA
				}{
					\vlin{\Krulep{\ta_n}}{}{
						\vdash \tdia[\ta_n]\tdia[\pP_n\ppar\pQ]\nfA, \tbox[\ta_n]\tbox[\pP_n]\fA
					}{
						\vlpr{\dD_n}{}{
							\vdash \tdia[\pP_n\ppar\pQ]\nfA, \tbox[\pP_n]\fA
						}
					}
				}
			}
		% }
	}$}
%  D_1                         D_n
% -----------------------     ---------------------
% ⟨a₁⟩⟨P₁|Q⟩¬ϕ,[a₁][P₁]ϕ      ⟨aₙ⟩⟨Pₙ|Q⟩¬ϕ,[aₙ][Pₙ]ϕ
% ========================W   ======================= W
% ⟨a₁⟩⟨P₁|Q⟩¬ϕ,…,[a₁][P₁]ϕ  … ⟨aₙ⟩⟨Pₙ|Q⟩¬ϕ,…,[aₙ][Pₙ]ϕ
% ---------------------------------------------------[𝒪]+⟨𝒪⟩
% ⟨P|Q⟩¬ϕ,[P]ϕ
% \vlderivation{
% 	\vliiin{\bosr+\dosr}{}
% 		{\vdash\tdia[\pP\ppar\pQ]\nfA,\tbox[\pP]\fA}
% 		{\vliq{\Wrule}{}
% 			{\vdash \tdia[\tat_1]\tdia[\pP_1|\pQ]\nfA,\dots,\tdia[\tat_n]\tdia[\pP_n|\pQ]\nfA,\tbox[\tat_1]\tbox[\pP_1]\fA}
% 			{\vlin{\Krule}{}
% 				{\vdash \tdia[\tat_1]\tdia[\pP_1|\pQ]\nfA,\tbox[\tat_1]\tbox[\pP_1]\fA}
% 				{\vlhy{\vdash \tdia[\pP_1|\pQ]\nfA,\tbox[\pP_1]\fA}}}}
% 		{\vlhy{\dots}}
% 		{\vliq{\Wrule}{}
% 			{\vdash \tdia[\tat_1]\tdia[\pP_1|\pQ]\nfA,\dots,\tdia[\tat_n]\tdia[\pP_n|\pQ]\nfA,\tbox[\tat_n]\tbox[\pP_n]\fA}
% 			{\vlin{\Krule}{}
% 				{\vdash \tdia[\tat_n]\tdia[\pP_n|\pQ]\nfA,\tbox[\tat_n]\tbox[\pP_n]\fA}
% 				{\vlhy{\vdash \tdia[\pP_n|\pQ]\nfA,\tbox[\pP_n]\fA}}}}}
\end{equation*}
Although the resulting proof may have infinite depth, it is progressive thanks to guardedness of recursive process definitions and image-finiteness of the semantics.
The same approach can be used to verify that the parallel operator ($\ppar$) is commutative, it is associative, and it has the inactive process ($\pnil$) as a unit, that is, the following are tautologies: 
$$
\tbox[\pP\ppar\pQ]\fA \fequiv \tbox[\pQ\ppar\pP]\fA
\qquad,\qquad
\tbox[(\pP\ppar\pQ)\ppar\pR]\fA \fequiv \tbox[\pP\ppar(\pQ\ppar\pR)]\fA
\qquad,\qquad
\tbox[\pnil\ppar\pP]\fA \fequiv \tbox[\pP]\fA
\mydot
$$
\end{example}

\let\pseq\fatsemi
\def\pseqenc#1#2{ \xi(#1, {#2}) }

\CCS lacks a general operator for composing two processes $\pP$ and $\pQ$ in sequence (a counterpart to `$;$') but it can be implemented as a parallel composition once we inject synchronisations between every exit point ($\pnil$) of $\pP$ and $\pQ$ over a dedicated collection of auxiliary channels ($\tat_1,\dots,\tat_n$).
We denote this encoded sequential composition for \CCS processes as `$\pseq$' to distinguish it from `$;$' of \OPDL.
The operator is given on any $\pP$ and $\pQ$ as follows:
\begin{alignat*}{2}
	\pP\pseq\pQ =  {} & \pres{\tat_1,\dots,\tat_n}{(\pseqenc{\pP}{\tat_1,\dots,\tat_n}\ppar \ntat_1.\cdots\ntat_n.\pQ)} &\qquad& \text{if } \tat_1,\dots,\tat_n \notin \pP\ppar \pQ\\
	\pseqenc{\pnil}{A} = {} & \tat_1.\cdots\tat_n.\pnil && \text{if } A = \tat_1,\dots,\tat_n\\
	\pseqenc{\tat.\pP}{A} = {} & \tat.\pseqenc{\pP}{A} && \text{if } \tat \notin A\\
	\pseqenc{\pP_1\ppar\pP_2}{A} = {} & \pseqenc{\pP_1}{A_1} \ppar \pseqenc{\pP_2}{A_2} && \text{if } A = A_1 \uplus A_2 \\
	\pseqenc{\pP_1+\pP_2}{A} = {} & \pseqenc{\pP_1}{A} + \pseqenc{\pP_2}{A} && \\
	\pseqenc{\pres a \pP}{A} = {} & \pres{a}{\pseqenc{\pP}{A}} && \text{if } \tat \notin A\\
	\pseqenc{\pX}{A} = {} & \pX_{A} && \text{if } X \defeq \pP \text{ and } \pX_A \defeq \pseqenc{\pP}{A}
\end{alignat*}
In the case for $\pnil$, the encoding must synchronise on all channels in the sequence $A$ in the specified order.
In the case for $\pP_1+\pP_2$, the two branches may have different numbers of exit points (e.g., $b.\pnil + (b.\pnil \ppar c.\pnil)$ ) but since either may be selected, the encoding needs to inject synchronisations on the exact same sequence of channels $A$ (hence the need to for multiple channels in the case of $\pnil$).
In the case for $\pP_1\ppar\pP_2$, the overall exit points are the combination of the exit points of $\pP_1$ and $\pP_2$ and thus the sequence of channels $A$ is divided into sequences $A_1$ and $A_2$ whose relative ordering is consistent with $A$ (e.g., $A = a,b,c$ may be split into $A_1 = a,c$ and $A_2 = b$ but not in $A_1 = c,a$ and $A_2 = b$) which we denote as $A = A_1 \uplus A_2$. 
In the case of $X$, the process name is replaced with one where the original body of $X$ is instrumented with the necessary synchronisations at its exit points.
Finally, note that as long as $\pseqenc{P}{\tat_1,\dots,\tat_n}$ is defined, the choice of the specific sequence $\tat_1,\dots,\tat_n$ is immaterial since these channels are restricted and do not occur in $P$ or $Q$.

% \begin{example}
% 	There is a correspondence between `$\pseq$' and `$;$': for any $\pP$, $\pQ$, and $\phi$, it holds that $\tbox[\pP;\pQ]\fA \fequiv \tbox[P \pseq Q]\fA$.
% 	\begin{equation*}

% 	\end{equation*}
% \end{example}

\begin{example}
	Having established a counterpart for `$;$', we can state the exchange law of sequential and parallel composition characteristic of Concurrent Kleene Algebras for CCS and use OPDL to prove it sound. 
	Specifically, the implication $\tbox[(\pP_1 \pseq \pQ_1) \ppar (\pP_2 \pseq \pQ_2)]\fA \imp \tbox[(\pP_1 \ppar \pP_2) ; (\pQ_1 \ppar \pQ_2)]\fA $ holds for any  $\pP_1$, $\pP_2$, $\pQ_1$ $\pQ_2$, and $\fA$.
	We need two key observations about $\pP \pseq \pQ$ from the semantics: first,
	$\pP\sosto\tl\pP'$ implies $\pP\pseq\pQ\sosto\tl\pP'\pseq\pQ$ and second, $\pP$ must execute completely before $\pQ$ may take a single step.
	Using this observation, the procedure for generating a (productive) proof is similar to ones used in the examples above. 
	For convenience of exposition, assume that neither $\pP_1$ nor $\pP_2$ is equivalent to $\pnil$.
	\begin{equation*}
	\vlderivation{\vlin{\bseqr}{}{
		\vdash \tdia[(\pP_1 \pseq \pQ_1) \ppar (\pP_2 \pseq \pQ_2)]\nfA,\tbox[(\pP_1 \ppar \pP_2) ; (\pQ_1 \ppar \pQ_2)]\fA
	}{
		\vlin{\bosr}{}{
			\vdash \tdia[(\pP_1 \pseq \pQ_1) \ppar (\pP_2 \pseq \pQ_2)]\nfA,\tbox[(\pP_1 \ppar \pP_2)][(\pQ_1 \ppar \pQ_2)]\fA
		}{
			\vlhy{\left\{
				\vlderivation{
					\vlin{\dosr+((n-1)\times \Wrule)}{}{
						\vdash \tdia[(\pP_1 \pseq \pQ_1) \ppar (\pP_2 \pseq \pQ_2)]\nfA,\tbox[\tat_i][(\pP_1^i \ppar \pP_2^i)][(\pQ_1 \ppar \pQ_2)]\fA
					}{
							\vlin{\Krule}{}{
								\vdash \tdia[\tat_i]\tdia[(\pP_1^i \pseq \pQ_1) \ppar (\pP_2^i \pseq \pQ_2)]\nfA,\tbox[\tat_i][(\pP_1^i \ppar \pP_2^i)][(\pQ_1 \ppar \pQ_2)]\fA
							}{
								\vlhy{{\vdash \tdia[(\pP_1^i \pseq \pQ_1) \ppar (\pP_2^i \pseq \pQ_2)]\nfA,[(\pP_1^i \ppar \pP_2^i)][(\pQ_1 \ppar \pQ_2)]\fA}}
							}
						% }
					}
				}
			\right\}_{i\in I}
			}
		}
	}}
	\end{equation*}
	where $I=\set{i \mid (\pP_1 \ppar \pP_2) \sosto{\tat_i} (\pP_{1}^i \ppar \pP_2^i)}$.
	% 	\begin{equation*}
	% \vlderivation{\vlin{\bosr}{}{
	% 	\vdash \tdia[(\pP_1 \pseq \pQ_1) \ppar (\pP_2 \pseq \pQ_2)]\nfA,\tbox[(\pP_1 \ppar \pP_2) \pseq (\pQ_1 \ppar \pQ_2)]\fA
	% }{
	% 	% \vlin{\bosr}{}{
	% 	% 	\vdash \tdia[(\pP_1 \pseq \pQ_1) \ppar (\pP_2 \pseq \pQ_2)]\nfA,\tbox[(\pP_1 \ppar \pP_2)][(\pQ_1 \ppar \pQ_2)]\fA
	% 	% }{
	% 		\vlhy{\left\{
	% 			\vlderivation{
	% 				\vlin{\dosr+((n-1)\times \Wrule)}{}{
	% 					\vdash \tdia[(\pP_1 \pseq \pQ_1) \ppar (\pP_2 \pseq \pQ_2)]\nfA,\tbox[\tat_i][(\pP_1^i \ppar \pP_2^i) \pseq (\pQ_1 \ppar \pQ_2)]\fA
	% 				}{
	% 					% \vliq{\Wrule}{}{
	% 					% 	\vdash \tdia[\tat_1]\tdia[(\pP_{1,1} \pseq \pQ_1) \ppar (\pP_2 \pseq \pQ_2)]\nfA,\dots,\tdia[\tat_n]\tdia[(\pP_{1,n} \pseq \pQ_1) \ppar (\pP_2 \pseq \pQ_2)]\nfA,\tbox[\tat_i][(\pP_{1,i} \ppar \pP_2)][(\pQ_1 \ppar \pQ_2)]\fA
	% 					% }{
	% 						\vlin{\Krule}{}{
	% 							\vdash \tdia[\tat_i]\tdia[(\pP_1^i \pseq \pQ_1) \ppar (\pP_2^i \pseq \pQ_2)]\nfA,\tbox[\tat_i][(\pP_1^i \ppar \pP_2^i)\pseq (\pQ_1 \ppar \pQ_2)]\fA
	% 						}{
	% 							\vlhy{{\vdash \tdia[(\pP_1^i \pseq \pQ_1) \ppar (\pP_2^i \pseq \pQ_2)]\nfA,\tbox[\tat_i][(\pP_1^i \ppar \pP_2^i)\pseq (\pQ_1 \ppar \pQ_2)]\fA}}
	% 						}
	% 					% }
	% 				}
	% 			}
	% 		\right\}_{i\in I}
	% 		}
	% 	% }
	% }}
	% \end{equation*}
	% where $I=\set{i \mid (\pP_1 \ppar \pP_2) \pseq (\pQ_1 \ppar \pQ_2) \sosto{\tat_i} (\pP_{1}^i \ppar \pP_2^i)\pseq(\pQ_1  \ppar \pQ_2)}$, 
	% that is, 
	% $\pP_2^i=\pP_2$ if $\tat_i=\parr_2$ \resp{$\pP_1^i=\pP_1$ if $\tat_i=\parr_1$},
	% while
	% $\pP_1^i\neq \pP_1$ and $\pP_2^i\neq \pP_2$ if $\tat_i=\tau$.
\end{example}

Although we considered a version of \CCS where infinite behaviours are achieved via recursion, instantiating our results to replication ($\CCS^!$) and iteration ($\CCS^\kstarsymb$) is straightforward.
In particular, the latter corresponds to the settings considered in \cite{PDL:interleaving,BENEVIDES201723}, as discussed in \Cref{ex:DOL:CCSstar}, which sits at the bottom of the expressiveness hierarchy formed by these three approaches \cite{busi:gab:zav:RRI}.
\Cref{thm:CCS:traceEq} subsumes results from \cite{PDL:interleaving,BENEVIDES201723} stating that structural congruence ($\peq$) and strong bisimilarity ($\sim$) are sound w.r.t.~logical equivalence.
Moreover, our treatment of \CCS is standard: parallel composition is a primitive of the calculus whereas in \cite{BENEVIDES201723} it is encoded using choices between sequential programs, an approach that is limited to $\CCS^\kstarsymb$ and results in exponentially larger formulas.

%%%%%%%%%%%%%%%%%%%%%%%%%%%%%%%%%%%%%%%%%%%%%%%%%%%%%%%%%%%%
%%%%%%%%%%%%%%%%%%%%%%%%%%%%%%%%%%%%%%%%%%%%%%%%%%%%%%%%%%%%
\subsection{Concurrency via out-of-order execution}
\label{sec:chor}
%%%%%%%%%%%%%%%%%%%%%%%%%%%%%%%%%%%%%%%%%%%%%%%%%%%%%%%%%%%%
%%%%%%%%%%%%%%%%%%%%%%%%%%%%%%%%%%%%%%%%%%%%%%%%%%%%%%%%%%%%

%%%%%%%%%%%%%%%%%%%%%%%%%%%%%%%%%%%%%%%%%%%%%%%%%%%%%%%%%%%%
% FIG CHOR
%%%%%%%%%%%%%%%%%%%%%%%%%%%%%%%%%%%%%%%%%%%%%%%%%%%%%%%%%%%%
\begin{figure}
	\adjustbox{max width=\textwidth}{$\begin{array}{c|c}
		\mbox{\defn{choreographies}}
		&
		\mbox{\defn{instructions}}
	\\
		\begin{array}{r@{\;}l@{\;}l|l}
			\cC  \coloneqq &
			\cnil	&	\mbox{inactive choreography}	
			&
			\pn(\cC)=\emptyset								
			\\
		\mid&
			\iI ; \cC	&	\mbox{sequential composition}	
			&
			\pn(\cC)=\pn(\iI) \cup \pn(\cC)						
			\\
		\mid&
			\ccond pb {\cC_1}{\cC_2} & \mbox{conditional}	
			& 
			\pn(\cC)=\set{\pp} \cup \pn(\cC_1) \cup \pn(\cC_2) 	
			\\
		\mid&
			\gencall	&	\mbox{call}	
			& 
			\pn(\cC) \text{ where } \cX\defeq\cC
	\end{array}
	&
	\begin{array}{r@{\;}l@{\;}l|c}
	    \iI  \coloneqq &\tgenassign&	\mbox{local assignment}			& 
	 		\pn(\iI)=\set{\pp}								
	 		\\
	 		\mid&\tgencom &	\mbox{communication}			& 
	 		\pn(\iI)=\set{\pp,\pq}							
	 		\\
	 		\mid&\tgensel	&	\mbox{selection}				& 
	 		\pn(\iI)=\set{\pp,\pq}
	    \\
	    \mid&\cont X{\pp}			&	\mbox{(call continuation, runtime)}				& 
	      	\pn(\iI)=\set{\pp}
	 		\\
	 		\mid&\gentest	&	\mbox{test } (\atomTests)	& 
	 		\pn(\iI)=\set{\pp}
			\\
	 		\mid&\ngentest &	\mbox{(negative) test} & 
	 		\pn(\iI)=\set{\pp}
		\end{array} 
	\end{array}
	$}
	\caption{Syntax of choreographies.}
	\label{fig:chor:synt}
\end{figure}

\mathchardef\mhyphen="2D
\def\sloc{\Sigma\mhyphen{\rname{Asg}}}
\def\scom{\Sigma\mhyphen{\rname{Com}}}
\def\stest{\Sigma\mhyphen{\rname{PosTest}}}
\def\sntest{\Sigma\mhyphen{\rname{NegTest}}}
\def\snoop{\Sigma\mhyphen{\rname{SelCall}}}
\def\extrule{\rname{SC}}
\begin{figure*}
	\adjustbox{max width=\textwidth}{$\begin{array}{lrclll}
 		\ratm		&\iI					&\sosto{\tI}	& \pepsi	&
%        \text{if $\iI$ is not of the form $\cont\cX{\pp_0,\dots,\pp_{n+1}}$}
		\\
		\rthen	&\ccond pb {\cC_1}{\cC_2}	&\sosto{\gentest}& \cC_1	&
		\\
		\relse	&\ccond pb {\cC_1}{\cC_2}	&\sosto{\ngentest}& \cC_2	&
		\\
		\rcall	&\pX 							&\sosto{\tgencont} &\cont\cX{\pp_1};\dots;\cont\cX{\pp_n};\cC &
		\mbox{if } \pX\defeq\cC \mbox{ and } \pn(\cX)=\set{\pq,\pp_1,\dots,\pp_n}
		\\
	    \ri	&\iI ; \cC						&\sosto{\tm}	& \cC	&	
    		\mbox{if } \iI\sosto{\tm}\pepsi
		\\
		\rdeli	&\iI ; \cC						&\sosto{\tm}	& \iI;\cC'	&	
		\mbox{if } \cC\sosto{\tm}\cC' \mbox{ and } \pn(\iI) \cap \pn(\tm) = \emptyset
		\\
		\rdelc	&\ccond pb {\cC_1}{\cC_2};\cC	&\sosto{\tm}	& \ccond pb {\cC'_1}{\cC' _2};\cC& 
		\mbox{if } \cC_i\sosto{\tm}\cC'_i \mbox{ and } \pp \notin \pn(\tm) 
	% \\\hline
	% 	\multirow{3}{*}{$\OS[\Sigma]$}&
	% 	\sloc	&	\Sigma	&	\sosto{\tm}	&	\Sigma\csupd{\pp}{x}{v}		&	\mbox{if }	\eval e\Sigma pv \text{ and } \tm = \genassign
	% 	\\&
	% 	\scom	&	\Sigma	&	\sosto{\tm}	&	\Sigma\csupd{\pq}{x}{v}		&	\mbox{if }	\eval e\Sigma pv \text{ and } \tm = \gencom
    % \\&
 	% 	\stest	&	\Sigma	&	\sosto{\gentest}	&	\Sigma						&	\mbox{if }	\eval e\Sigma p v \text{ and } v = \mathsf{true}
    % \\&
 	% 	\sntest	&	\Sigma	&	\sosto{\ngentest}	&	\Sigma						&	\mbox{if }	\eval e\Sigma p v \text{ and } v \neq \mathsf{true}
	% 	\\&
	% 	\snoop	&	\Sigma	&	\sosto{\tm}	&	\Sigma						&
    % \text{if } \tm = \gensel \text{ or } \tm = \gencont
	% \\\hline
	% 	\OS[\tuple{\cC,\Sigma}]
	% 	&
	% 	\extrule
	% 	&
	% 	\tuple{\cC, \Sigma}	& \sosto{\tm}	& \tuple{\cC', \Sigma'}
	% 	&
	% 	\mbox{if } \cC\sosto{\tm}\cC'\quand \Sigma \sosto{\tm} \Sigma'
	\end{array}
	$}
	\caption{Operational semantics of symbolic choreographies $\OS[\cC]$.
	% , for memory storage $\OS[\Sigma]$, and for stateful choreographies $\OS[\tuple{\cC,\Sigma}]$.
	}
	\label{fig:chor:sem}
\end{figure*}
%%%%%%%%%%%%%%%%%%%%%%%%%%%%%%%%%%%%%%%%%%%%%%%%%%%%%%%%%%%%
%%%%%%%%%%%%%%%%%%%%%%%%%%%%%%%%%%%%%%%%%%%%%%%%%%%%%%%%%%%%

Choreographies, in general, are coordination plans that define the expected collective behaviour of concurrent and distributed systems \cite{wscdl,bpmn,montesi:book}.
In the programming paradigm of Choreographic Programming, choreographies are programs that describe the interaction and local computation of processes participating in the system and that can be compiled to executable implementations for each participant (a procedure called endpoint projection) \cite{montesi:book}.
The standard way of supporting concurrency in choreographic programming is to execute independent instructions out of order w.r.t. their syntactic position in the program.
This is an example of a technique found in many programming languages, compilers, and CPUs, to parallelise the execution of code written as sequential.

We consider a reference theory of choreographic programming from \cite{montesi:book} that includes out-of-order execution, recursion, and local computations. 
For conciseness, we do not include the details of the memory model in this presentation and instead treat local computations symbolically (hence the name `\emph{symbolic choreographies}').
To recover the full specification of the language from \emph{op.~cit.}, one can define an operational semantics for the memory model and then `synchronise' it with the one for symbolic choreographies similarly to \cite{GMG18}. 
Except for this presentational difference, our definitions are essentially as for the tail-recursive language given in \cite{montesi:book}.

Choreographic programs (or just choreographies for short) are described by the terms in \Cref{fig:chor:synt}.
Their semantics is given as the LTS induced by the derivation rules in \Cref{fig:chor:sem}.
Both definitions are parametrised in a shared language for expressions that are evaluated by processes locally (i.e., without accessing the state of other processes) and which are used to model local computation.
% We write $\eval e\Sigma pv$ to denote that the expression $e$ evaluates to the value $v$ given the assignment $\Sigma(\pp)$ for the variables local to process $\pp$.
Both definitions are also parametrised in a shared set of \defn{choreography definitions} (objects of the form $\pX\defeq\cC$ where $\pn(C)$ is not empty) which are used to express infinite behaviours via recursion. 

Instructions are performed atomically and describe interactions among processes and, when included in the mode, with the memory.
% $\Sigma$ (via rule $\extrule$ that interfaces choreographies and memory).
An instruction $\gensel$ describes the communication of a constant value $\albl$ used to communicate a local \defn{selection} from process $\pp$ to process $\pq$. 
%(without requiring any interaction with the memory, cf. rule $\snoop$).
An instruction $\gencom$ describes the \defn{communication} of a value computed by $\pp$ evaluating the expression $e$ to $\pq$ which stores it into its local variable $x$.
% (cf., rule $\scom$).
An instruction $\genassign$ represents the \defn{local assignment} at $\pp$.
% (cf., rule $\sloc$).
An instruction $\gentest$ denotes a \defn{test} where $\pp$ evaluates the condition $b$ proceeding only if successful.
Likewise, $\ngentest$ represents a negative test.
Test instructions are not part of the language \cite{montesi:book}; we decided to include them to illustrate the use of tests in $\POL$.

The term $\cnil$ denotes the \defn{inactive choreography} and has no transitions.
A term $\iI;\cC$ denotes the \defn{sequential composition} of the instruction $\iI$ (discussed below) and choreography $\cC$.
The resulting choreography can execute $\iI$ before continuing as $\cC$ via rule $\ratm$ (similarly action prefixes and rule $\prer$ in \CCS) or delay $\iI$ by executing a transition of $\cC$ that does not involve any of the processes occurring in $\iI$ via rule $\rdeli$.
(Labels are instructions and thus carry all the information required to determine, using the function $\pn$, which processes are involved in a transition.)
This relaxed semantics for sequential composition is an instance of out-of-order execution of instructions and introduces concurrency in the model by allowing programs to interleave the execution of instructions at distinct processes (while instructions within the same process remain sequential).
A term $\gencond$ denotes a \defn{conditional} where either $\cC_1$ or $\cC_2$ is chosen depending on whether the test $b$ performed by process $\pp$ is successful (rule $\rthen$) or not (rule $\relse$).
A term $\cX$ denotes a recursive \defn{call} to the choreography definition $\cX \defeq \cC$. 
Its semantics is rather more involved than recursive process calls in \CCS because in recursive choreography calls  multiple processes can join the call concurrently without coordination (this is to capture the decentralised nature of the underlying process model).
The standard device used to achieve this behaviour is the (\defn{runtime}) instruction $\cont X{\pp}$, a syntactic gadget introduced by the first unfolding of a call and used to track processes have yet to join (and prevent erroneous applications of $\rdeli$). 
The resulting semantics is finitely branching, we assume each choreography definition involves finitely many processes (i.e., $\pn(X)$ is finite for any $\cX \defeq \cC$).

The operational semantics $\OS[\cC]$ for this theory of choreographic programming 
%  (abstracting from memory configurations%
% \footnote{
% 	Works on choreographic programming with memory updates usually consider the semantics of choreographies equipped with memory configurations (as in $\OS[\tuple{\cC,\Sigma}]$ \Cref{fig:chor:sem}).  However, the separation adopted in this presentation does not limit the precision of results expected from a theory of choreographies e.g., the correctness of EndPoint Projection: one only needs to ensure that labels used to interface programs and memory are used coherently by the target language.  
% 	In other words, a presentation like ours treats memory configuration as part of the context of the computation whether programs are expressed as choreographies or their projection.
% }%
% )
has choreographies as programs, 
the instructions of the form $\gentest$ or $\ngentest$ as tests, 
and the set of instructions. 
By \Cref{thm:DOL:traceEq}, logical equivalence in ${\POL[{\OS[\cC]}]}$ captures trace equivalence for choreographies.

\begin{corollary}
	\label{thm:chor:traceEq}
$\proves[{\POL[{\OS[\cC]}]}] \tbox[\cC_1]\fA\fequiv\tbox[\cC_2]\fA
	\mbox{ iff }
	\cC_1\treq\cC_2
	\mydot$
\end{corollary}

Likewise, we instantiate $\POL$ to the theory of choreographic programming with memory updates simply following the steps above while pairing choreographies and memory configurations.

%%%%%%%%%%%%%%%%%%%%%%%%%%%%%%%%%%%%%%%%%%%%%%%%%%%%%%%%%%%%
\begin{example}
\def\comone{\iI_1}
\def\comtwo{\iI_2}
\def\tcomone{\tI_1}
\def\tcomtwo{\tI_2}
%%%%%%%%%%%%%%%%%%%%%%%%%%%%%%%%%%%%%%%%%%%%%%%%%%%%%%%%%%%%
%%%%%%%%%%%%%%%%%%%%%%%%%%%%%%%%%%%%%%%%%%%%%%%%%%%%%%%%%%%%
\begin{figure}
	$$\begin{array}{c}
		\vlderivation{
			\vliin{\land}{}{\vdash 
				\tbox[\comone;\comtwo;\cC]\fA 
				\fequiv 
				\tbox[\comtwo;\comone;\cC]\fA
			}{
				\vlin{\lor}{}{\vdash 
					\tdia[\comone;\comtwo;\cC]\nfA 
					\lor 
					\tbox[\comtwo;\comone;\cC]\fA
				}{
					\vliin{\bosr}{}{\vdash
						\tdia[\comone;\comtwo;\cC]\nfA 
						,
						\tbox[\comtwo;\comone;\cC]\fA
					}{
						\vlpr{\dD_1}{}{
							\tbox[\tcomone]\tbox[\comtwo;\cC]\fA
							,
							\tdia[\comtwo;\comone;\cC]\nfA 
						}
					}{
						\vlpr{\dD_2}{}{
							\tdia[\comone;\comtwo;\cC]\nfA 
							,
							\tbox[\tcomtwo]\tbox[\comone;\cC]\fA
						} 
					}
				}
			}{
				\vlin{\lor}{}{\vdash 
					\tbox[\comone;\comtwo;\cC]\fA 
					\lor 
					\tdia[\comtwo;\comone;\cC]\nfA
				}{
					\vliin{\bosr}{}{\vdash
						\tbox[\comone;\comtwo;\cC]\fA
						,
						\tdia[\comtwo;\comone;\cC]\nfA 
					}{
						\vlpr{\dD_1}{}{
							\tbox[\tcomone]\tbox[\comtwo;\cC]\fA
							,
							\tdia[\comtwo;\comone;\cC]\nfA 
						}
					}{
						\vlpr{\dD_2}{}{
							\tdia[\comone;\comtwo;\cC]\nfA 
							,
							\tbox[\tcomtwo]\tbox[\comone;\cC]\fA
						} 
					}
				}
			}
		}
	\\\\[10pt]
		\text{where }\qquad
		\begin{cases}
			\vlderivation{\vlpr{\dD_1}{}{\vdash\fGam_1}} = 
		&
			\vlderivation{
				\vliq{\dosr+\Wrule}{}{\vdash 
					\tdia[\comtwo;\comone;\cC]\nfA 
					,
					\tbox[\tcomone]\tbox[\comtwo;\cC]\fA
				}{
					\vlin{\AXrule}{}{\vdash
						\tdia[\tcomone]\tdia[\comtwo;\cC]\nfA 
						,
						\tbox[\tcomone]\tbox[\comtwo;\cC]\fA
					}{\vlhy{}}
				}
			}
		\\	
			\vlderivation{\vlpr{\dD_2}{}{\vdash\fGam_2}} = 
		&
			\vlderivation{
				\vliq{\dosr+\Wrule}{}{\vdash 
					\tdia[\comone;\comtwo;\cC]\nfA 
					,
					\tbox[\tcomtwo]\tbox[\comone;\cC]\fA
				}{
					\vlin{\AXrule}{}{\vdash
						\tdia[\tcomtwo]\tdia[\comone;\cC]\nfA 
						,
						\tbox[\tcomtwo]\tbox[\comone;\cC]\fA
					}{\vlhy{}}
				}
			}
		\end{cases}
	\end{array}$$

	\caption{
		The derivation proving $\comone;\comtwo;\cC\treq \comtwo;\comone;\cC$ whenever $\pn(\comtwo)\cap\left(\set{\pp}\cup\pn(\comone)\right)=\emptyset$ (because of the out-of-order execution).
	}
	\label{fig:ChorEx}
\end{figure}
%%%%%%%%%%%%%%%%%%%%%%%%%%%%%%%%%%%%%%%%%%%%%%%%%%%%%%%%%%%%
%%%%%%%%%%%%%%%%%%%%%%%%%%%%%%%%%%%%%%%%%%%%%%%%%%%%%%%%%%%%

Consider the choreographies $(\comone;\comtwo)$ and $(\comtwo;\comone)$.
The communications are the same save for their syntactic position and, since they involve distinct processes, out-of-order execution (rule $\rdeli$) ensures that these can fire concurrently. 
Indeed, these two choreographies are trace equivalent as shown, invoking \Cref{thm:chor:traceEq}, by the derivation reported in \Cref{fig:ChorEx}.
%%%%%%%%%%%%%%%%%%%%%%%%%%%%%%%%%%%%%%%%%%%%%%%%%%%%%%%%%%%%
%%%%%%%%%%%%%%%%%%%%%%%%%%%%%%%%%%%%%%%%%%%%%%%%%%%%%%%%%%%%
\begin{figure}
	\def\comone{\iI_1}
	\def\comtwo{\iI_2}
	\def\tcomone{\tI_1}
	\def\tcomtwo{\tI_2}
	$$\begin{array}{c}
		\vlderivation{
			\vlin{\lor}{}{\vdash 
				\tbox[\ccond pb {(\comone;\comtwo;\cC_1)}{(\comtwo;\cC_2)}]\fA 
				\imp
				\tbox[\comtwo;\ccond pb {\comone}{\cnil}]\fA
			}{
				\vlin{\bosr}{}{\vdash
					\tdia[\ccond pb {(\comone;\comtwo;\cC_1)}{(\comtwo;\cC_2)}]\nfA 
					, 
					\tbox[\comtwo;\ccond pb {\comone}{\cnil}]\fA
				}{
					\vliq{\dosr+2\times\Wrule}{}{\vdash 
						\tdia[\ccond pb {(\comone;\comtwo;\cC_1)}{(\comtwo;\cC_2)}]\nfA 
						, 
						\tbox[\tcomtwo]\tbox[\ccond pb {\comone}{\cnil}]\fA
					}{
						\vlin{\AXrule}{}{
							\tdia[\tcomtwo]\tdia[\ccond pb {(\comone;\cC_1)}{\cC_2}]\nfA 
							, 
							\tbox[\tcomtwo]\tbox[\ccond pb {\comone}{\cnil}]\fA
						}{\vlhy{}}
					}
				}
			}
		}
	\\\\[10pt]
		\vlderivation{
			\vlin{\lor}{}{\vdash 
				\tbox[\comtwo;\ccond pb {(\comone;\cC_1)}{\cC_2}]\fA
				\imp
				\tbox[\ccond pb {(\comtwo;\comone;\cC_1)}{(\comtwo;\cC_2)}]\fA 
			}{
				\vliiin{\bosr}{}{\vdash
					\tdia[\comtwo;\ccond pb {(\comone;\cC_1)}{\cC_2}]\nfA
					, 
					\tbox[\ccond pb {(\comtwo;\comone;\cC_1)}{(\comtwo;\cC_2)}]\fA 
				}{
					\vlpr{\dD_1}{}{\vdash\fGam_1}
				}{
					\vlpr{\dD_2}{}{\hskip8em\vdash\fGam_2\hskip8em}
				}{
					\vlpr{\dD_3}{}{\vdash\fGam_3}
				}
			}
		}
	\\\\[10pt]
		\text{where\qquad }
		\begin{cases}
			\vlderivation{\vlpr{\dD_1}{}{\vdash\fGam_1}} = 
		&
			\vlderivation{
				\vliq{\dagger}{}{\vdash 
					\tdia[\comtwo;\ccond pb {(\comone;\cC_1)}{\cC_2}]\nfA
					, 
					\tbox[\tcomtwo]\tbox[\ccond pb {(\comone;\cC_1)}{\cC_2}]\fA 
				}{
					\vlin{\AXrule}{}{
						\tdia[\tcomtwo]\tdia[\ccond pb {(\comone;\cC_1)}{\cC_2}]\nfA
						, 
						\tbox[\tcomtwo]\tbox[\ccond pb {(\comone;\cC_1)}{\cC_2}]\fA 
					}{\vlhy{}}
				}
			}
		\\
			\vlderivation{\vlpr{\dD_2}{}{\fGam_2}} = 
		&
			\vlderivation{
				\vliq{\dosr+2\times\Wrule}{}{\vdash 
					\tdia[\comtwo;\ccond pb {(\comone;\cC_1)}{\cC_2}]\nfA
					, 
					\tbox[\gentest]\tbox[\comtwo;\comone;\cC_1]\fA 
				}{
					\vlin{\AXrule}{}{
						\tdia[\gentest]\tdia[\comtwo;\comone;\cC_1]\fA 
						, 
						\tbox[\gentest]\tbox[\comtwo;\comone;\cC_1]\fA 
					}{
						\vlhy{}
					}
				}
			}
		\\
			\vlderivation{\vlpr{\dD_3}{}{\vdash\fGam_3}} = 
		&
			\vlderivation{
				\vliq{\dosr+2\times\Wrule}{}{\vdash 
					\tdia[\comtwo;\ccond pb {(\comone;\cC_1)}{\cC_2}]\nfA
					, 
					\tbox[\ngentest]\tbox[\comtwo;\cC_2]\fA 
				}{
					\vlin{\AXrule}{}{
						\tdia[\ngentest]\tdia[\comtwo;\cC_2]\fA 
						, 
						\tbox[\ngentest]\tbox[\comtwo;\cC_2]\fA 
					}{
						\vlhy{}
					}
				}
			}
		\end{cases}
	\end{array}
	$$
	\caption{
		The derivations proving $\ccond pb {(\comtwo;\comone;\cC_1)}{(\comtwo;\cC_2)}\treq\comtwo;\ccond pb {\comone;\cC_1}{\cC_2}$ whenever $\pn(\comtwo)\cap\left(\set{\pp}\cup\pn(\comone)\right)=\emptyset$ (because of the out-of-order execution).
		%  where $\dagger=\dosr+2\times\Wrule$.
	}
	\label{fig:ChorExIf}
\end{figure}
%%%%%%%%%%%%%%%%%%%%%%%%%%%%%%%%%%%%%%%%%%%%%%%%%%%%%%%%%%%%
%%%%%%%%%%%%%%%%%%%%%%%%%%%%%%%%%%%%%%%%%%%%%%%%%%%%%%%%%%%%
A similar case that illustrates concurrent conditionals and instructions is shown in \Cref{fig:ChorExIf}.
\end{example}
%%%%%%%%%%%%%%%%%%%%%%%%%%%%%%%%%%%%%%%%%%%%%%%%%%%%%%%%%%%%

%\begin{remark}
%	A key result in Choreographic programming is the EndPoint Projection Theorem, 
%	ensuring that the behaviour of a choroegraphy $\cC$ and a network obtained by projecting the  
%	
%	To be proven, the EndPoint Projection Theorem requires the use of the stateful semantics $\OS[\tuple{\cC,\Sigma}]$ in \Cref{fig:chor:sem}, and a similar semantics for the networks calculus.
%	
%	However, such a result should be provable in a sufficient expressive $\POL$, over an operational semantics for programs defined as choreographies or networks,
%	without the need of memory storage since this information.
%\end{remark}

%%%%%%%%%%%%%%%%%%%%%%%%%%%%%%%%%%%%%%%%%%%%%%%%%%%%%%%%%%%%
%%%%%%%%%%%%%%%%%%%%%%%%%%%%%%%%%%%%%%%%%%%%%%%%%%%%%%%%%%%%
%%%%%%%%%%%%%%%%%%%%%%%%%%%%%%%%%%%%%%%%%%%%%%%%%%%%%%%%%%%%
\section{Conclusion}\label{sec:perp}
%%%%%%%%%%%%%%%%%%%%%%%%%%%%%%%%%%%%%%%%%%%%%%%%%%%%%%%%%%%%
%%%%%%%%%%%%%%%%%%%%%%%%%%%%%%%%%%%%%%%%%%%%%%%%%%%%%%%%%%%%
%%%%%%%%%%%%%%%%%%%%%%%%%%%%%%%%%%%%%%%%%%%%%%%%%%%%%%%%%%%%

We have extended \PDL by decoupling reasoning on programs from reasoning on traces, bridged by a new axiom that integrates the two aspects.
This decoupling allowed us to create an axiom scheme parameterised on the operational semantics of the programs under consideration.
The result, $\POL$, subsumes a number of previous extensions of \PDL by seeing them as particular instantiations of this schema.
Furthermore, $\POL$ can be instantiated for programming languages out of reach of previous approaches, because of problematic standard features such as recursion, interleaving, or out-of-order execution.
Thus, we are hopeful that $\POL$ can be a useful tool for the future study of dynamic logic and formal methods.
We mention next a few interesting perspectives.

$\POL$, like standard \PDL, captures trace equivalence. Trace equivalence can be used to capture finer equivalences by decorating traces with information about choices \cite{vG90,JCDN13}, which for example was used in the context of \PDL and a simpler iterative process calculus (\CCSs) in~\cite{BENEVIDES201723}.
We plan to investigate this in the more general setting of $\POL$.
% \todo{
	We can already observe the following fundamental difference between our approach and the one of \emph{Concurrent Kleene Algebra} (CKA) \cite{HSMSZ16}: in CKA, the stuck program behaves as the absorbing element for the parallel composition (i.e., $p\parallel\mathbf{0}=\mathbf{0}$), while in our approach we are able to rule out this equivalence which is not necessarely sound for the semantics of the language under consideration, as, e.g., in the semantics of \CCS, where the set of traces of $\pnil \ppar A$ is not equal to the set of traces of $\pnil$.
% }

Having captured \CCS, a natural next step would be investigating how to capture even richer process calculi. 
The prime example would be the $\pi$-calculus~\cite{SW01}, which allows for dynamically creating and transmitting actions.
Work on \PDL for the $\pi$-calculus covers iteration \cite{PDL:pi}, but neither of the standard constructs for infinite behaviours, i.e., recursion and replication.
While $\POL$ can be directly instantiated with the standard $\pi$-calculus (retracing the steps for \CCS), the resulting notion of equivalence merits attention: the $\pi$-calculus has a richer behavioural theory than \CCS, which for example introduces the problem of equating traces up to action equivalence.

Likewise, there are numerous choreographic programming languages that would be interesting to study in $\POL$, because they pose additional challenges on top of out-of-order execution. Examples include dynamic process spawning~\cite{CM17}, parametric recursive procedures~\cite{CM17,montesi:book}, and higher-order composition~\cite{GMP24,CGLMP23,SKK23,HG22}.
As we mentioned, a key aspect of choreographic programming is endpoint projection: a mechanical mapping of choreographies into distributed implementations, usually given in terms of a process calculus.
Proving that endpoint projection is correct (an operational correspondence result) requires tedious work~\cite{CMP23}: $\POL$ could provide a unifying framework for these proofs, obtained by instantiating it with the union of the choreographic and target process languages.
Adopting this approach might make proofs more robust and reusable.

$\POL$ inherits the feature from \PDL that Hoare clauses $\set{\fA} \alpha \set{\fB}$ can be encoded as $\fA\imp\tbox[\alpha]\fB$.
All rules in Hoare logic are then derivable.
Thus, for example, our instantiation of $\POL$ with choreographic programming yields a direct generalisation of the previous development of a Hoare logic for choreographies~\cite{CGMP23}, providing a basis for its extension to more sophisticated languages.

Another line of future work is the study of the decision problem in (instantiations of) $\POL$.
In $\PDL$ the so-called \emph{small world model} is constructed using (the finiteness of) the Fisher-Ladner closure of a formula and provides a naive deterministic decidability procedure for the satisfability problem.
In $\POL$ the Fisher-Ladner closure is not guaranteed to be finite, an aspect that depends on the operational semantics under consideration.
In general, as shown in \cite{harel1996more}, any non-regular program add expressiveness power to \PDL, and the decision problem for a \PDL in which programs may have non-regular set of traces is known to be already $\Pi^1_1$-complete \cite{harel1983propositional}.
The validity problem for context-free \PDL is undecidable because so is the equivalence problem for general context-free languages \cite{hoperoft1979introduction,kozen2007automata}.
This is not surprising, since logical equivalence in \PDL captures trace equivalence.
In concurrency theory, there is an extensive literature on the relation between the design of concurrent languages and decidability of different program equivalences~\cite{Aceto_Ingolfsdottir_Srba_2011,busi:gab:zav:RRI}.
The methods studied therein might be useful for exploring decision problems in $\POL$, for example by establishing properties on specific operational semantics and how they are defined (rule formats, etc.).

Finally, it would be interesting to model a similar separation between trace reasoning and the operational semantics of programs in algebraic approaches for proving program equivalence.
For this we foresee the possibility of defining structures in which an operational semantics is `nested' inside a Kleene algebra.
Intuitively, such structures should be defined as Kleene algebras freely generated by a set of programs $\progSet$ and a set of atomic actions $\atomProg$ provided with a relation 
$\mathcal O \subseteq \progSet \times \atomProg \times \progSet$ (in general, a coalgebra $\mathcal O \subseteq \progSet \to \mathcal B(\atomProg \times \progSet$) representing the operational semantics of the set of programs.

The technique used in papers \cite{e-passport,e-passport:journal}, addressing the longstanding problem of unlinkability for e-passports by proving that the e-passport is not unlinkable, uses weak bisimilarity. For this purpose, the authors employ a technique based on a modal logic developed in \cite{hor:ahn:tiu:open}, which allows them to go beyond the expressiveness of process algebras -- currently insufficient to represent name passing. 
Thanks to the contribution in our paper, such modal logic has a straightforward interpretation: it corresponds to an intuitionistic version of $\POL$ instantiated over the internal $\pi$-calculus. By `intuitionistic', we mean the logic obtained by replacing the axiom $\mathbf{PL}$ in \Cref{fig:axPDL} with an axiomatization of intuitionistic logic. 
Note that the semantics of the box and diamond modalities in this logic are tuned on the standard definitions in the context of intuitionistic modal logic (see \cite{simpson:phd}, or \cite{das:mar:Imodal} for a more concise discussion on the different semantics between intuitionistic and constructive modal logics).

\begin{acks}
	Work co-funded by Villum Fonden (grant no. 50079) and the European Union (ERC, CHORDS, 101124225). Views and opinions expressed are however those of the authors only and do not necessarily reflect those of the European Union or the European Research Council. Neither the European Union nor the granting authority can be held responsible for them.
	
	We thank 
	Anupam Das,
	Sonia Marin, and
	Paaras Padhiar
	for the discussions on the proof of cut-elimination, and the remark on the need of restriction on the tests to make the cut-elimination proof work (see \Cref{rem:tests}).
	We would also thank Davide Catta for the fruitful discussion on the definition of cut-elimination strategy.
\end{acks}
%%%%%%%%%%%%%%%%%%%%%%%%%%%%%%%%%%%%%%%%%%%%%%%%%%%%%%%%%%%%
%%%%%%%%%%%%%%%%%%%%%%%%%%%%%%%%%%%%%%%%%%%%%%%%%%%%%%%%%%%%
%%%%%%%%%%%%%%%%%%%%%%%%%%%%%%%%%%%%%%%%%%%%%%%%%%%%%%%%%%%%

%%%%%%%%%%%%%%%%%%%%%%%%%%%%%%%%%%%%%%%%%%%%%%%%%%%%%%%%%%%%
%%%%%%%%%%%%%%%%%%%%%%%%%%%%%%%%%%%%%%%%%%%%%%%%%%%%%%%%%%%%
%%%%%%%%%%%%%%%%%%%%%%%%%%%%%%%%%%%%%%%%%%%%%%%%%%%%%%%%%%%%
\bibliographystyle{ACM-Reference-Format}
\bibliography{main}

@article{peleg:CPDL,
author = {Peleg, David},
title = {Concurrent dynamic logic},
year = {1987},
issue_date = {April 1987},
publisher = {Association for Computing Machinery},
address = {New York, NY, USA},
volume = {34},
number = {2},
issn = {0004-5411},
url = {https://doi.org/10.1145/23005.23008},
doi = {10.1145/23005.23008},
journal = {J. ACM},
month = apr,
pages = {450–479},
numpages = {30}
}

@inproceedings{hor:ahn:tiu:open,
	author = {Horne, Ross and Ahn, Ki Yung and Lin, Shang-wei and Tiu, Alwen},
	title = {Quasi-Open Bisimilarity with Mismatch is Intuitionistic},
	year = {2018},
	isbn = {9781450355834},
	publisher = {Association for Computing Machinery},
	address = {New York, NY, USA},
	url = {https://doi.org/10.1145/3209108.3209125},
	doi = {10.1145/3209108.3209125},
	booktitle = {Proceedings of the 33rd Annual ACM/IEEE Symposium on Logic in Computer Science},
	pages = {26–35},
	numpages = {10},
	location = {Oxford, United Kingdom},
	series = {LICS '18}
}

@InProceedings{aze:zhang:gabo:kleene,
  author =	{Azevedo de Amorim, Arthur and Zhang, Cheng and Gaboardi, Marco},
  title =	{{Kleene Algebra with Commutativity Conditions Is Undecidable}},
  booktitle =	{33rd EACSL Annual Conference on Computer Science Logic (CSL 2025)},
  pages =	{36:1--36:25},
  series =	{Leibniz International Proceedings in Informatics (LIPIcs)},
  ISBN =	{978-3-95977-362-1},
  ISSN =	{1868-8969},
  year =	{2025},
  volume =	{326},
  editor =	{Endrullis, J\"{o}rg and Schmitz, Sylvain},
  publisher =	{Schloss Dagstuhl -- Leibniz-Zentrum f{\"u}r Informatik},
  address =	{Dagstuhl, Germany},
  URL =		{https://drops.dagstuhl.de/entities/document/10.4230/LIPIcs.CSL.2025.36},
  URN =		{urn:nbn:de:0030-drops-227933},
  doi =		{10.4230/LIPIcs.CSL.2025.36}
}

@article{simpson:phd,
  title={The proof theory and semantics of intuitionistic modal logic},
  author={Simpson, Alex K},
  year={1994},
  publisher={University of Edinburgh. College of Science and Engineering. School of~…}
}

@InProceedings{das:mar:Imodal,
author="Das, Anupam
and Marin, Sonia",
editor="Ramanayake, Revantha
and Urban, Josef",
title="On Intuitionistic Diamonds (and Lack Thereof)",
booktitle="Automated Reasoning with Analytic Tableaux and Related Methods",
year="2023",
publisher="Springer Nature Switzerland",
address="Cham",
pages="283--301",
isbn="978-3-031-43513-3"
}

@InProceedings{e-passport,
	author="Filimonov, Ihor
	and Horne, Ross
	and Mauw, Sjouke
	and Smith, Zach",
	editor="Sako, Kazue
	and Schneider, Steve
	and Ryan, Peter Y. A.",
	title="Breaking Unlinkability of the ICAO 9303 Standard for e-Passports Using Bisimilarity",
	booktitle="Computer Security -- ESORICS 2019",
	year="2019",
	publisher="Springer International Publishing",
	address="Cham",
	pages="577--594",
	isbn="978-3-030-29959-0"
}

@article{e-passport:journal,
	title      = {Discovering ePassport Vulnerabilities using Bisimilarity},
	author     = {Ross Horne and Sjouke Mauw},
	url        = {https://lmcs.episciences.org/6117},
	doi        = {10.23638/LMCS-17(2:24)2021
	
	},
	journal    = {Logical Methods in Computer Science},
	issn       = {1860-5974},
	volume     = {Volume 17, Issue 2},
	eid        = 24,
	year       = {2021},
	month      = {Jun},
}

@article{HSMSZ16,
  author       = {Tony Hoare and
                  Stephan van Staden and
                  Bernhard M{\"{o}}ller and
                  Georg Struth and
                  Huibiao Zhu},
  title        = {Developments in concurrent Kleene algebra},
  journal      = {J. Log. Algebraic Methods Program.},
  volume       = {85},
  number       = {4},
  pages        = {617--636},
  year         = {2016},
  url          = {https://doi.org/10.1016/j.jlamp.2015.09.012},
  doi          = {10.1016/J.JLAMP.2015.09.012},
  bibsource    = {dblp computer science bibliography, https://dblp.org}
}

@article{HMSW11,
  author       = {Tony Hoare and
                  Bernhard M{\"{o}}ller and
                  Georg Struth and
                  Ian Wehrman},
  title        = {Concurrent Kleene Algebra and its Foundations},
  journal      = {J. Log. Algebraic Methods Program.},
  volume       = {80},
  number       = {6},
  pages        = {266--296},
  year         = {2011},
  url          = {https://doi.org/10.1016/j.jlap.2011.04.005},
  doi          = {10.1016/J.JLAP.2011.04.005},
  bibsource    = {dblp computer science bibliography, https://dblp.org}
}

@inproceedings{KB0WZ20,
  author       = {Tobias Kapp{\'{e}} and
                  Paul Brunet and
                  Alexandra Silva and
                  Jana Wagemaker and
                  Fabio Zanasi},
  editor       = {Jean Goubault{-}Larrecq and
                  Barbara K{\"{o}}nig},
  title        = {Concurrent Kleene Algebra with Observations: From Hypotheses to Completeness},
  booktitle    = {Foundations of Software Science and Computation Structures - 23rd
                  International Conference, {FOSSACS} 2020, Held as Part of the European
                  Joint Conferences on Theory and Practice of Software, {ETAPS} 2020,
                  Dublin, Ireland, April 25-30, 2020, Proceedings},
  series       = {Lecture Notes in Computer Science},
  volume       = {12077},
  pages        = {381--400},
  publisher    = {Springer},
  year         = {2020},
  url          = {https://doi.org/10.1007/978-3-030-45231-5\_20},
  doi          = {10.1007/978-3-030-45231-5\_20},
  bibsource    = {dblp computer science bibliography, https://dblp.org}
}

@inproceedings{BPS17,
  author       = {Paul Brunet and
                  Damien Pous and
                  Georg Struth},
  editor       = {Roland Meyer and
                  Uwe Nestmann},
  title        = {On Decidability of Concurrent Kleene Algebra},
  booktitle    = {28th International Conference on Concurrency Theory, {CONCUR} 2017,
                  September 5-8, 2017, Berlin, Germany},
  series       = {LIPIcs},
  volume       = {85},
  pages        = {28:1--28:15},
  publisher    = {Schloss Dagstuhl - Leibniz-Zentrum f{\"{u}}r Informatik},
  year         = {2017},
  url          = {https://doi.org/10.4230/LIPIcs.CONCUR.2017.28},
  doi          = {10.4230/LIPICS.CONCUR.2017.28},
  bibsource    = {dblp computer science bibliography, https://dblp.org}
}

@inbook{Aceto_Ingolfsdottir_Srba_2011, place={Cambridge}, series={Cambridge Tracts in Theoretical Computer Science}, title={The algorithmics of bisimilarity}, booktitle={Advanced Topics in Bisimulation and Coinduction}, publisher={Cambridge University Press}, author={Aceto, Luca and Ingolfsdottir, Anna and Srba, Jirí}, editor={Sangiorgi, Davide and Rutten, JanEditors}, year={2011}, pages={100–172}, collection={Cambridge Tracts in Theoretical Computer Science}}

@InProceedings{kleene_modules,
	author="Ehm, Thorsten
	and M{\"o}ller, Bernhard
	and Struth, Georg",
	editor="Berghammer, Rudolf
	and M{\"o}ller, Bernhard
	and Struth, Georg",
	title="Kleene Modules",
	booktitle="Relational and Kleene-Algebraic Methods in Computer Science",
	year="2004",
	publisher="Springer Berlin Heidelberg",
	address="Berlin, Heidelberg",
	pages="112--123",
	isbn="978-3-540-24771-5"
}

@article{hoperoft1979introduction,
	title={Introduction to automata theory, languages, and computation},
	author={Hoperoft, John E and Ullman, Jeffrey D},
	journal={Addison-Welsey, NY},
	year={1979}
}

@article{andreoli1992logic,
	title={Logic programming with focusing proofs in linear logic},
	author={Andreoli, Jean-Marc},
	journal={Journal of logic and computation},
	volume={2},
	number={3},
	pages={297--347},
	year={1992},
	publisher={Dov Gabbay}
}

@unpublished{liang:hal-03457379,
	TITLE = {{Focusing Gentzen's LK proof system}},
	AUTHOR = {Liang, Chuck and Miller, Dale},
	URL = {https://hal.science/hal-03457379},
	NOTE = {working paper or preprint},
	YEAR = {2021},
	MONTH = Nov,
	HAL_ID = {hal-03457379},
	HAL_VERSION = {v1},
}

@article{hemer2002don,
	title={Don't care non-determinism in logic program refinement},
	author={Hemer, David and Colvin, Robert and Hayes, Ian and Strooper, Paul},
	journal={Electronic Notes in Theoretical Computer Science},
	volume={61},
	pages={101--121},
	year={2002},
	publisher={Elsevier}
}

@article{harel1996more,
	title={More on nonregular PDL: Finite models and Fibonacci-like programs},
	author={Harel, David and Singerman, Eli},
	journal={information and computation},
	volume={128},
	number={2},
	pages={109--118},
	year={1996},
	publisher={Elsevier}
}

@article{harel1983propositional,
	title={Propositional dynamic logic of nonregular programs},
	author={Harel, David and Pnueli, Amir and Stavi, Jonathan},
	journal={Journal of Computer and System Sciences},
	volume={26},
	number={2},
	pages={222--243},
	year={1983},
	publisher={Elsevier}
}

@book{kozen2007automata,
	title={Automata and computability},
	author={Kozen, Dexter C},
	year={2007},
	publisher={Springer Science \& Business Media}
}

@article{HG22,
  author       = {Andrew K. Hirsch and
                  Deepak Garg},
  title        = {Pirouette: higher-order typed functional choreographies},
  journal      = {Proc. {ACM} Program. Lang.},
  volume       = {6},
  number       = {{POPL}},
  pages        = {1--27},
  year         = {2022},
  url          = {https://doi.org/10.1145/3498684},
  doi          = {10.1145/3498684},
  bibsource    = {dblp computer science bibliography, https://dblp.org}
}

@article{SKK23,
  author       = {Gan Shen and
                  Shun Kashiwa and
                  Lindsey Kuper},
  title        = {HasChor: Functional Choreographic Programming for All (Functional
                  Pearl)},
  journal      = {Proc. {ACM} Program. Lang.},
  volume       = {7},
  number       = {{ICFP}},
  pages        = {541--565},
  year         = {2023},
  url          = {https://doi.org/10.1145/3607849},
  doi          = {10.1145/3607849},
  bibsource    = {dblp computer science bibliography, https://dblp.org}
}

@article{GMP24,
  author = {Giallorenzo, Saverio and Montesi, Fabrizio and Peressotti, Marco},
  title = {Choral: Object-oriented Choreographic Programming},
  year = {2024},
  issue_date = {March 2024},
  publisher = {Association for Computing Machinery},
  address = {New York, NY, USA},
  journal = {ACM Trans. Program. Lang. Syst.},
  month = jan,
  articleno = {1},
  numpages = {59},
  volume = {46},
  number = {1},
  issn = {0164-0925},
  url = {https://doi.org/10.1145/3632398},
  doi = {10.1145/3632398},
  month_numeric = {1}
}

@inproceedings{CGLMP23,
  author       = {Lu{\'{\i}}s Cruz{-}Filipe and
                  Eva Graversen and
                  Lovro Lugovic and
                  Fabrizio Montesi and
                  Marco Peressotti},
  editor       = {Karim Ali and
                  Guido Salvaneschi},
  title        = {Modular Compilation for Higher-Order Functional Choreographies},
  booktitle    = {37th European Conference on Object-Oriented Programming, {ECOOP} 2023,
                  July 17-21, 2023, Seattle, Washington, United States},
  series       = {LIPIcs},
  volume       = {263},
  pages        = {7:1--7:37},
  publisher    = {Schloss Dagstuhl - Leibniz-Zentrum f{\"{u}}r Informatik},
  year         = {2023},
  url          = {https://doi.org/10.4230/LIPIcs.ECOOP.2023.7},
  doi          = {10.4230/LIPICS.ECOOP.2023.7},
  bibsource    = {dblp computer science bibliography, https://dblp.org}
}

@inproceedings{CM17,
  author       = {Lu{\'{\i}}s Cruz{-}Filipe and
                  Fabrizio Montesi},
  editor       = {Ahmed Bouajjani and
                  Alexandra Silva},
  title        = {Procedural Choreographic Programming},
  booktitle    = {Formal Techniques for Distributed Objects, Components, and Systems
                  - 37th {IFIP} {WG} 6.1 International Conference, {FORTE} 2017, Held
                  as Part of the 12th International Federated Conference on Distributed
                  Computing Techniques, DisCoTec 2017, Neuch{\^{a}}tel, Switzerland,
                  June 19-22, 2017, Proceedings},
  series       = {Lecture Notes in Computer Science},
  volume       = {10321},
  pages        = {92--107},
  publisher    = {Springer},
  year         = {2017},
  url          = {https://doi.org/10.1007/978-3-319-60225-7\_7},
  doi          = {10.1007/978-3-319-60225-7\_7},
  bibsource    = {dblp computer science bibliography, https://dblp.org}
}

@book{SW01,
  author       = {Davide Sangiorgi and
                  David Walker},
  title        = {The Pi-Calculus - a theory of mobile processes},
  publisher    = {Cambridge University Press},
  year         = {2001},
  isbn         = {978-0-521-78177-0},
  bibsource    = {dblp computer science bibliography, https://dblp.org}
}

@inproceedings{GMG18,
  author       = {Saverio Giallorenzo and
                  Fabrizio Montesi and
                  Maurizio Gabbrielli},
  editor       = {Christel Baier and
                  Lu{\'{\i}}s Caires},
  title        = {Applied Choreographies},
  booktitle    = {Formal Techniques for Distributed Objects, Components, and Systems
                  - 38th {IFIP} {WG} 6.1 International Conference, {FORTE} 2018, Held
                  as Part of the 13th International Federated Conference on Distributed
                  Computing Techniques, DisCoTec 2018, Madrid, Spain, June 18-21, 2018,
                  Proceedings},
  series       = {Lecture Notes in Computer Science},
  volume       = {10854},
  pages        = {21--40},
  publisher    = {Springer},
  year         = {2018},
  url          = {https://doi.org/10.1007/978-3-319-92612-4\_2},
  doi          = {10.1007/978-3-319-92612-4\_2},
  bibsource    = {dblp computer science bibliography, https://dblp.org}
}

@PhdThesis{M13:phd,
  author = "Fabrizio Montesi",
  title = "Choreographic {P}rogramming",
  school = "IT University of Copenhagen",
  type = "Ph.{D}. Thesis",
  year = 2013,
  note = {\url{https://www.fabriziomontesi.com/files/choreographic-programming.pdf}}
}

@book{M80,
  author       = {Robin Milner},
  title        = {A Calculus of Communicating Systems},
  series       = {Lecture Notes in Computer Science},
  volume       = {92},
  publisher    = {Springer},
  year         = {1980},
  url          = {https://doi.org/10.1007/3-540-10235-3},
  doi          = {10.1007/3-540-10235-3},
  isbn         = {3-540-10235-3},
  bibsource    = {dblp computer science bibliography, https://dblp.org}
}

@article{K97,
  author       = {Dexter Kozen},
  title        = {Kleene Algebra with Tests},
  journal      = {{ACM} Trans. Program. Lang. Syst.},
  volume       = {19},
  number       = {3},
  pages        = {427--443},
  year         = {1997},
  url          = {https://doi.org/10.1145/256167.256195},
  doi          = {10.1145/256167.256195},
  bibsource    = {dblp computer science bibliography, https://dblp.org}
}

@article{BENEVIDES201723,
title = {Bisimilar and logically equivalent programs in PDL with parallel operator},
journal = {Theoretical Computer Science},
volume = {685},
pages = {23-45},
year = {2017},
note = {Logical and Semantic Frameworks with Applications},
issn = {0304-3975},
doi = {https://doi.org/10.1016/j.tcs.2017.02.037},
url = {https://www.sciencedirect.com/science/article/pii/S0304397517302785},
author = {Mario Benevides}
}

@inproceedings{KBRSWZ19,
  author       = {Tobias Kapp{\'{e}} and
                  Paul Brunet and
                  Jurriaan Rot and
                  Alexandra Silva and
                  Jana Wagemaker and
                  Fabio Zanasi},
  editor       = {Wan J. Fokkink and
                  Rob van Glabbeek},
  title        = {Kleene Algebra with Observations},
  booktitle    = {30th International Conference on Concurrency Theory, {CONCUR} 2019,
                  August 27-30, 2019, Amsterdam, the Netherlands},
  series       = {LIPIcs},
  volume       = {140},
  pages        = {41:1--41:16},
  publisher    = {Schloss Dagstuhl - Leibniz-Zentrum f{\"{u}}r Informatik},
  year         = {2019},
  url          = {https://doi.org/10.4230/LIPIcs.CONCUR.2019.41},
  doi          = {10.4230/LIPICS.CONCUR.2019.41},
  bibsource    = {dblp computer science bibliography, https://dblp.org}
}

@inproceedings{SKS23,
  author       = {Todd Schmid and
                  Tobias Kapp{\'{e}} and
                  Alexandra Silva},
  editor       = {Thomas Wies},
  title        = {A Complete Inference System for Skip-free Guarded Kleene Algebra with
                  Tests},
  booktitle    = {Programming Languages and Systems - 32nd European Symposium on Programming,
                  {ESOP} 2023, Held as Part of the European Joint Conferences on Theory
                  and Practice of Software, {ETAPS} 2023, Paris, France, April 22-27,
                  2023, Proceedings},
  series       = {Lecture Notes in Computer Science},
  volume       = {13990},
  pages        = {309--336},
  publisher    = {Springer},
  year         = {2023},
  url          = {https://doi.org/10.1007/978-3-031-30044-8\_12},
  doi          = {10.1007/978-3-031-30044-8\_12},
  bibsource    = {dblp computer science bibliography, https://dblp.org}
}

@article{hopcroft2001introduction,
  title={Introduction to automata theory, languages, and computation},
  author={Hopcroft, John E and Motwani, Rajeev and Ullman, Jeffrey D},
  journal={Acm Sigact News},
  volume={32},
  number={1},
  pages={60--65},
  year={2001},
  publisher={ACM New York, NY, USA}
}

@article{stirling1991local,
  title={Local model checking in the modal mu-calculus},
  author={Stirling, Colin and Walker, David},
  journal={Theoretical Computer Science},
  volume={89},
  number={1},
  pages={161--177},
  year={1991},
  publisher={Elsevier}
}

@inproceedings{CGKSVWW13,
  author       = {Sjoerd Cranen and
                  Jan Friso Groote and
                  Jeroen J. A. Keiren and
                  Frank P. M. Stappers and
                  Erik P. de Vink and
                  Wieger Wesselink and
                  Tim A. C. Willemse},
  editor       = {Nir Piterman and
                  Scott A. Smolka},
  title        = {An Overview of the mCRL2 Toolset and Its Recent Advances},
  booktitle    = {Tools and Algorithms for the Construction and Analysis of Systems
                  - 19th International Conference, {TACAS} 2013, Held as Part of the
                  European Joint Conferences on Theory and Practice of Software, {ETAPS}
                  2013, Rome, Italy, March 16-24, 2013. Proceedings},
  series       = {Lecture Notes in Computer Science},
  volume       = {7795},
  pages        = {199--213},
  publisher    = {Springer},
  year         = {2013},
  url          = {https://doi.org/10.1007/978-3-642-36742-7\_15},
  doi          = {10.1007/978-3-642-36742-7\_15},
  bibsource    = {dblp computer science bibliography, https://dblp.org}
}

@article{MILLER1991125,
title = {Uniform proofs as a foundation for logic programming},
journal = {Annals of Pure and Applied Logic},
volume = {51},
number = {1},
pages = {125-157},
year = {1991},
issn = {0168-0072},
doi = {https://doi.org/10.1016/0168-0072(91)90068-W},
url = {https://www.sciencedirect.com/science/article/pii/016800729190068W},
author = {Dale Miller and Gopalan Nadathur and Frank Pfenning and Andre Scedrov}
}

@book{lloyd2012foundations,
  title={Foundations of logic programming},
  author={Lloyd, John W},
  year={2012},
  publisher={Springer Science \& Business Media}
}

@inproceedings{CP10,
  author       = {Lu{\'{\i}}s Caires and
                  Frank Pfenning},
  editor       = {Paul Gastin and
                  Fran{\c{c}}ois Laroussinie},
  title        = {Session Types as Intuitionistic Linear Propositions},
  booktitle    = {{CONCUR} 2010 - Concurrency Theory, 21th International Conference,
                  {CONCUR} 2010, Paris, France, August 31-September 3, 2010. Proceedings},
  series       = {Lecture Notes in Computer Science},
  volume       = {6269},
  pages        = {222--236},
  publisher    = {Springer},
  year         = {2010},
  url          = {https://doi.org/10.1007/978-3-642-15375-4\_16},
  doi          = {10.1007/978-3-642-15375-4\_16},
  bibsource    = {dblp computer science bibliography, https://dblp.org}
}

@article{W15,
  author       = {Philip Wadler},
  title        = {Propositions as types},
  journal      = {Commun. {ACM}},
  volume       = {58},
  number       = {12},
  pages        = {75--84},
  year         = {2015},
  url          = {https://doi.org/10.1145/2699407},
  doi          = {10.1145/2699407},
  bibsource    = {dblp computer science bibliography, https://dblp.org}
}

@InProceedings{acc:cat:games,
	author="Acclavio, Matteo
	and Catta, Davide",
	editor="Malvone, Vadim
	and Murano, Aniello",
	title="Lorenzen-Style Strategies as Proof-Search Strategies",
	booktitle="Multi-Agent Systems",
	year="2023",
	publisher="Springer Nature Switzerland",
	address="Cham",
	pages="150--166",
	isbn="978-3-031-43264-4"
}

@article{hill:pog:PDL,
	author = {Hill, Brian and Poggiolesi, Francesca},
	date = {2010/02/01},
	doi = {10.1007/s11225-010-9224-z},
	id = {Hill2010},
	isbn = {1572-8730},
	journal = {Studia Logica},
	number = {1},
	pages = {47--72},
	title = {A Contraction-free and Cut-free Sequent Calculus for Propositional Dynamic Logic},
	url = {https://doi.org/10.1007/s11225-010-9224-z},
	volume = {94},
	year = {2010}}

@article{gir:ll,
	title = {Linear logic},
	journal = {Theoretical Computer Science},
	volume = {50},
	number = {1},
	pages = {1-101},
	year = {1987},
	issn = {0304-3975},
	doi = {10.1016/0304-3975(87)90045-4},
	author = {Jean-Yves Girard}
}

@inproceedings{maz:ter:15,
	author    = {Damiano Mazza and
	Kazushige Terui},
	editor    = {Magn{\'{u}}s M. Halld{\'{o}}rsson and
	Kazuo Iwama and
	Naoki Kobayashi and
	Bettina Speckmann},
	title     = {Parsimonious Types and Non-uniform Computation},
	booktitle = {Automata, Languages, and Programming - 42nd International Colloquium,
	{ICALP} 2015, Kyoto, Japan, July 6-10, 2015, Proceedings, Part {II}},
	series    = {Lecture Notes in Computer Science},
	volume    = {9135},
	pages     = {350--361},
	publisher = {Springer},
	year      = {2015},
	url       = {https://doi.org/10.1007/978-3-662-47666-6\_28},
	doi       = {10.1007/978-3-662-47666-6\_28},
	bibsource = {dblp computer science bibliography, https://dblp.org}
}

@InProceedings{doc:row:PDL,
	author="Docherty, Simon
	and Rowe, Reuben N. S.",
	editor="Cerrito, Serenella
	and Popescu, Andrei",
	title="A Non-wellfounded, Labelled Proof System for Propositional Dynamic Logic",
	booktitle="Automated Reasoning with Analytic Tableaux and Related Methods",
	year="2019",
	publisher="Springer International Publishing",
	address="Cham",
	pages="335--352",
	isbn="978-3-030-29026-9"
}

@article{girard:98,
	title = {Light Linear Logic},
	journal = {Information and Computation},
	volume = {143},
	number = {2},
	pages = {175-204},
	year = {1998},
	issn = {0890-5401},
	doi = {10.1006/inco.1998.2700},
	ignoreurl = {https://www.sciencedirect.com/science/article/pii/S0890540198927006},
	author = {Jean-Yves Girard}
}

@inproceedings{mazza:15,
	author    = {Damiano Mazza},
	ignoreeditor    = {Stephan Kreutzer},
	title     = {Simple Parsimonious Types and Logarithmic Space},
	booktitle = {24th {EACSL} Annual Conference on Computer Science Logic, {CSL} 2015},
	series    = {LIPIcs},
	volume    = {41},
	pages     = {24--40},
	publisher = {Schloss Dagstuhl - Leibniz-Zentrum f{\"{u}}r Informatik},
	year      = {2015},
	ignoreurl       = {https://doi.org/10.4230/LIPIcs.CSL.2015.24},
	doi       = {10.4230/LIPIcs.CSL.2015.24},
	bibsource = {dblp computer science bibliography, https://dblp.org}
}

@inproceedings{bae:dou:sau:16,
	author    = {David Baelde and
	Amina Doumane and
	Alexis Saurin},
	editor    = {Jean{-}Marc Talbot and
	Laurent Regnier},
	title     = {Infinitary Proof Theory: the Multiplicative Additive Case},
	booktitle = {25th {EACSL} Annual Conference on Computer Science Logic, {CSL} 2016,
	August 29 - September 1, 2016, Marseille, France},
	series    = {LIPIcs},
	volume    = {62},
	pages     = {42:1--42:17},
	publisher = {Schloss Dagstuhl - Leibniz-Zentrum f{\"{u}}r Informatik},
	year      = {2016},
	url       = {https://doi.org/10.4230/LIPIcs.CSL.2016.42},
	doi       = {10.4230/LIPIcs.CSL.2016.42},
	bibsource = {dblp computer science bibliography, https://dblp.org}
}

@book{montesi:book,
	author={Montesi, Fabrizio},
	title={Introduction to Choreographies},
	place={Cambridge},
	doi={10.1017/9781108981491},
	publisher={Cambridge University Press},
	year={2023}
}

@inproceedings{CGMP23,
  author = {Cruz{-}Filipe, Lu{\'{\i}}s and Graversen, Eva and Montesi, Fabrizio and Peressotti, Marco},
  editor = {Jongmans, Sung{-}Shik and Lopes, Ant{\'{o}}nia},
  title = {Reasoning About Choreographic Programs},
  booktitle = {Coordination Models and Languages},
  series = {Lecture Notes in Computer Science},
  volume = {13908},
  pages = {144--162},
  publisher = {Springer},
  year = {2023},
  url = {https://doi.org/10.1007/978-3-031-35361-1\_8},
  doi = {10.1007/978-3-031-35361-1\_8},
  isbn = {978-3-031-35361-1}
}

@article{CMP23,
  author = {Cruz{-}Filipe, Lu{\'{\i}}s and Montesi, Fabrizio and Peressotti, Marco},
  title = {A Formal Theory of Choreographic Programming},
  journal = {Journal of Automated Reasoning},
  volume = {67},
  number = {21},
  pages = {1--34},
  year = {2023},
  issn = {1573-0670},
  url = {https://doi.org/10.1007/s10817-023-09665-3},
  doi = {10.1007/s10817-023-09665-3}
}

@misc{wscdl,
  author       = {{W3C}},
  editor       = {Nickolas Kavantzas, David Burdett, Gregory Ritzinger, Tony Fletcher,
                  Yves Lafon and Charlton Barreto},
  title        = {{WS Choreography Description Language}},
  howpublished = {\href{http://www.w3.org/TR/ws-cdl-10/}{http://www.w3.org/TR/ws-cdl-10/}},
  year         = {2004}
}

@misc{bpmn,
  title        = {{B}usiness {P}rocess {M}odel and {N}otation},
  author       = {{O}bject {M}anagement {G}roup},
  howpublished = {\href{http://www.omg.org/spec/BPMN/2.0/}{http://www.omg.org/spec/BPMN/2.0/}},
  year         = {2011}
}

@article{mil:par:wal:pi,
	title = {A calculus of mobile processes, I},
	journal = {Information and Computation},
	volume = {100},
	number = {1},
	pages = {1-40},
	year = {1992},
	issn = {0890-5401},
	doi = {https://doi.org/10.1016/0890-5401(92)90008-4},
	url = {https://www.sciencedirect.com/science/article/pii/0890540192900084},
	author = {Robin Milner and Joachim Parrow and David Walker}
}

@article{sequent:modalMu,
	ISSN = {00393215, 15728730},
	URL = {http://www.jstor.org/stable/40268983},
	author = {Thomas Studer},
	journal = {Studia Logica: An International Journal for Symbolic Logic},
	number = {3},
	pages = {343--363},
	publisher = {Springer},
	title = {On the Proof Theory of the Modal mu-Calculus},
	urldate = {2024-01-01},
	volume = {89},
	year = {2008}
}

@InProceedings{kozen:und,
	author="Kozen, Dexter",
	editor="Margaria, Tiziana and Steffen, Bernhard",
	title="Kleene algebra with tests and commutativity conditions",
	booktitle="Tools and Algorithms for the Construction and Analysis of Systems",
	year="1996",
	publisher="Springer Berlin Heidelberg",
	address="Berlin, Heidelberg",
	pages="14--33",
	isbn="978-3-540-49874-2"
}

@article{PDLcompleteness,
	title = {An elementary proof of the completeness of PDL},
	journal = {Theoretical Computer Science},
	volume = {14},
	number = {1},
	pages = {113-118},
	year = {1981},
	issn = {0304-3975},
	doi = {https://doi.org/10.1016/0304-3975(81)90019-0},
	url = {https://www.sciencedirect.com/science/article/pii/0304397581900190},
	author = {Dexter Kozen and Rohit Parikh}
}

@article{acc:cur:gue:CSL24,
	author       = {Matteo Acclavio and
	Gianluca Curzi and
	Giulio Guerrieri},
	title        = {Infinitary cut-elimination via finite approximations},
	journal      = {CoRR},
	volume       = {abs/2308.07789},
	year         = {2023},
	url          = {https://doi.org/10.48550/arXiv.2308.07789},
	doi          = {10.48550/ARXIV.2308.07789},
	eprinttype    = {arXiv},
	eprint       = {2308.07789},
	bibsource    = {dblp computer science bibliography, https://dblp.org}
}

@misc{acc:cur:gue:CSL24ext,
	title={Infinitary cut-elimination via finite approximations (extended version)}, 
	author={Matteo Acclavio and Gianluca Curzi and Giulio Guerrieri},
	year={2024},
	eprint={2308.07789},
	archivePrefix={arXiv},
	primaryClass={cs.LO},
	url={https://arxiv.org/abs/2308.07789}, 
}

@Inbook{DLbook,
	author="Harel, David and Kozen, Dexter and Tiuryn, Jerzy",
	editor="Gabbay, D. M. and Guenthner, F.",
	title="Dynamic Logic",
	bookTitle="Handbook of Philosophical Logic",
	year="2002",
	publisher="Springer Netherlands",
	address="Dordrecht",
	pages="99--217",
	isbn="978-94-017-0456-4",
	doi="10.1007/978-94-017-0456-4_2",
	url="https://doi.org/10.1007/978-94-017-0456-4_2"
}

@InProceedings{das:gir:TCL,
	author="Das, Anupam
	and Girlando, Marianna",
	editor="Blanchette, Jasmin
	and Kov{\'a}cs, Laura
	and Pattinson, Dirk",
	title="Cyclic Proofs, Hypersequents, and Transitive Closure Logic",
	booktitle="Automated Reasoning",
	year="2022",
	publisher="Springer International Publishing",
	address="Cham",
	pages="509--528",
	isbn="978-3-031-10769-6"
}

@article{das:gir:journal,
	author = {Das, Anupam and Girlando, Marianna},
	date = {2023/08/16},
	doi = {10.1007/s10817-023-09675-1},
	id = {Das2023},
	isbn = {1573-0670},
	journal = {Journal of Automated Reasoning},
	number = {3},
	pages = {27},
	title = {Cyclic Hypersequent System for Transitive Closure Logic},
	url = {https://doi.org/10.1007/s10817-023-09675-1},
	volume = {67},
	year = {2023}
}

@misc{das:gir:TCLext,
	title={Cyclic Proofs, Hypersequents, and Transitive Closure Logic}, 
	author={Anupam Das and Marianna Girlando},
	year={2022},
	eprint={2205.08616},
	archivePrefix={arXiv},
	primaryClass={cs.LO}
}

@article{PDL:pi,
	title = {A Propositional Dynamic Logic for Concurrent Programs Based on the $\pi$-Calculus},
	journal = {Electronic Notes in Theoretical Computer Science},
	volume = {262},
	pages = {49-64},
	year = {2010},
	note = {Proceedings of the 6th Workshop on Methods for Modalities (M4M-6 2009)},
	issn = {1571-0661},
	doi = {https://doi.org/10.1016/j.entcs.2010.04.005},
	url = {https://www.sciencedirect.com/science/article/pii/S1571066110000277},
	author = {Mario R.F. Benevides and L. Menasché Schechter}
}

@article{pel:CDL,
	author = {Peleg, David},
	title = {Concurrent dynamic logic},
	year = {1987},
	issue_date = {April 1987},
	publisher = {Association for Computing Machinery},
	address = {New York, NY, USA},
	volume = {34},
	number = {2},
	issn = {0004-5411},
	url = {https://doi.org/10.1145/23005.23008},
	doi = {10.1145/23005.23008},
	journal = {J. ACM},
	month = {apr},
	pages = {450–479},
	numpages = {30}
}

@article{pel:concurrentCom,
	title = {Communication in concurrent dynamic logic},
	journal = {Journal of Computer and System Sciences},
	volume = {35},
	number = {1},
	pages = {23-58},
	year = {1987},
	issn = {0022-0000},
	doi = {https://doi.org/10.1016/0022-0000(87)90035-3},
	url = {https://www.sciencedirect.com/science/article/pii/0022000087900353},
	author = {David Peleg}
}

@article{pel:concurrentschemes,
	title = {Concurrent program schemes and their logics},
	journal = {Theoretical Computer Science},
	volume = {55},
	number = {1},
	pages = {1-45},
	year = {1987},
	issn = {0304-3975},
	doi = {https://doi.org/10.1016/0304-3975(87)90088-0},
	url = {https://www.sciencedirect.com/science/article/pii/0304397587900880},
	author = {David Peleg}
}

@article{PDL:interleaving,
	title = {The complexity of PDL with interleaving},
	journal = {Theoretical Computer Science},
	volume = {161},
	number = {1},
	pages = {109-122},
	year = {1996},
	issn = {0304-3975},
	doi = {https://doi.org/10.1016/0304-3975(95)00095-X},
	url = {https://www.sciencedirect.com/science/article/pii/030439759500095X},
	author = {Alain J. Mayer and Larry J. Stockmeyer}
}

@book{troelstra_schwichtenberg_2000, 
	place={Cambridge}, 
	edition={2}, series={Cambridge Tracts in Theoretical Computer Science},
	 title={Basic Proof Theory}, 
	DOI={10.1017/CBO9781139168717}, 
	publisher={Cambridge University Press}, 
	author={Troelstra, A. S. and Schwichtenberg, H.}, 
	year={2000}, 
	collection={Cambridge Tracts in Theoretical Computer Science}
}

@InProceedings{pratt:PDL,
	author="Pratt, V. R.",
	editor="Kozen, Dexter",
	title="Using graphs to understand PDL",
	booktitle="Logics of Programs",
	year="1982",
	publisher="Springer Berlin Heidelberg",
	address="Berlin, Heidelberg",
	pages="387--396",
	isbn="978-3-540-39047-3"
}

@article{PDL:flow,
	title = {Propositional dynamic logic of flowcharts},
	journal = {Information and Control},
	volume = {64},
	number = {1},
	pages = {119-135},
	year = {1985},
	note = {International Conference on Foundations of Computation Theory},
	issn = {0019-9958},
	doi = {https://doi.org/10.1016/S0019-9958(85)80047-4},
	url = {https://www.sciencedirect.com/science/article/pii/S0019995885800474},
	author = {D. Harel and R. Sherman}
}

@InProceedings{busi:gab:zav:RRI,
	author="Busi, Nadia
	and Gabbrielli, Maurizio
	and Zavattaro, Gianluigi",
	editor="D{\'i}az, Josep
	and Karhum{\"a}ki, Juhani
	and Lepist{\"o}, Arto
	and Sannella, Donald",
	title="Comparing Recursion, Replication, and Iteration in Process Calculi",
	booktitle="Automata, Languages and Programming",
	year="2004",
	publisher="Springer Berlin Heidelberg",
	address="Berlin, Heidelberg",
	pages="307--319",
	isbn="978-3-540-27836-8"
}

@article{games:PDL,
	title={Games for modal and temporal logics},
	author={Lange, Martin},
	year={2003},
	publisher={University of Edinburgh. College of Science and Engineering. School of~…}
}

@article{games:modal-mu,
	title = {Games for the $\mu$-calculus},
	journal = {Theoretical Computer Science},
	volume = {163},
	number = {1},
	pages = {99-116},
	year = {1996},
	issn = {0304-3975},
	doi = {https://doi.org/10.1016/0304-3975(95)00136-0},
	url = {https://www.sciencedirect.com/science/article/pii/0304397595001360},
	author = {Damian Niwiński and Igor Walukiewicz}
}

@article{maz:LLandP, 
	title={Linear logic and polynomial time}, 
	volume={16}, 
	DOI={10.1017/S0960129506005688}, 
	number={6}, 
	journal={Mathematical Structures in Computer Science}, 
	publisher={Cambridge University Press}, 
	author={Mazza, Damiano}, 
	year={2006}, 
	pages={947–988}
}

@book{gir:proot,
	title={Proofs and types},
	author={Girard, Jean-Yves and Taylor, Paul and Lafont, Yves},
	volume={7},
	year={1989},
	publisher={Cambridge university press Cambridge}
}

@article{laf:soft,
	title={Soft linear logic and polynomial time},
	author={Lafont, Yves},
	journal={Theoretical computer science},
	volume={318},
	number={1-2},
	pages={163--180},
	year={2004},
	publisher={Elsevier}
}

@InProceedings{vG90,
author="van Glabbeek, R. J.",
editor="Baeten, J. C. M.
and Klop, J. W.",
title="The linear time - branching time spectrum",
booktitle="CONCUR '90 Theories of Concurrency: Unification and Extension",
year="1990",
publisher="Springer Berlin Heidelberg",
address="Berlin, Heidelberg",
pages="278--297",
isbn="978-3-540-46395-5"
}

@article{JCDN13,
title = {Algebraic characterizations of trace and decorated trace equivalences over tree-like structures},
journal = {Theoretical Computer Science},
volume = {254},
number = {1},
pages = {337-361},
year = {2001},
issn = {0304-3975},
doi = {https://doi.org/10.1016/S0304-3975(99)00300-X},
url = {https://www.sciencedirect.com/science/article/pii/S030439759900300X},
author = {Xiao {Jun Chen} and Rocco De Nicola},
}

@incollection{AFV01,
  author       = {Luca Aceto and
                  Wan J. Fokkink and
                  Chris Verhoef},
  editor       = {Jan A. Bergstra and
                  Alban Ponse and
                  Scott A. Smolka},
  title        = {Structural Operational Semantics},
  booktitle    = {Handbook of Process Algebra},
  pages        = {197--292},
  publisher    = {North-Holland / Elsevier},
  year         = {2001},
  url          = {https://doi.org/10.1016/b978-044482830-9/50021-7},
  doi          = {10.1016/B978-044482830-9/50021-7}
}

@article{GV92,
  author       = {Jan Friso Groote and
                  Frits W. Vaandrager},
  title        = {Structured Operational Semantics and Bisimulation as a Congruence},
  journal      = {Inf. Comput.},
  volume       = {100},
  number       = {2},
  pages        = {202--260},
  year         = {1992},
  url          = {https://doi.org/10.1016/0890-5401(92)90013-6},
  doi          = {10.1016/0890-5401(92)90013-6}
}
%%%%%%%%%%%%%%%%%%%%%%%%%%%%%%%%%%%%%%%%%%%%%%%%%%%%%%%%%%%%
%%%%%%%%%%%%%%%%%%%%%%%%%%%%%%%%%%%%%%%%%%%%%%%%%%%%%%%%%%%%
%%%%%%%%%%%%%%%%%%%%%%%%%%%%%%%%%%%%%%%%%%%%%%%%%%%%%%%%%%%%

%%%%%%%%%%%%%%%%%%%%%%%%%%%%%%%%%%%%%%%%%%%%%%%%%%%%%%%%%%%%%
%%%%%%%%%%%%%%%%%%%%%%%%%%%%%%%%%%%%%%%%%%%%%%%%%%%%%%%%%%%%%
%%%%%%%%%%%%%%%%%%%%%%%%%%%%%%%%%%%%%%%%%%%%%%%%%%%%%%%%%%%%%
%\clearpage
%\appendix
%%%%%%%%%%%%%%%%%%%%%%%%%%%%%%%%%%%%%%%%%%%%%%%%%%%%%%%%%%%%%
%%%%%%%%%%%%%%%%%%%%%%%%%%%%%%%%%%%%%%%%%%%%%%%%%%%%%%%%%%%%%
%%%%%%%%%%%%%%%%%%%%%%%%%%%%%%%%%%%%%%%%%%%%%%%%%%%%%%%%%%%%%

%%%%%%%%%%%%%%%%%%%%%%%%%%%%%%%%%%%%%%%%%%%%%%%%%%%%%%%%%%%%
%%%%%%%%%%%%%%%%%%%%%%%%%%%%%%%%%%%%%%%%%%%%%%%%%%%%%%%%%%%%
%%%%%%%%%%%%%%%%%%%%%%%%%%%%%%%%%%%%%%%%%%%%%%%%%%%%%%%%%%%%
%%%%%%%%%%%%%%%%%%%%%%%%%%%%%%%%%%%%%%%%%%%%%%%%%%%%%%%%%%%%
\end{document}